\documentclass[a4paper]{article}
\usepackage{fullpage}

\usepackage{booktabs}   %% For formal tables:
\usepackage{subcaption} %% For complex figures with subfigures/subcaptions
\usepackage{amsmath}
\usepackage{amsthm}
\usepackage{amsfonts}
\usepackage{amssymb}
\usepackage{cmll}

\usepackage{tikz}
\usetikzlibrary{matrix}
\usetikzlibrary{cd}
\usepackage{pgfplots}
\usepackage{stmaryrd}
\usepackage{ebproof}
\usepackage{url}
\usepackage{lineno}
\newcommand\CMLLPAR{
\usepackage{cmll}
\newcommand\IPar{\mathord{\parr}}
}

\CMLLPAR
\makeatletter
\newcommand*{\inlineeq}[2][]{%
  \begingroup
    \refstepcounter{equation}%
    \ifx\\#1\\%
    \else
      \label{#1}%
    \fi
    \relpenalty=10000 %
    \binoppenalty=10000 %
    \ensuremath{%
      #2%
    }%
    ~\@eqnnum
  \endgroup
}
\makeatother

\theoremstyle{definition}

\newtheorem{definition}{Definition}[section]
\newtheorem{remark}{Remark}[section]

\theoremstyle{theorem}
\newtheorem{lemma}{Lemma}[section]
\newtheorem{theorem}{Theorem}[section]
\newtheorem{proposition}{Proposition}[section]

\newcounter{examplectr}

\newenvironment{Example}{%
   \bigbreak\noindent%
   \refstepcounter{examplectr}%
   \textbf{$\blacktriangleright$\ Example \theexamplectr.\ }%
   }{\ \hfill$\blacktriangleleft$\par\bigbreak}
\numberwithin{examplectr}{section}

\newenvironment{Axicond}[1]
{\smallbreak\noindent{#1}\,}
{\smallbreak}

\newcommand\Proofcase{\smallbreak\noindent$\triangleright$\ }

\newcommand{\Endproof}{
  \ifmmode % if math mode, assume display: omit penalty etc.
  \else \leavevmode\unskip\penalty9999 \hbox{}\nobreak\hfill
  \fi
  \quad\hbox{$\Box$}
  \par\medskip}

\newcommand\Eqref[1]{(\ref{#1})}

\renewcommand{\phi}{\varphi}
\renewcommand\epsilon{\varepsilon}

\newcommand{\Implies}{\Rightarrow}

\newcommand\Equiv{\Leftrightarrow}
\newcommand{\St}{\mid}

\renewcommand{\Bot}{{\mathord{\perp}}}
\newcommand{\Top}{\top}

\newcommand\cB{\mathcal{B}}

\newcommand\cL{\mathcal{L}}

\newcommand\Fini{{\mathrm{fin}}}

\newcommand\Union{\bigcup}

\newcommand{\Linarrow}{\multimap}

\newcommand\Myleft{}
\newcommand\Myright{}

\newcommand\Web[1]{\Myleft|{#1}\Myright|}

\newcommand\Supp[1]{\operatorname{\mathsf{supp}}({#1})}

\newcommand\Emptymset{[\,]}
\newcommand\Mset[1]{[{#1}]}

\newcommand\Cl[1]{\mbox{\textrm{Cl}}({#1})}

\newcommand\CohName{\mathbf{Coh}}
\newcommand\NCohName{\mathbf{NCoh}}
\newcommand\COH{\CohName}
\newcommand\NCOH{\NCohName}

\newcommand\ITens{\otimes}
\newcommand\Tens[2]{{#1}\ITens{#2}}
\newcommand\Tensp[2]{({#1}\ITens{#2})}
\newcommand\IWith{\mathrel{\&}}
\newcommand\With[2]{{#1}\IWith{#2}}
\newcommand\Withp[2]{\left({#1}\IWith{#2}\right)}
\newcommand\IPlus{\oplus}
\newcommand\Plus[2]{{#1}\IPlus{#2}}
\newcommand\Orth[2][]{#2^{\Bot_{#1}}}

\newcommand\Bwith{\mathop{\&}}

\newcommand\Inj[1]{\overline\pi_{#1}}

\newcommand\One{1}

\newcommand\Card[1]{\#{#1}}

\newcommand\Locun[1]{1^J}

\newcommand\Isom\simeq

\newcommand\NUCS{\mathbf{nCoh}}

\newcommand\Comp{\mathrel\circ}

\newcommand\PSET{\mathbf{Set}_0}

\newcommand\Limpl[2]{{#1}\Linarrow{#2}}
\newcommand\Limplp[2]{\left({#1}\Linarrow{#2}\right)}

\newcommand\Nat{{\mathbb{N}}}

\newcommand\Bool{\mathbf{Bool}}
\newcommand\True{\mathbf t}
\newcommand\False{\mathbf f}

\newcommand\Diff[3]{\mathrm D_{#1}{#2}\cdot{#3}}

\newcommand\App[2]{({#1}){#2}}

\newcommand\Abst[3]{\lambda#1^{#2}\,{#3}}
\newcommand\Diffp[3]{\frac{\partial{#1}}{\partial{#2}}\cdot{#3}}
\newcommand\Diffpev[4]{\frac{\partial{#1}}{\partial{#2}}(#3)\cdot{#4}}

\newcommand\Derp[3]{\frac{\partial{#1}}{\partial{#2}}\cdot{#3}}
\newcommand\Derd[3]{\frac{d{#1}}{d{#2}}\cdot{#3}}
\newcommand\Derdn[4]{\frac{d^{#1}{#2}}{d{#3}^{#1}}\cdot{#4}}
\newcommand\Derdev[4]{\frac{d{#1}}{d{#2}}(#3)\cdot{#4}}

\newcommand\List[3]{#1_{#2},\dots,#1_{#3}}

\newcommand\Subst[3]{{#1}\left[{#2}/{#3}\right]}

\newcommand\Real{\mathbb{R}}
\newcommand\Realp{\mathbb{R}_{\geq 0}}

\newcommand\Intercc[2]{[#1,#2]}

\newcommand\Mfin[1]{\mathcal M_\Fini({#1})}

\newcommand\Ev{\operatorname{\mathsf{Ev}}}
\newcommand\Evlin{\operatorname{\mathsf{ev}}}
\newcommand\REL{\operatorname{\mathbf{Rel}}}

\newcommand\Norm[1]{\|{#1}\|}

\newcommand\Red{\beta_\Delta}
\newcommand\Redst[1]{\mathop{\mathsf{Red}}}

\newcommand\Symgrp[1]{\mathfrak S_{#1}}

\newcommand\Tuple[1]{\langle{#1}\rangle}
\newcommand\Cotuple[1]{\left[{#1}\right]}

\newcommand\Msetofsubst[1]{\bar F}

\newcommand\Inv[1]{{#1}^{-1}}
\newcommand\Invp[1]{({#1})^{-1}}

\newcommand\PCOH{\mathbf{Pcoh}}

\newcommand\Leftu{\lambda}
\newcommand\Rightu{\rho}
\newcommand\Assoc{\alpha}
\newcommand\Sym{\gamma}

\newcommand\Absval[1]{\left|{#1}\right|}

\newcommand\Msetsum[1]{\Sigma{#1}}

\newcommand\Retri\zeta
\newcommand\Retrp\rho

\newcommand\Impl[2]{{#1}\Rightarrow{#2}}
\newcommand\Implp[2]{({#1}\Rightarrow{#2})}

\newcommand\Tnat\iota

\newcommand\Num[1]{\underline{#1}}
\newcommand\Loop\Omega

\newcommand\Tseq[3]{{#1}\vdash{#2}:{#3}}

\newcommand\Timpl\Impl
\newcommand\Timplp\Implp
\newcommand\Simpl\Impl

\newcommand\Weak[1]{\operatorname{\mathsf{weak}}_{#1}}
\newcommand\Contr[1]{\operatorname{\mathsf{contr}}_{#1}}

\newcommand\Der[1]{\operatorname{\mathsf{der}}_{#1}}

\newcommand\Digg[1]{\operatorname{\mathsf{dig}}_{#1}}

\newcommand\Fun[1]{\widehat{#1}}

\newcommand\Id{\operatorname{\mathsf{Id}}}

\newcommand\Proj[1]{\mathsf{pr}_{#1}}

\newcommand\Excl[1]{\oc{#1}}
\newcommand\Excll[1]{\oc\oc{#1}}
\newcommand\Exclp[1]{\oc({#1})}

\newcommand\Relincl\eta
\newcommand\Relrestr\rho

\newcommand\Seely{\mathsf m}
\newcommand\Seelyz{\Seely^0}
\newcommand\Seelyt{\Seely^2}
\newcommand\Monoidal{\mu}
\newcommand\Monz{\Monoidal^0}
\newcommand\Mont{\Monoidal^2}

\newcommand\Compl{\,}
\newcommand\Curlin{\operatorname{\mathsf{cur}}}
\newcommand\Curlinp[1]{\Curlin(#1)}
\newcommand\Cur{\operatorname{\mathsf{Cur}}}

\newcommand\Kl[1]{{#1}_\oc}
\newcommand\Em[1]{{#1}^\oc}

\newcommand\Coalgw[1]{\mathsf w_{#1}}
\newcommand\Coalgc[1]{\mathsf c_{#1}}

\newcommand\Bnfeq{\mathrel{\mathord:\mathord=}}
\newcommand\Bnfor{\,\,\mathord|\,\,}

\newcommand\Eset[1]{\{#1\}}

\newcommand\Sfun{\mathbf S}
\newcommand\Scfun{\Sfun_{\Into}}
\newcommand\Sproj[1]{\pi_{#1}}
\newcommand\Sin[1]{\iota_{#1}}
\newcommand\Ssum{\sigma}
\newcommand\Sflip{\mathsf c}

\newcommand\Scflip{\Sym_{\Into,\Into}}
\newcommand\Stuple[1]{\Tuple{#1}_{\Sfun}}

\newcommand\Smont{\mathsf L}
\newcommand\Scmont{\widetilde{\Smont}}
\newcommand\Scmontn[1]{\widetilde{\Smont}^{(#1)}}
\newcommand\Sstr{\phi^1}
\newcommand\Sstrs{\phi^0}
\newcommand\Sstrc{\phi^{\multimap}}
\newcommand\Sdiff{\partial}

\newcommand\Sdfun{\mathsf D}
\newcommand\Coh[3]{#2\coh_{#1}#3}
\newcommand\Scoh[3]{#2\scoh_{#1}#3}
\newcommand\Incoh[3]{#2\incoh_{#1}#3}
\newcommand\Sincoh[3]{#2\sincoh_{#1}#3}
\newcommand\Neu[3]{#2\equiv_{#1}#3}
\newcommand\Matappa[2]{{#1}\cdot{#2}}

\newcommand\Saxcom{(\textbf{S-com})}
\newcommand\Saxzero{(\textbf{S-zero})}
\newcommand\Saxass{(\textbf{S-assoc})}
\newcommand\Saxwit{(\textbf{S-witness})}
\newcommand\Saxdist{(\textbf{S$\ITens$-dist})}
\newcommand\Saxfun{(\textbf{S$\ITens$-fun})}
\newcommand\Saxprod{(\textbf{S$\IWith$-pres})}

\newcommand\Csaxepi{(\textbf{CS-epi})}

\newcommand\Ccsaxepi{(\textbf{CCS-epi})}

\newcommand\Sone{\One}
\newcommand\Sonelem{\ast}

\newcommand\Win[1]{\Inj{#1}}
\newcommand\Wdiag{\Delta}
\newcommand\Into{\mathsf I}

\newcommand\Diffofst[1]{{#1}^{+}}

\newcommand\Stofdiff[1]{{#1}^{-}}

\newcommand\Sdiffst{\widetilde\partial}
\newcommand\Sdiffca{\delta}

\newcommand\Diffrac[2]{\frac{d #1}{d #2}}

\newcommand\Sfunadd{\tau}

\newcommand\Tinto{\widetilde\Sfun_\Into}

\newcommand\Adj{\mathrel{\dashv}}

\newcommand\Ftunit{\eta}
\newcommand\Ftcounit{\epsilon}

\newcommand\Daxchain{($\partial$\textbf{-chain})}
\newcommand\Daxlocal{($\partial$\textbf{-local})}
\newcommand\Daxclocal{(\textbf{C}$\partial$\textbf{-local})}
\newcommand\Daxclin{(\textbf{C}$\partial$\textbf{-lin})}
\newcommand\Daxlin{($\partial$\textbf{-lin})}
\newcommand\Daxwith{($\partial$\textbf{-}$\mathord{\IWith}$)}
\newcommand\Daxcwith{(\textbf{C}$\partial$\textbf{-}$\mathord{\IWith}$)}
\newcommand\Daxschwarz{($\partial$\textbf{-Schwarz})}
\newcommand\Daxcschwarz{(\textbf{C}$\partial$\textbf{-Schwarz})}
\newcommand\Daxcchain{(\textbf{C}$\partial$\textbf{-chain})}
\newcommand\Daxcalocal{($\partial$\textbf{ca-local})}

\newcommand\Daxcalin{($\partial$\textbf{ca-lin})}

\newcommand\Treesep{\quad\quad}
\newcommand\Textsep{\hspace{6em}}
\newcommand\Formsep{\hspace{8em}}

\newcommand\Sdfunit{\zeta}
\newcommand\Sdfmult{\theta}
\newcommand\Sdfstr{\psi^1}
\newcommand\Sdfstrs{\psi^0}

\newcommand\Kllin{\mathsf{Lin}_{\mathord\oc}}

\newcommand\Tdiffsymb{\mathsf D}
\newcommand\Tdiff[1]{\Tdiffsymb{#1}}
\newcommand\Tdiffm[2]{\Tdiffsymb^{#1}{#2}}
\newcommand\Tdnat[1]{\Tdiffsymb^{#1}\Tnat}

\newcommand\Lprojd[3]{\Sproj{#1}^{#2}(#3)}
\newcommand\Lprojs[1]{\Sproj{#1}}
\newcommand\Linj[2]{\Sin{#1}(#2)}
\newcommand\Linjd[3]{\Sin{#1}^{#2}(#3)}

\newcommand\Lsumd[2]{\Sdfmult^{#1}(#2)}

\newcommand\Lflipd[2]{\Sflip^{#1}(#2)}

\newcommand\Lsums{\Sdfmult}
\newcommand\Ldlet[3]{\partial(#1,#2,#3)}

\newcommand\Lzero{0}

\newcommand\Lplus[2]{#1+#2}
\newcommand\Ldiff[1]{\Tdiff{#1}}

\newcommand\Ldiffp[1]{\Tdiff({#1})}
\newcommand\Lfix[1]{\mathsf Y#1}

\renewcommand\Red{\to}

\newcommand\Tbundle[1]{\mathsf T#1}

\newcommand\Linred{\to_{\mathsf{lin}}}

\newcommand\Scfunnt[1]{\mathsf{nt}(#1)}

\newcommand\Coalgca[1]{\underline{#1}}
\newcommand\Coalgm[1]{h_{#1}}
\newcommand\Projc[1]{\Proj{#1}^\ITens}
\newcommand\Tuplec[1]{\Tuple{#1}^\ITens}

\newcommand\Comonca[1]{\underline{#1}}
\newcommand\Comonw[1]{\mathsf w_{#1}}
\newcommand\Comonc[1]{\mathsf c_{#1}}

\newcommand\Cm[1]{{#1}^{\ITens}}

\newcommand\Calgofcmon{\mathsf A}
\newcommand\Cmonofcalg{\mathsf M}

\newcommand\Cproj[1]{\Proj{#1}^\ITens}

\newcommand\Tot[2]{{#1}^{\mathord\otimes #2}}

\title{Coherent differentiation}
\author{Thomas Ehrhard\\
IRIF, CNRS and Université de Paris}

\begin{document}

\maketitle

\begin{abstract}
The categorical models of the differential lambda-calculus are additive categories because of the Leibniz rule which requires the summation of two expressions. This means that, as far as the differential lambda-calculus and differential linear logic are concerned, these models feature finite non-determinism and indeed these languages are essentially non-deterministic. We introduce a categorical framework for differentiation which does not require additivity and is compatible with deterministic models such as coherence spaces and probabilistic models such as probabilistic coherence spaces. Based on this semantics we sketch the syntax of a deterministic version of the differential lambda-calculus.

% https://popl22.hotcrp.com/
% ehrhard@irif.fr
% !POPL1904?

\end{abstract}

\section*{Introduction}
The differential $\lambda$-calculus has been introduced
in~\cite{EhrhardRegnier02}, starting from earlier investigations on
the semantics of Linear Logic (LL) in models based on various kinds of
topological vector spaces~\cite{Ehrhard00b,Ehrhard00c}. Later on we
proposed in~\cite{EhrhardRegnier04a,Ehrhard18} an extension of LL
featuring differential operations which appear as an additional
structure of the exponentials (the resource modalities of LL),
offering a perfect duality to the standard rules of dereliction,
weakening and contraction. The differential $\lambda$-calculus and
differential LL are about computing formal derivatives of programs and
from this point of view are deeply connected to the kind of formal
differentiation of programs used in Machine Learning for propagating
gradients (that is, differentials viewed as vectors of partial
derivatives) within formal neural networks. As shown by the
recent~\cite{BrunelMazzaPagani20,MazzaPagani21} formal transformations
of programs related to the differential $\lambda$-calculus can be used
for efficiently implementing gradient back-propagation in a purely
functional framework.  The differential $\lambda$-calculus and the
differential linear logic are also useful as the foundation for an
approach to finite approximations of programs based on the Taylor
expansion~\cite{EhrhardRegnier06a,BarbarossaManzonetto20} which
provides a precise analysis of the use of resources during the
execution of a functional program deeply related with implementations
of the $\lambda$-calculus in abstract machines such as the Krivine
Machine~\cite{EhrhardRegnier06b}.

One should insist on the fact that in the differential
$\lambda$-calculus derivatives are not taken wrt.~to a ground type of
real numbers as in~\cite{BrunelMazzaPagani20,MazzaPagani21} but can be
computed wrt.~elements of all types. For instance it makes sense to
compute the derivative of a function $M:{\Timplp\Tnat\Tnat}\to\Tnat$
wrt.~its argument which is a function from $\Tnat$, the type of
integers, to itself, thus suggesting the possibility of using this
formalism for optimization purposes in a model such as probabilistic
coherence spaces~\cite{DanosEhrhard08} (PCS) where a program of type
$\Tnat\to\Tnat$ is seen as an analytic function transforming
probability distributions on the integers. In~\cite{Ehrhard19} it is
also shown how such derivatives can be used to compute the expectation
of the number of steps in the execution of a program. A major obstacle
on the extension of programming languages with such derivatives is the
fact that probabilistic coherence spaces are not a model of the
differential $\lambda$-calculus in spite of the fact that the
morphisms, being analytic, are obviously differentiable. The main goal
of this paper being to circumvent this obstacle, let us first
understand it better.

These differential extensions of the $\lambda$-calculus and of LL
require the possibility of adding terms of the same type. For
instance, to define the operational semantics of the differential
$\lambda$-calculus, given a term $t$ such that $\Tseq{x:A}tB$ and a
term $u$ such that $\Tseq\Gamma uA$ one has to define a term
$\Diffp txu$ such that $\Tseq{\Gamma,x:A}{\Diffp txu}B$ which can be
understood as a \emph{linear substitution} of $u$ for $x$ in $t$ and
is actually a formal differentiation: $x$ has no reason to occur
linearly in $t$ so this operation involves the creation of linear
occurrences of $x$ in $t$ and this is done applying the rules of
ordinary differential calculus. The most important case is when $t$ is
an application $t=\App{t_1}{t_2}$ where
$\Tseq{\Gamma,x:A}{t_1}{\Impl CB}$ and $\Tseq{\Gamma,x:A}{t_2}{C}$. In
that case we set
\begin{align*}
  \Diffp{\App{t_1}{t_2}}{x}{u}
  =\App{\Diffp{t_1}xu}{t_2}
  +\App{\Diff{}{t_1}{(\Diffp{t_2}xu)}}{t_2}
\end{align*}
where we use \emph{differential application} which is a syntactic
construct of the language: given $\Tseq\Gamma s{\Impl CB}$ and
$\Tseq\Gamma vC$, we have $\Tseq\Gamma{\Diff{}sv}{\Impl CB}$.  This
crucial definition involves a sum corresponding to the fact that $x$
can appear free in $t_1$ and in $t_2$: this is the essence of the
``Leibniz rule'' $(fg)'=f'g+fg'$ which has nothing to do with
multiplication but everything with the fact that both $f$ and $g$ can
have non-zero derivatives wrt.~a common variable they share (logically
this sharing is implemented by a contraction rule).

For this reason the syntax of the differential $\lambda$-calculi and
linear logic features an addition operation on terms of the same type
and accordingly the categorical models of these formalisms are based
on additive categories. Operationally such sums correspond to a form
of finite non-determinism: for instance if the language has a ground
type of integers $\Tnat$ with constants $\Num n$ such that
$\Tseq\Gamma{\Num n}\Tnat$ for each $n\in\Nat$, we are allowed to
consider sums such as $\Num{42}+\Num{57}$ corresponding to the
non-deterministic superposition of the two integers (and not at all to
their sum $\Num{99}$ in the usual sense!). This can be considered as a
weakness of this approach since, even if one has nothing against
non-determinism \emph{per se} it is not satisfactory to be obliged to
enforce it for allowing differential operations which have nothing to
do with it \emph{a priori}.  So the fundamental question is:
\begin{quote}
  Does every logical approach to differentiation require non-determinism?
\end{quote}
We ground our negative answer to this question on the observation made
in~\cite{Ehrhard19} that, in the category of PCS, morphisms of the
associated cartesian closed category are analytic functions and
therefore admit all iterated derivatives (at least in the ``interior''
of the domain where they are defined).  Consider for instance in this
category an analytic $f:\One\to\One$ where $\One$ (the $\ITens$ unit
of LL) is the $[0,1]$ interval, meaning that
$f(x)=\sum_{n=0}^\infty a_nx^n$ with coefficient $a_n\in\Realp$ such
% TYPO
% that $\sum_{n=0}^\infty a_nx^n\leq 1$. The derivative
that $\sum_{n=0}^\infty a_n\leq 1$. The derivative
$f'(x)=\sum_{n=0}^\infty (n+1)a_{n+1}x^n$ has no reason to map $[0,1]$
to $[0,1]$ and can even be unbounded on $[0,1)$ and undefined at $x=1$
(and there are programs whose interpretation behaves in that
way). Though, if $(x,u)\in[0,1]^2$ satisfy $x+u\in[0,1]$ then
$f(x)+f'(x)u\leq f(x+u)\in[0,1]$. This is true actually of any
analytic morphism $f$ between two PCSs $X$ and $Y$: we can see the
differential of $f$ as mapping a summable pair $(x,u)$ of elements of
$X$ to the summable pair $(f(x),f'(x)\cdot u)$ of elements of
$Y$. Seeing the differential as such a pair of functions is central in
differential geometry as it allows, thanks to the chain rule, to turn
it into a \emph{functor} mapping a smooth map $f:X\to Y$ (where $X$
and $Y$ are now manifolds) to the function
$\Tbundle f:\Tbundle X\to\Tbundle Y$ which maps $(x,u)$ to
$(f(x),f'(x)\cdot u)$ where $\Tbundle X$ is the tangent bundle of $X$,
a manifold whose elements are the pairs $(x,u)$ of a point $x$ of $X$
and of a vector $u$ tangent to $X$ at $x$. The concept of
\emph{tangent category} has been introduced
in~\cite{Rosicky84,CockettCruttwell14} precisely to describe
categorically this construction and its properties. In spite of this
formal similarity our central concept of summability cannot be
compared with tangent categories in terms of generality, first because
when $(x,u)\in\Tbundle X$ it makes no sense to add $x$ and $u$ or to
consider $u$ alone (independently of $x$), and second because, given
$(x,u_0),(x,u_1)\in\Tbundle X$, the local sum
$(x,u_0+u_1)\in\Tbundle X$ is always defined in the tangent bundle,
whereas in our summability setting, the pair $(u_0,u_1)$ has no reason
to be summable.

\paragraph*{Content.}
We base our approach on a concept of summable pair that we axiomatize
as a general categorical notion in Section~\ref{sec:sum-cat}: a
\emph{summable category} is a category $\cL$ with
$0$-morphisms\footnote{That is, whose hom-sets are pointed sets.}
together with a functor $\Sfun:\cL\to\cL$ equipped with three natural
transformations from $\Sfun X$ to $X$: two projections and a sum
operation. The first projection also exists in the ``tangent bundle''
functor of a tangent category but the two other morphisms do not. Such
a summability structure induces a monad structure on $\Sfun$ (a
similar phenomenon occurs in tangent categories).
In Section~\ref{sec:sum-moncat} we consider the case where the
category is a cartesian SMC equipped with a resource comonad $\Excl\_$
in the sense of LL where we present differentiation as a distributive
law between the monad $\Sfun$ and the comonad $\Excl\_$. This allows
to extend $\Sfun$ to a strong monad $\Sdfun$ on the Kleisli category
$\Kl\cL$ which implements differentiation of non-linear maps.
In Section~\ref{sec:canonical-sum} we study the case where the functor
$\Sfun$ can be defined using a more basic structure of $\cL$ based on
the object $\With\Sone\Sone$ where $\IWith$ is the cartesian product
and $\Sone$ is the unit of $\ITens$: this is actually what happens in
the concrete situations we have in mind. Then the existence of the
summability structure becomes a \emph{property} of $\cL$ and not an
additional structure. We also study the differential structure in this
setting, showing that it boils down to a simple $\oc$-coalgebra
structure on $\With\Sone\Sone$. 

As a running example along the presentation of our categorical
constructions we use the category of coherence spaces, the first model
of LL historically~\cite{Girard87}. There are many reasons for this
choice. It is one of the most popular models of LL and of functional
languages, it is a typical example of a model of LL which is not an
additive category (in contrast with the relational model or the models
of profunctors), \emph{a priori} it does not exhibit the usual
features of a model of the differential calculus (no coefficients, no
vector spaces etc) and it strongly suggests that our coherent approach
to the differential $\lambda$-calculus might be applied to programming
languages which have nothing to do with probabilites, deep learning
or non-determinism.
In Section~\ref{sec:coh-diff-str} we describe the differential
structure of the coherence space model, showing that it provides an
example of a canonically summable differential category. We observe
that, in the \emph{uniform} setting of Girard's coherence space, our
differentiation does not satisfy the Taylor formula but that this
formula will hold if we use instead \emph{non-uniform} coherence
spaces of which we describe the differential structure.

In Section~\ref{sec:SMCC-summability} we consider the situation where
the underlying SMC is closed, that is, has internal hom objects. In
that case an additional condition on the summability structure is
required, expressing intuitively that the sum of two morphisms is
computed pointwise.

Last in Section~\ref{sec:syntax} we outline a syntax for a
differential $\lambda$-calculus corresponding to this semantics. This
concluding section should only be considered as an appetizer for a
more consistent paper on a differential and deterministic extension of
PCF which will be available soon.

\paragraph*{Related work.}
As already mentioned our approach has strong similarities with tangent
categories which have been a major source of inspiration, we explained
above the differences. There are also strong connections with
differential categories~\cite{BluteCockettLemaySeely20} with the main
difference again that differential categories are left-additive which
is generally not the case of $\Kl\cL$ in our case, we explained
why. There are also interesting similarities
with~\cite{CockettLemayLucyshyn20} (still in an additive setting): our
distributive law $\Sdiff_X$ might play a role similar to the one of
the distributive law introduced in the Section~5 of that paper. This
needs further investigations.

Recently~\cite{KerjeanPedrot20} have exhibited a striking connection
between Gödel's Dialectica interpretation and the differential
$\lambda$-calculus and differential linear logic, with applications to
gradient back-propagation in differential programming. One distinctive
feature of Pédrot's approach to Dialectica~\cite{Pedrot15} is to use a
``multiset parameterized type'' $\mathfrak M$ whose purpose is
apparently to provide some control on the summations allowed when
performing Pédrot's analogue of the Leibniz rule (under the
Dialectica/differential correspondence of~\cite{KerjeanPedrot20}) and
might therefore play a role similar to our summability functor
$\Sfun$. The precise technical connection is not clear at all but we
believe that this analogy will lead to a unified framework for
Dialectica interpretation and coherent differentiation of programs and
proofs involving denotational semantics, proof theory and differential
programming.

The differential $\lambda$-calculus that we obtain in
Section~\ref{sec:syntax} features strong similarities with the
calculus introduced in~\cite{BrunelMazzaPagani20,MazzaPagani21} for
dealing with gradient propagation in a functional setting. Both
calculi handle tuples of terms in the spirit of tangent categories
which allows to make the chain rule functorial thus allowing to reduce
differential terms without creating explicit summations.

% Several proofs which are not in the main text can be found in the Appendix.

\section{Preliminary notions and results}
This section provides some more or less standard technical material
useful to understand the paper. It can be skipped and used when
useful, in call-by-need manner.

\subsection{Finite multisets}
A finite multiset on a set $A$ is a function $m:A\to\Nat$ such that
the set $\Supp m=\Eset{a\in A\St m(a)\not=0}$ is finite, we use
$\Mfin A$ for the set of all finite multisets of elements of $A$. The
cardinality of $m$ is $\Card m=\sum_{a\in A}m(a)$. We use $\Emptymset$
for the empty multiset (so that $\Supp\Emptymset=\emptyset$ where
$\Supp m=\Eset{a\in A\St m(a)\not=0}$ is the \emph{support} of $m$)
and if $m_0,m_1\in\Mfin A$ then $m_0+m_1\in\Mfin A$ is defined by
$(m_0+m_1)(a)=m_0(a)+m_1(a)$. If $\List a1n\in A$ we use
$\Mset{\List a1n}$ for the $m\in\Mfin A$ such that $m(a)$ is the
number of $i\in\{1,\dots,n\}$ such that $a_i=a$. If
$m=\Mset{\List a1n}\in\Mfin A$ and $p=\Mset{\List b1p}\in\Mfin B$ then
$m\times p=\Mset{(a_i,b_j)\St i\in\Eset{1,\dots ,n}\text{ and
  }j\in\Eset{1,\dots,p}}\in\Mfin{A\times B}$. If
$M=\Mset{\List m0n}\in\Mfin{\Mfin A}$ we set
$\Msetsum M=\sum_{i=0}^nm_i\in\Mfin A$.

\subsection{The SMCC of pointed sets}\label{sec:pointed-sets}
Let $\PSET$ be the category of pointed sets. We use $0_X$ or simply
$0$ for the distinguished point of the object $X$. A morphism
$f\in\PSET(X,Y)$ is a function $f:X\to Y$ such that $f(0_X)=0_Y$.  The
terminal object is the singleton $\{0\}$. The cartesian product
$\With XY$ is the ordinary cartesian product, with
$0_{\With XY}=(0_X,0_Y)$. The tensor product $\Tens XY$ is defined as
\begin{align*}
  \Tens XY=\{(x,y)\in X\times Y\St x=0\Equiv y=0\}
\end{align*}
with $0_{\Tens XY}=(0_X,0_Y)$. The unit of the tensor product is the
object $\Sone=\{0,\Sonelem\}$ of $\PSET$. This category is enriched
over itself, the distinguished point of $\PSET(X,Y)$ being the
constantly $0_Y$ function. Actually, it is monoidal closed with
$\Limpl XY=\PSET(X,Y)$ and $0_{\Limpl XY}$ defined by
$0_{\Limpl XY}(x)=0_Y$ for all $x\in X$. A mono in $\PSET$ is a
morphism of $\PSET$ which is injective as a function.

Unless explicitly stipulated, all the categories $\cL$ we consider in
this paper are enriched over pointed sets, so this assumption will not
be mentioned any more. In the case of symmetric monoidal categories,
this also means that the tensor product of morphisms is ``bilinear''
wrt.~the pointed structure, that is: if $f\in\cL(X_0,Y_0)$ then
$\Tens f0=0\in\cL(\Tens{X_0}{X_1},\Tens{Y_0}{Y_1})$ and by symmetry we
have $\Tens 0f=0$.

\subsection{Monoidal and resource categories} %
\label{sec:resource-cat}
A  symmetric monoidal category (SMC) is a category $\cL$
equipped with a bifunctor $\cL^2\to\cL$ denoted as $\ITens$,
a monoidal unit $\Sone$ which is an object of $\cL$ and  %
$\Leftu_X\in\cL(\Tens\Sone X,X)$, %
$\Rightu_X\in\cL(\Tens X\Sone,X)$, %
$\Assoc_{X_0,X_1,X_2}
\in\cL(\Tens{\Tensp{X_0}{X_1}}{X_2},\Tens{X_0}{\Tensp{X_1}{X_2}})$ %
and %
$\Sym_{X_0,X_1}\in\cL(\Tens{X_0}{X_1},\Tens{X_1}{X_0})$ as associated
isomorphisms satisfying the usual McLane coherence commutations.
Given objects $\List X0{n-1}$ and $i<j$ in $\Eset{0,\dots,n-1}$, we
use %
$\Sym_{i,j}$ for the canonical swapping iso in %
$\cL(X_0\ITens\cdots\ITens X_{n-1},X_0\ITens\cdots\ITens X_{i-1}\ITens
X_j\ITens X_{i+1}\cdots \ITens X_{j-1}\ITens X_i\ITens
X_{j+1}\ITens\cdots\ITens X_{n-1})$.

\subsubsection{Commutative comonoids}

\begin{definition}
  In a SMC $\cL$ (with the usual notations), a commutative comonoid is
  a tuple %
  $C=(\Comonca C,\Comonw C,\Comonc C)$ where %
  $\Comonca C\in\cL$, %
  $\Comonw C\in\cL(\Comonca C,\Sone)$ and %
  $\Comonc C\in\cL(\Comonca C,\Tens{\Comonca C}{\Comonca C})$ %
  are such that the following diagrams commute.
  \[
    \begin{tikzcd}
      \Comonca C
      \ar[r,"\Comonc C"]
      \ar[dr,swap,"\Invp{\Leftu_{\Comonca C}}"]
      &
      \Tens{\Comonca C}{\Comonca C}
      \ar[d,"\Tens{\Comonw C}{\Comonca C}"]
      \\
      &
      \Tens\Sone{\Comonca C}
    \end{tikzcd}
    \Treesep
    \begin{tikzcd}
      \Comonca C
      \ar[r,"\Comonc C"]
      \ar[dr,swap,"\Comonc C"]
      &
      \Tens{\Comonca C}{\Comonca C}
      \ar[d,"\Sym_{\Comonca C,\Comonca C}"]
      \\
      &
      \Tens{\Comonca C}{\Comonca C}
    \end{tikzcd}
    \Treesep
    \begin{tikzcd}
      \Comonca C
      \ar[rr,"\Comonc C"]
      \ar[d,swap,"\Comonc C"]
      &[1em]&[1em]
      \Tens{\Comonca C}{\Comonca C}
      \ar[d,"\Tens{\Comonca C}{\Comonc C}"]
      \\
      \Tens{\Comonca C}{\Comonca C}
      \ar[r,"\Tens{\Comonc C}{\Comonca C}"]
      &
      \Tens{\Tensp{\Comonca C}{\Comonca C}}{\Comonca C}
      \ar[r,"\Assoc_{\Comonca C,\Comonca C,\Comonca C}"]
      &
      \Tens{\Comonca C}{\Tensp{\Comonca C}{\Comonca C}}
    \end{tikzcd}
  \]
  The category $\Cm\cL$ of commutative comonoids has these tuples as
  objects, and an element of $\Cm\cL(C,D)$ is an
  $f\in\cL(\Comonca C,\Comonca D)$ such that the two following
  diagrams commute
  \[
    \begin{tikzcd}
      \Comonca C
      \ar[r,"f"]
      \ar[rd,swap,"\Comonw C"]
      &
      \Comonca D
      \ar[d,"\Comonw D"]\\
      &
      \Sone
    \end{tikzcd}
    \Treesep
    \begin{tikzcd}
      \Comonca C
      \ar[d,swap,"\Comonc C"]
      \ar[r,"f"]
      &
      \Comonca D
      \ar[d,"\Comonc D"]\\
      \Tens{\Comonca C}{\Comonca C}
      \ar[r,"\Tens ff"]
      &
      \Tens{\Comonca D}{\Comonca D}
    \end{tikzcd}
  \]
\end{definition}

\begin{theorem}\label{th:comon-cat-cart}
  For any SMC $\cL$ the category $\Cm\cL$ is cartesian.
  The terminal object is %
  $(\Sone,\Id_\Sone,\Invp{\Leftu_\Sone})$ (remember that %
  $\Leftu_\Sone=\Rightu_\Sone$) simply denoted as $\Sone$ %
  and for any object $C$ the unique morphism $C\to\Sone$ is %
  $\Comonw C$.
  The cartesian product of $C_0,C_1\in\Cm\cL$ is %
  the object $\Tens{C_0}{C_1}$ of $\Cm\cL$ such that %
  $\Comonca{\Tens{C_0}{C_1}}=\Tens{\Comonca{C_0}}{\Comonca{C_1}}$ and
  the structure maps are defined as
  \[
    \begin{tikzcd}
      \Tens{\Comonca{C_0}}{\Comonca{C_1}}
      \ar[r,"\Tens{\Comonw{C_0}}{\Comonw{C_1}}"]
      &[2em]
      \Tens\Sone\Sone
      \ar[r,"\Leftu_\Sone"]
      &
      \Sone
      \\[-1em]
      \Tens{\Comonca{C_0}}{\Comonca{C_1}}
      \ar[r,"\Tens{\Comonc{C_0}}{\Comonc{C_1}}"]
      &
      \Tens{\Tens{\Comonca{C_0}}{\Comonca{C_0}}}
      {\Tens{\Comonca{C_1}}{\Comonca{C_1}}}
      \ar[r,"\Sym_{2,3}"]
      &
      \Tens{\Tens{\Comonca{C_0}}{\Comonca{C_1}}}
      {\Tens{\Comonca{C_0}}{\Comonca{C_1}}}      
    \end{tikzcd}
  \]
  The projections $\Cproj i\in\Cm\cL(\Tens{C_0}{C_1},C_i)$ are given by
  \[
    \begin{tikzcd}
      \Tens{\Comonca{C_0}}{\Comonca{C_1}}
      \ar[r,"\Tens{\Comonw{C_0}}{\Comonca{C_1}}"]
      &[2em]
      \Tens{\Sone}{\Comonca{C_1}}
      \ar[r,"\Leftu_{\Comonca{C_1}}"]
      &
      \Comonca{C_1}\\[-1em]
      \Tens{\Comonca{C_0}}{\Comonca{C_1}}
      \ar[r,"\Tens{\Comonca{C_0}}{\Comonw{C_1}}"]
      &[2em]
      \Tens{\Comonca{C_0}}{\Sone}
      \ar[r,"\Rightu_{\Comonca{C_1}}"]
      &
      \Comonca{C_0}
    \end{tikzcd}\,.
  \]
\end{theorem}
The proof is straightforward.
In a commutative monoid $M$, multiplication is a monoid morphism
$M\times M\to M$. The following is in the vein of this simple observation.
\begin{lemma}\label{ref:weak-contr-comon-morph}
  If $C\in\Cm\cL$ then %
  $\Comonw C\in\Cm\cL(C,\Sone)$ and
  $\Comonc C\in\Cm\cL(C,\Tens CC)$.
\end{lemma}
\begin{proof}
  The second statement amounts to the following commutation
  \[
    \begin{tikzcd}
      \Comonca C
      \ar[rr,"\Comonc C"]
      \ar[d,swap,"\Comonc C"]
      &[0.4em]&[-0.8em]
      \Tens{\Comonca C}{\Comonca C}
      \ar[d,"\Tens{\Comonc C}{\Comonc C}"]
      \\
      \Tens{\Comonca C}{\Comonca C}
      \ar[r,"\Tens{\Comonc C}{\Comonc C}"]
      &
      \Tens{\Tens{\Comonca C}{\Comonca C}}{\Tens{\Comonca C}{\Comonca C}}
      \ar[r,"\Sym_{2,3}"]
      &
      \Tens{\Tens{\Comonca C}{\Comonca C}}{\Tens{\Comonca C}{\Comonca C}}
    \end{tikzcd}
  \]
  which results from the commutativity of $C$. The first statement is
  similarly trivial.
\end{proof}

\subsubsection{Resource categories}
The notion of resource category is more general than that of a Seely
category in the sense of~\cite{Mellies09}. We keep only the part of
the structure and axioms that we need to define our notion of
differential structure and keep our setting as general as possible.

An object $X$ of an SMC $\cL$ is exponentiable if the functor %
$\Tens\_ X$ has a right adjoint, denoted as %
$\Limpl X\_$. In that case, we use %
$\Evlin\in\cL(\Tens{\Limplp XY}{X},Y)$ for the counit of the
adjunction and, given %
$f\in\cL(\Tens ZX,Y)$ we use %
$\Curlin f$ for the associated morphism %
$\Curlin f\in\cL(Z,\Limpl XY)$.

We say that the SMC $\cL$ is closed (is an SMCC) if any object of
$\cL$ is exponentiable.

A category $\cL$ is a \emph{resource category} if
\begin{itemize}
\item $\cL$ is an SMC;
\item $\cL$ is cartesian with terminal object $\Top$ %
  (so that $0$ is the unique element of $\cL(X,\Top)$) %
  and cartesian product of $X_0$, $X_1$ denoted
  $(\With{X_0}{X_1},\Proj 0,\Proj 1)$ and pairing of morphisms
  $(f_i\in\cL(Y,X_i))_{i=0,1}$ denoted
  $\Tuple{f_0,f_1}\in\cL(Y,\With{X_0}{X_1})$;
\item and $\cL$ is equipped with a \emph{resource comonad}, that is a
  tuple $(\Excl\_,\Der{},\Digg{},\Seelyz,\Seelyt)$ where $\Excl\_$ is
  a functor $\cL\to\cL$ which is a comonad with counit $\Der{}$ and
  comultiplication $\Digg{}$, and $\Seelyz\in\cL(\Sone,\Excl\Top)$ and
  $\Seelyt\in\cL(\Tens{\Excl X}{\Excl Y},\Excl{(\With XY)})$ are the
  Seely isomorphisms subject to conditions that we do not recall here,
  see for instance~\cite{Mellies09} apart for the following which
  explains how $\Digg{}$ interacts with $\Seelyt$.
  \begin{equation}\label{eq:seely-digg-comm}
    \begin{tikzcd}
      \Tens{\Excl{X_0}}{\Excl{X_0}}
      \ar[rr,"\Tens{\Digg{X_0}}{\Digg{X_1}}"]
      \ar[d,swap,"\Seelyt_{X_0,X_1}"]
      &[1.8em]
      &[2em]
      \Tens{\Excll{X_0}}{\Excll{X_1}}
      \ar[d,"\Seelyt_{\Excl{X_0},\Excl{X_1}}"]
      \\
      \Excl{\Withp{X_0}{X_1}}
      \ar[r,"\Digg{\With{X_0}{X_1}}"]
      &
      \Excll{\Withp{X_0}{X_1}}
      \ar[r,"\Excl{\Tuple{\Excl{\Proj0},\Excl{\Proj1}}}"]
      &
      \Excl{\Withp{\Excl{X_0}}{\Excl{X_1}}}
    \end{tikzcd}
  \end{equation}
\end{itemize}

Then $\Excl\_$ inherits a \emph{lax} symmetric monoidality
$\Monz,\Mont$ on $\cL$ (considered as an SMC). This means that one can
define %
$\Monz\in\cL(1,\Excl\Sone)$ and %
$\Mont_{X_0,X_1}\in\cL(\Tens{\Excl{X_0}}{\Excl{X_1}},\Excl{\Tensp{X_0}{X_1}})$
satisfying suitable coherence commutations. Explicitly these morphisms
are given by
\[
  \begin{tikzcd}
    \Sone\ar[r,"\Seelyz"]
    &[-1em]
    \Excl\Top\ar[r,"\Digg\Top"]
    &
    \Excll\Top\ar[r,"\Excl{\Invp{\Seelyz}}"]
    &
    \Excl\Sone
  \end{tikzcd}
\]
\[
  \begin{tikzcd}
    \Tens{\Excl{X_0}}{\Excl{X_1}}\ar[r,"\Seelyt_{X_0,X_1}"]
    &[1em]
    \Excl{\Withp{X_0}{X_1}}\ar[r,"\Digg{\With{X_0}{X_1}}"]
    &[1.6em]
    \Excll{\Withp{X_0}{X_1}}\ar[r,"\Excl{\Invp{\Seelyt_{X_0,X_1}}}"]
    &[2em]
    \Excl{\Tensp{\Excl{X_0}}{\Excl{X_1}}}
    \ar[r,"\Excl{\Tensp{\Der{X_0}}{\Der{X_1}}}"]
    &[3em]
    \Excl{\Tensp{{X_0}}{{X_1}}}
  \end{tikzcd}
\]

\begin{lemma}\label{lemma:zero-mont}
  The following diagram commutes
  % \[
  %   \begin{tikzcd}
  %     \Tens{\Excl{X_0}}{\Excl{X_1}}\ar[rr,"\Mont_{X_0,X_1}"]
  %     \ar[d,swap,"\Tens{\Excl 0}{\Excl 0}"]
  %     &&
  %     \Excl{\Tensp{X_0}{X_1}}\ar[d,"\Excl 0"]\\
  %     \Tens{\Excl\Top}{\Excl\Top}\ar[r,"\Mont_{\Top,\Top}"]
  %     &
  %     \Exclp{\With\Top\Top}\ar[r,"\Excl 0"]
  %     &
  %     \Excl\Top
  %   \end{tikzcd}
  % \]
  \[
    \begin{tikzcd}
      \Tens{\Excl{X_0}}{\Excl{X_1}}\ar[r,"\Mont_{X_0,X_1}"]
      \ar[d,swap,"\Seelyt_{X_0,X_1}"]
      &
      \Excl{\Tensp{X_0}{X_1}}\ar[d,"\Excl 0"]\\
      \Excl{\Withp{X_0}{X_1}}\ar[r,"\Excl 0"]
      &
      \Excl\Top
    \end{tikzcd}
  \]
\end{lemma}
\begin{proof}
  This results from the definition of $\Mont$ and from the following
  commutation
  \[
    \begin{tikzcd}
      \Excl X\ar[r,"\Digg X"]\ar[rd,swap,"\Excl 0"]
      &
      \Excll X\ar[d,"\Excl 0"]\\
      &
      \Excl\Top
    \end{tikzcd}
  \]
  which results from the observation that %
  $\Excl0\in\cL(\Excll X,\Excl 0)$ can be written %
  $\Excl 0=\Exclp{0\Compl\Der X}$.
\end{proof}

For any $X\in\cL$ it is possible to define a contraction morphism
$\Contr X\in\cL(\Excl X,\Tens{\Excl X}{\Excl X})$ and a weakening
morphism $\Weak X\in\cL(\Excl X,\Sone)$ turning $\Excl X$ into a
commutative comonoid. These morphisms are defined as follows:
\[
  \begin{tikzcd}
    \Excl X\ar[r,"\Excl 0"]
    &\Excl\Top\ar[r,"\Inv{(\Seelyz)}"]
    &\Sone
  \end{tikzcd}
  \quad
  \begin{tikzcd}
    \Excl X\ar[r,"\Excl{\Tuple{\Id,\Id}}"]
    &\Excl{(\With XX)}\ar[r,"\Inv{(\Seelyt)}"]
    &\Tens{\Excl X}{\Excl X}
  \end{tikzcd}\,.
\]
% Then the triple $(\Excl X,\Weak X,\Contr X)$ is a commutative
% comonoid.

\begin{lemma}\label{lemma:seelyt-mont-commut}
  The two following diagrams commute in any resource category $\cL$.
  \[
    \begin{tikzcd}
      \Tens{\Sone}{\Excl Y}\ar[r,"\Tens\Sone{\Weak Y}"]
      \ar[d,left,"\Tens{\Seelyz}{\Excl Y}"]
      &
      \Tens\Sone\Sone\ar[r,"\Leftu_\Sone"]
      &[-1em]
      \Sone\ar[d,"\Seelyz"]
      \\
      \Tens{\Excl\Top}{\Excl Y}
      \ar[r,"\Mont_{\Top,Y}"]
      &
      \Excl{\Tensp\Top Y}\ar[r,"\Excl0"]
      &
      \Excl\Top
    \end{tikzcd}
  \]
  \[
    \begin{tikzcd}
      \Excl{X_0}\ITens\Excl{X_1}\ITens\Excl Y
      \ar[r,"\Tens\Id{\Contr Y}"]
      \ar[d,swap,"\Tens{\Seelyt_{X_0,X_1}}{\Excl Y}"]
      &[1em]
      \Excl{X_0}\ITens\Excl{X_1}\ITens\Excl Y\ITens\Excl Y
      \ar[r,"\Sym_{2,3}"]
      &
      \Excl{X_0}\ITens\Excl Y\ITens\Excl{X_1}\ITens\Excl Y
      \ar[d,"\Tens{\Mont_{X_0,Y}}{\Mont_{X_1,Y}}"]
      \\
      \Excl{\Withp{X_0}{X_1}}\ITens\Excl Y
      \ar[d,swap,"\Mont_{\With{X_0}{X_1},Y}"]
      &
      &
      \Excl{\Tensp{X_0}{Y}}\ITens\Excl{\Tensp{X_1}{Y}}
      \ar[d,"\Seelyt_{\Tensp{X_0}{Y},\Tensp{X_1}{Y}}"]
      \\
      \Excl{\Tensp{\Withp{X_0}{X_1}}{Y}}
      \ar[rr,"\Excl{\Tuple{\Tens{\Proj0}{Y},\Tens{\Proj1}{Y}}}"]
      &&
      \Excl{\Withp{\Tensp{X_0}{Y}}{\Tensp{X_1}{Y}}}
    \end{tikzcd}
  \]
\end{lemma}
\begin{proof}
  For the first diagram we have
  \begin{align*}
    \Excl0
    \Compl\Mont_{\Top,Y}
    \Compl\Tensp{\Seelyz}{\Excl Y}
    &=\Excl0
      \Compl\Seelyt_{\Top,Y}
      \Compl\Tensp{\Seelyz}{\Excl Y}
      \text{\quad by Lemma~\ref{lemma:zero-mont}}\\
    &=\Excl0
      \Compl\Excl{\Tuple{\Top,Y}}
      \Compl\Leftu_{\Excl Y}
      \text{\quad by the monoidality equations of }\Seelyz,\Seelyt\\
    &=\Excl0
      \Compl\Leftu_{\Excl Y}
  \end{align*}
  and
  \begin{align*}
    \Seelyz
    \Compl\Leftu_\Sone
    \Compl\Tensp{\Sone}{\Weak Y}
    &=\Seelyz
      \Compl\Weak Y
      \Compl\Leftu_{\Excl Y}
      \text{\quad by naturality of }\Leftu\\
    &=\Excl0
      \Compl\Leftu_{\Excl Y}
      \text{\quad by definition of }\Weak{Y}\,.
  \end{align*}
  For the second diagram, we compute
  \begin{align*}
    f_1
    &=\Excl{\Tuple{\Tens{\Proj0}{Y},\Tens{\Proj1}{Y}}}
      \Compl\Mont_{\With{X_0}{X_1},Y}\\
    &=\Excl{\Tuple{\Tens{\Proj0}{Y},\Tens{\Proj1}{Y}}}
      \Compl\Excl{\Tensp{\Der{\With{X_0}{X_1}}}{\Der Y}}
      \Compl\Excl{\Invp{\Seelyt_{\With{X_0}{X_1},Y}}}
      \Compl\Digg{X_0\IWith X_1\IWith Y}
      \Compl\Seelyt_{\With{X_0}{X_1},Y}\\
    &\hspace{10em}\text{by definition of }\Mont\\
    &=\Exclp{\With{\Tensp{\Der{X_0}}{\Der Y}}{\Tensp{\Der{X_1}}{\Der Y}}}
      \Compl\Excl{\Tuple{\Tens{\Excl{\Proj0}}{\Excl Y},
      \Tens{\Excl{\Proj1}}{\Excl Y}}}
      \Compl\Excl{\Invp{\Seelyt_{\With{X_0}{X_1},Y}}}
      \Compl\Digg{X_0\IWith X_1\IWith Y}
      \Compl\Seelyt_{\With{X_0}{X_1},Y}\\
    &\hspace{10em}\text{by naturality of }\Der{}\\
    &=\Exclp{\With{\Tensp{\Der{X_0}}{\Der Y}}{\Tensp{\Der{X_1}}{\Der Y}}}
      \Compl f_2
  \end{align*}
  where
  \begin{align*}
    f_2
    &=\Excl{\Tuple{\Tens{\Excl{\Proj0}}{\Excl Y},
      \Tens{\Excl{\Proj1}}{\Excl Y}}}
      \Compl\Excl{\Invp{\Seelyt_{\With{X_0}{X_1},Y}}}
      \Compl\Digg{X_0\IWith X_1\IWith Y}
      \Compl\Seelyt_{\With{X_0}{X_1},Y}\\
    &=\Exclp{\With{\Invp{\Seelyt_{X_0,Y}}}{\Invp{\Seelyt_{X_1,Y}}}}
      \Compl\Excl{\Tuple{\Excl{\Proj0},\Excl{\Proj1}}}
      \Compl\Excll q
      \Compl\Digg{X_0\IWith X_1\IWith Y}
      \Compl\Seelyt_{\With{X_0}{X_1},Y}
  \end{align*}
  by naturality of $\Seelyt$. %
  In that expression %
  $\Proj i\in\cL(X_0\IWith Y\IWith X_1\IWith Y,\With{X_i}Y)$ and %
  $q=\Tuple{\Proj0,\Proj2,\Proj1,\Proj2} \in\cL(X_0\IWith X_1\IWith
  Y,X_0\IWith Y\IWith X_1\IWith Y)$. We have used the commutation of
  the following diagram
  \[
    \begin{tikzcd}
      \Exclp{X_0\IWith X_1\IWith Y}
      \ar[r,"\Excl q"]
      \ar[d,swap,"\Invp{\Seelyt_{\With{X_0}{X_1},Y}}"]
      &[-1em]
      \Exclp{X_0\IWith Y\IWith X_1\IWith Y}
      \ar[r,"\Tuple{\Excl{\Proj0},\Excl{\Proj1}}"]
      &[1em]
      \With{\Excl{\Withp{X_0}{Y}}}{\Excl{\Withp{X_1}{Y}}}
      \ar[d,"\With{\Invp{\Seelyt_{X_0,Y}}}{\Invp{\Seelyt_{X_1,Y}}}"]
      \\
      \Tens{\Excl{\Withp{X_0}{X_1}}}{\Excl Y}
      \ar[rr,"\Tuple{\Tens{\Excl{\Proj0}}{\Excl Y},
        \Tens{\Excl{\Proj1}}{\Excl Y}}"]
      &&
      \With{\Tensp{\Excl{X_0}}{\Excl Y}}{\Tensp{\Excl{X_1}}{\Excl Y}}
    \end{tikzcd}
  \]
  which is easily proved by post-composing the two equated morphisms
  with
  $\Proj i\in\cL(\With{\Tensp{\Excl{X_0}}{\Excl
      Y}}{\Tensp{\Excl{X_1}}{\Excl Y}},\Tensp{\Excl{X_i}}{\Excl Y})$
  for $i=0,1$.

  Observe that
  \begin{align*}
    &\Exclp{\With{\Tensp{\Der{X_0}}{\Der Y}}{\Tensp{\Der{X_1}}{\Der Y}}}\\
    &\hspace{8em}=\Seelyt_{\Tens{X_0}{Y},\Tens{X_1}{Y}}
      \Compl{\Tensp{\Exclp{\Tens{\Der{X_0}}{\Der Y}}}
      {\Exclp{\Tens{\Der{X_1}}{\Der Y}}}}
      \Compl
      \Invp{\Seelyt_{\Tens{\Excl{X_0}}{\Excl Y},\Tens{\Excl{X_1}}{\Excl Y}}}
  \end{align*}
  by naturality of $\Seelyt$. The following diagram commutes
  \[
    \begin{tikzcd}
      \Exclp{\With{\Exclp{\With{X_0}{Y}}}{\Exclp{\With{X_1}{Y}}}}
      \ar[r,"\Excl{\Withp{\Invp{\Seelyt_{X_0,Y}}}{\Invp{\Seelyt_{X_1,Y}}}}"]
      \ar[d,swap,"\Invp{\Seelyt_{\Excl{\Withp{X_0}{Y}},\Excl{\Withp{X_1}{Y}}}}"]
      &[8em]
      \Exclp{\With{\Tensp{\Excl{X_0}}{\Excl Y}}{\Tensp{\Excl{X_1}}{\Excl Y}}}
      \ar[d,"\Invp{\Seelyt_{\Tens{\Excl{X_0}}{\Excl Y},\Tens{\Excl{X_1}}{\Excl Y}}}"]
      \\
      \Tens{\Excll{\Withp{X_0}{Y}}}{\Excll{\Withp{X_1}{Y}}}
      \ar[r,"\Tensp{\Excl{\Invp{\Seelyt_{X_0,Y}}}}
      {\Excl{\Invp{\Seelyt_{X_1,Y}}}}"]
      &
      \Tens{\Excl{\Tensp{\Excl{X_0}}{\Excl Y}}}{\Excl{\Tensp{\Excl{X_1}}{\Excl Y}}}
    \end{tikzcd}
  \]
  % \begin{align*}
  %   \Invp{\Seelyt_{\Tens{\Excl{X_0}}{\Excl Y},\Tens{\Excl{X_1}}{\Excl Y}}}
  %   \Compl\Excl{\Withp{\Invp{\Seelyt_{X_0,Y}}}{\Invp{\Seelyt_{X_1,Y}}}}
  %   =\Tensp{\Invp{\Seelyt_{X_0,Y}}}{\Invp{\Seelyt_{X_1,Y}}}
  %   \Compl\Invp{\Seelyt_{\Excl{\Withp{X_0}{Y}},\Excl{\Withp{X_1}{Y}}}}
  % \end{align*}
  and hence
  \begin{align*}
    f_1
    &=\Seelyt_{\Tens{X_0}{Y},\Tens{X_1}{Y}}
      \Compl{\Tensp{\Exclp{\Tens{\Der{X_0}}{\Der Y}}}
      {\Exclp{\Tens{\Der{X_1}}{\Der Y}}}}
      \Compl
      \Invp{\Seelyt_{\Tens{\Excl{X_0}}{\Excl Y},\Tens{\Excl{X_1}}{\Excl Y}}}
      \Compl\Exclp{\With{\Invp{\Seelyt_{X_0,Y}}}{\Invp{\Seelyt_{X_1,Y}}}}\\
    &\Formsep
      \Compl\Excl{\Tuple{\Excl{\Proj0},\Excl{\Proj1}}}
      \Compl\Excll q
      \Compl\Digg{X_0\IWith X_1\IWith Y}
      \Compl\Seelyt_{\With{X_0}{X_1},Y}      
      \\
    &=\Seelyt_{\Tens{X_0}{Y},\Tens{X_1}{Y}}
      \Compl{\Tensp{\Exclp{\Tens{\Der{X_0}}{\Der Y}}}
      {\Exclp{\Tens{\Der{X_1}}{\Der Y}}}}
      \Compl\Tensp{\Excl{\Invp{\Seelyt_{X_0,Y}}}}{\Excl{\Invp{\Seelyt_{X_1,Y}}}}
      \Compl\Invp{\Seelyt_{\Excl{\Withp{X_0}{Y}},\Excl{\Withp{X_1}{Y}}}}\\
    &\Formsep
      \Compl\Excl{\Tuple{\Excl{\Proj0},\Excl{\Proj1}}}
      \Compl f_3
  \end{align*}
  where, by naturality of $\Digg{}$,
  \begin{align*}
    f_3
    &=\Excll q
      \Compl\Digg{X_0\IWith X_1\IWith Y}
      \Compl\Seelyt_{\With{X_0}{X_1},Y}\\
    &=\Digg{X_0\IWith Y\IWith X_1\IWith Y}
      \Compl\Excl q
      \Compl\Seelyt_{\With{X_0}{X_1},Y}
      \in\cL(\Tens{\Excl{\Withp{X_0}{X_1}}}{\Excl Y},
      \Excll{(X_0\IWith Y\IWith X_1\IWith Y)})
  \end{align*}
  and hence, by the diagram~\Eqref{eq:seely-digg-comm}
  \begin{align*}
    \Excl{\Tuple{\Excl{\Proj0},\Excl{\Proj1}}}\Compl f_3
    &=\Excl{\Tuple{\Excl{\Proj0},\Excl{\Proj1}}}
      \Compl\Digg{X_0\IWith Y\IWith X_1\IWith Y}
      \Compl\Excl q
      \Compl\Seelyt_{\With{X_0}{X_1},Y}\\
    &={\Seelyt_{\Exclp{\With{X_0}{Y}},\Exclp{\With{X_1}{Y}}}}
      \Compl\Tensp{\Digg{\With{X_0}Y}}{\Digg{\With{X_1}Y}}
      \Compl\Invp{\Seelyt_{\With{X_0}Y,\With{X_1}Y}}
      \Compl\Excl q
      \Compl\Seelyt_{\With{X_0}{X_1},Y}\,.
  \end{align*}
  It follows that
  \begin{align*}
    f_1
    &=\Seelyt_{\Tens{X_0}{Y},\Tens{X_1}{Y}}
      \Compl{\Tensp{\Exclp{\Tens{\Der{X_0}}{\Der Y}}}
      {\Exclp{\Tens{\Der{X_1}}{\Der Y}}}}
      \Compl\Tensp{\Excl{\Invp{\Seelyt_{X_0,Y}}}}
      {\Excl{\Invp{\Seelyt_{X_1,Y}}}}\\
    &\Textsep
      \Tensp{\Digg{\With{X_0}Y}}{\Digg{\With{X_1}Y}}
      \Compl\Invp{\Seelyt_{\With{X_0}Y,\With{X_1}Y}}
      \Compl\Excl q
      \Compl\Seelyt_{\With{X_0}{X_1},Y}\\
    &=\Seelyt_{\Tens{X_0}{Y},\Tens{X_1}{Y}}
      \Compl{\Tensp{\Exclp{\Tens{\Der{X_0}}{\Der Y}}}
      {\Exclp{\Tens{\Der{X_1}}{\Der Y}}}}
      \Compl\Tensp{\Excl{\Invp{\Seelyt_{X_0,Y}}}}
      {\Excl{\Invp{\Seelyt_{X_1,Y}}}}\\
    &\Textsep
      \Tensp{\Digg{\With{X_0}Y}}{\Digg{\With{X_1}Y}}
      \Compl\Tensp{\Seelyt_{X_0,Y}}{\Seelyt_{X_1,Y}}\\
    &\Textsep
      \Tensp{\Invp{\Seelyt_{X_0,Y}}}{\Invp{\Seelyt_{X_1,Y}}}
      \Compl\Invp{\Seelyt_{\With{X_0}Y,\With{X_1}Y}}
      \Compl\Excl q
      \Compl\Seelyt_{\With{X_0}{X_1},Y}\\
    &=\Seelyt_{\Tens{X_0}{Y},\Tens{X_1}{Y}}
      \Compl\Tensp{\Mont_{X_0,Y}}{\Mont_{X_1,Y}}
      \Compl\Tensp{\Invp{\Seelyt_{X_0,Y}}}{\Invp{\Seelyt_{X_1,Y}}}
      \Compl\Invp{\Seelyt_{\With{X_0}Y,\With{X_1}Y}}
      \Compl\Excl q
      \Compl\Seelyt_{\With{X_0}{X_1},Y}
  \end{align*}
  hence
  \begin{align*}
    f_1\Compl\Tensp{\Seelyt_{X_0,X_1}}{\Excl Y}
    &=\Seelyt_{\Tens{X_0}{Y},\Tens{X_1}{Y}}
      \Compl\Tensp{\Mont_{X_0,Y}}{\Mont_{X_1,Y}}
      \Compl\Tensp{\Invp{\Seelyt_{X_0,Y}}}{\Invp{\Seelyt_{X_1,Y}}}
      \Compl\Invp{\Seelyt_{\With{X_0}Y,\With{X_1}Y}}\\
    &\Textsep
      \Compl\Excl q
      \Compl\Seelyt_{\With{X_0}{X_1},Y}
      \Compl\Tensp{\Seelyt_{X_0,X_1}}{\Excl Y}\\
    &=\Seelyt_{\Tens{X_0}{Y},\Tens{X_1}{Y}}
      \Compl\Tensp{\Mont_{X_0,Y}}{\Mont_{X_1,Y}}
      \Compl\Invp{\Seely^4_{X_0,Y,X_1,Y}}
      \Compl\Excl q
      \Compl\Seely^3_{X_0,X_1,Y}\\
    &=\Seelyt_{\Tens{X_0}{Y},\Tens{X_1}{Y}}
      \Compl\Tensp{\Mont_{X_0,Y}}{\Mont_{X_1,Y}}
      \Compl\Sym_{2,3}
      \Compl(\Excl{X_0}\ITens\Excl{X_1}\ITens\Contr Y)
  \end{align*}
  by the monoidality properties of the Seely isomorphisms.
\end{proof}

\subsubsection{Coalgebras of the resource comonad} %
\label{sec:bang-coalgebras}
A $\oc$-coalgebra is a pair %
$P=(\Coalgca P,\Coalgm P)$ where $\Coalgca P$ is an object of $\cL$
and %
$\Coalgm P\in\cL(\Coalgca P,\Excl{\Coalgca P})$ satisfies
\[
  \begin{tikzcd}
    \Coalgca P\ar[r,"\Coalgm P"]
    \ar[rd,swap,"\Id"]
    &
    \Excl{\Coalgca P}\ar[d,"\Der{\Coalgca P}"]
    \\
    &
    \Coalgca P
  \end{tikzcd}
  \Treesep
  \begin{tikzcd}
    \Coalgca P\ar[r,"\Coalgm P"]\ar[d,swap,"\Coalgm P"]
    &
    \Excl{\Coalgca P}\ar[d,"\Digg{\Coalgca P}"]
    \\
    \Excl{\Coalgca P}\ar[r,"\Excl{\Coalgm P}"]
    &
    \Excll{\Coalgca P}
  \end{tikzcd}
\]
Given coalgebras $P$ and $Q$, a coalgebra morphism from $P$ to $Q$ is
an $f\in\cL(\Coalgca P,\Coalgca Q)$ such that the following square
commutes
\[
  \begin{tikzcd}
    \Coalgca P\ar[r,"f"]\ar[d,swap,"\Coalgm P"]
    &
    \Coalgca Q\ar[d,"\Coalgm Q"]\\
    \Excl{\Coalgca P}\ar[r,"\Excl f"]
    &
    \Excl{\Coalgca Q}
  \end{tikzcd}
\]
The category so defined is the Eilenberg-Moore category $\Em\cL$
associated with the comonad $\Excl\_$. We will use the following
standard result for which we refer to~\cite{Mellies09}.
\begin{theorem}
  The Eilenberg-Moore category $\Em\cL$ of a resource category $\cL$
  is cartesian with final object $(\Sone,\Monz)$ simply denoted as
  $\Sone$ and cartesian product of $P_0,P_1$ the coalgebra %
  $(\Tens{\Coalgca{P_0}}{\Coalgca{P_0}},\Mont_{\Coalgca{P_0},\Coalgca{P_1}}
  \Compl\Tensp{\Coalgm{P_0}}{\Coalgm{P_1}})$ denoted as
  $\Tens{P_0}{P_1}$ with projection %
  $\Projc0\in\Em\cL(\Tens{P_0}{P_1},P_0)$ defined as the following
  composition of morphisms
  \[
    \begin{tikzcd}
      \Tens{\Coalgca{P_0}}{\Coalgca{P_1}}
      \ar[r,"\Tens{\Coalgm{P_0}}{\Coalgca{P_1}}"]
      &[1.4em]
      \Tens{\Excl{\Coalgca{P_0}}}{\Coalgca{P_1}}
      \ar[r,"\Tens{\Weak{\Coalgca{P_0}}}{\Coalgca{P_1}}"]
      &[2em]
      \Tens{\Sone}{\Coalgca{P_1}}\ar[r,"\Leftu_{\Coalgca{P_1}}"]
      &[-1em]
      \Coalgca{P_1}
    \end{tikzcd}
  \]
  and similarly for $\Projc1\in\Em\cL(\Tens{P_0}{P_1},P_1)$. %
  And given $f_i\in\Em\cL(Q,P_i)$ for $i=0,1$, the unique morphism %
  $\Tuplec{f_0,f_1}\in\Em\cL(Q,\Tens{P_0}{P_1})$ such that %
  $\Projc i\Compl\Tuplec{f_0,f_1}=f_i$ is defined as the following
  composition of morphisms
  \[
    \begin{tikzcd}
      \Coalgca Q\ar[r,"\Coalgm Q"]
      &[-1em]
      \Excl{\Coalgca Q}\ar[r,"\Contr{\Coalgca Q}"]
      &
      \Tens{\Excl{\Coalgca Q}}{\Excl{\Coalgca Q}}
      \ar[r,"\Tens{\Der{\Coalgca Q}}{\Der{\Coalgca Q}}"]
      &[1.8em]
      \Tens{\Coalgca Q}{\Coalgca Q}\ar[r,"\Tens{f_0}{f_1}"]
      &
      \Tens{\Coalgca{P_0}}{\Coalgca{P_1}}
    \end{tikzcd}
  \]
  Last, the unique morphism $P\to\Sone$ in $\Em\cL$ is
  \(
    \begin{tikzcd}
      \Coalgca P\ar[r,"\Coalgm P"]
      &[-1em]
      \Excl{\Coalgca P}\ar[r,"\Weak{\Coalgca P}"]
      &
      \Sone
    \end{tikzcd}
  \).
\end{theorem}

An immediate consequence of this theorem is the following observation.
\begin{proposition}\label{prop:coalg-comon}
  Let $P$ be an object of $\Em\cL$, %
  $u\in\Em\cL(P,\Sone)$ and %
  $d\in\Em\cL(P,\Tens PP)$ be such that %
  \[
    \begin{tikzcd}
      \Coalgca P\ar[r,"d"]\ar[dr,swap,"\Inv{\Leftu_{\Coalgca P}}"]
      &
      \Tens{\Coalgca P}{\Coalgca P}\ar[d,"\Tens u{\Coalgca P}"]\\
      &
      \Tens\Sone{\Coalgca P}
    \end{tikzcd}
    \Treesep
    \begin{tikzcd}
      \Coalgca P\ar[r,"d"]\ar[dr,swap,"\Inv{\Rightu_{\Coalgca P}}"]
      &
      \Tens{\Coalgca P}{\Coalgca P}\ar[d,"\Tens{\Coalgca P}u"]\\
      &
      \Tens\Sone{\Coalgca P}
    \end{tikzcd}
  \]
  commute.
  Then $u=\Weak P\Compl\Coalgm P$ and %
  $d=\Tuplec{\Coalgca P,\Coalgca P} =\Tensp{\Der{\Coalgca
      P}}{\Der{\Coalgca P}} \Compl\Contr{\Coalgca P} \Compl\Coalgm P$
  and the following diagram commutes in $\cL$.
  \[
    \begin{tikzcd}
      \Coalgca P\ar[r,"d"]\ar[d,swap,"\Coalgm P"]
      &
      \Tens{\Coalgca P}{\Coalgca P}\ar[d,"\Tens{\Coalgm P}{\Coalgm P}"]
      \\
      \Excl{\Coalgca P}\ar[r,"\Contr{\Coalgca P}"]
      &
      \Tens{\Excl{\Coalgca P}}{\Excl{\Coalgca P}}
    \end{tikzcd}
  \]
\end{proposition}
\begin{proof}
  The first equation results from the universal property of the
  terminal object. The second one results from the universal property
  of the cartesian product and from the commutation of
  \[
    \begin{tikzcd}
      \Coalgca P\ar[r,"d"]\ar[dr,swap,"\Coalgca P"]
      &
      \Tens{\Coalgca P}{\Coalgca P}\ar[d,"\Projc i"]
      \\
      &
      \Coalgca P
    \end{tikzcd}
  \]
  since
  \( \Projc0\Compl d =\Leftu_{\Coalgca P} \Compl\Tensp{\Weak{\Coalgca
      P}}{\Coalgca P} \Compl\Tensp{\Coalgm P}{\Coalgca P} \Compl d
  =\Leftu_{\Coalgca P} \Compl\Tensp{u}{\Coalgca P} \Compl
  d=\Id_{\Coalgca P} \) and similarly for $\Projc1$.

  For the last commutation, we have
  \begin{align*}
    \Tensp{\Coalgm P}{\Coalgm P}\Compl d
    &=\Tensp{\Coalgm P}{\Coalgm P}
      \Compl\Tensp{\Der{\Coalgca P}}{\Der{\Coalgca P}}
      \Compl\Contr{\Coalgca P}
      \Compl\Coalgm P\\
    &=\Tensp{\Der{\Excl{\Coalgca P}}}{\Der{\Excl{\Coalgca P}}}
      \Tensp{\Excl{\Coalgm P}}{\Excl{\Coalgm P}}
      \Compl\Contr{\Coalgca P}
      \Compl\Coalgm P
      \text{\quad by naturality of }\Der{}\\
    &=\Tensp{\Der{\Excl{\Coalgca P}}}{\Der{\Excl{\Coalgca P}}}
      \Compl\Contr{\Excl{\Coalgca P}}
      \Compl\Excl{\Coalgm P}
      \Compl\Coalgm P
      \text{\quad by naturality of }\Contr{}\\
    &=\Tensp{\Der{\Excl{\Coalgca P}}}{\Der{\Excl{\Coalgca P}}}
      \Compl\Contr{\Excl{\Coalgca P}}
      \Compl\Digg{\Coalgca P}
      \Compl\Coalgm P
      \text{\quad since }\Coalgm{P}\text{ is a coalgebra structure}\\
    &=\Tensp{\Der{\Excl{\Coalgca P}}}{\Der{\Excl{\Coalgca P}}}
      \Compl\Tensp{\Digg{\Coalgca P}}{\Digg{\Coalgca P}}
      \Compl\Contr{\Coalgca P}
      \Compl\Coalgm P
      \text{\quad by definition of }\Contr{}\\
    &=\Contr{\Coalgca P}
      \Compl\Coalgm P\,.
  \end{align*}
\end{proof}

\subsubsection{Lafont categories and the free exponential} %
\label{sec:Lafont-cat}
%
% A commutative comonoid in $\cL$ is a tuple %
% $C=(\Comonca C,\Comonw C,\Comonc C)$ such where %
% $\Comonca C$ is an object of $\cL$, %
% $\Comonw C\in\cL(\Comonca C,\Sone)$ and %
% $\Comonc C\in\cL(\Comonca C,\Tens{\Comonca C}{\Comonca C})$ such that
% the following diagrams commute
% \[
%   \begin{tikzcd}
%     \Tens\Sone{\Comonca C}
%     &[1em]
%     \Tens{\Comonca C}{\Comonca C}
%     \ar[l,swap,"\Tens{\Comonw C}{\Comonca C}"]
%     \ar[r,"\Tens{\Comonca C}{\Comonw C}"]
%     &[1em]
%     \Tens{\Comonca C}{\Sone}
%     \\
%     &
%     \Comonca C
%     \ar[u,"\Comonc C"]
%     \ar[ul,"\Invp{\Leftu_{\Comonca C}}"]
%     \ar[ur,swap,"\Invp{\Rightu_{\Comonca C}}"]
%   \end{tikzcd}
%   \Treesep
%   \begin{tikzcd}
%     \Comonca C
%     &
%     \Tens{\Comonca C}{\Comonca C}
%     &
%     \Tens{\Tensp{\Comonca C}{\Comonca C}}{\Comonca C}
%     \\
%     \Tens{\Comonca C}{\Comonca C}
%     &&
%     \Tens{\Comonca C}{\Tensp{\Comonca C}{\Comonca C}}    
%   \end{tikzcd}
% \]
% Let $\Cm\cL$ be the category of commutative comonoids over $\cL$. %

In many interesting models of LL, the exponential resource modality is
completely determined by the tensor product; in that case one says
that the exponential is free. We provide the precise definition of
such categories and give some of their properties that we shall use in
the paper.

Let $\cL$ be an SMC. %
Remember from~\cite{Mellies09} that $\cL$ is a \emph{Lafont category}
if the forgetful functor %
$U:\Cm\cL\to\cL$ has a right adjoint $E:\cL\to\Cm\cL$ which maps an
object $X$ to a commutative comonoid %
$(\Excl X,\Weak X,\Contr X)$.
In that case we use $(\Excl\_,\Der{},\Digg{})$ for the associated
comonad $UE$ called the \emph{free exponential} of the SMC $\cL$.

More explicitly this means that for any object $X$ of $\cL$, %
for any commutative comonoid %
$C=(\Comonca C,\Comonw C:\Comonca C\to\Sone, \Comonc C:\Comonca
C\to\Tens{\Comonca C}{\Comonca C})$ and any %
$f\in\cL(\Comonca C,X)$, there is exactly one morphism %
$f^\ITens\in\cL/X((\Comonca C,f),(\Excl X,\Der X))$ which is a
comonoid morphism.
In other words there is exactly one morphism %
$f^\ITens\in\cL(\Coalgca C,\Excl X)$ such that the three following
diagrams commute.
\[
  \begin{tikzcd}
    \Comonca C
    \ar[r,"f^\ITens"]
    \ar[rd,swap,"f"]
    &
    \Excl X
    \ar[d,"\Der X"]
    \\
    &
    X
  \end{tikzcd}
  \Treesep
  \begin{tikzcd}
    \Comonca C
    \ar[r,"f^\ITens"]
    \ar[rd,swap,"\Comonw C"]
    &
    \Excl X
    \ar[d,"\Weak X"]
    \\
    &
    \Sone
  \end{tikzcd}
  \Treesep
  \begin{tikzcd}
    \Comonca C
    \ar[r,"f^\ITens"]
    \ar[d,swap,"\Comonc C"]
    &[1em]
    \Excl X
    \ar[d,"\Contr X"]
    \\
    \Tens{\Comonca C}{\Comonca C}
    \ar[r,"\Tens{f^\ITens}{f^\ITens}"]
    &
    \Tens{\Excl X}{\Excl X}
  \end{tikzcd}
\]
\begin{lemma}\label{lemma:comon-coalg}
  Let $\cL$ be a Lafont category. For any commutative comonoid $C$
  there is exactly one morphism
  $\delta_C\in\cL(\Comonca C,\Excl{\Comonca C})$ such that the
  following diagrams commute.
  \[
    \begin{tikzcd}
      \Comonca C \ar[r,"\delta_C"] \ar[rd,swap,"\Id"]
      & \Excl{\Comonca C}
      \ar[d,"\Der{\Comonca C}"]
      \\
      & \Comonca C
    \end{tikzcd}
    \Treesep
    \begin{tikzcd}
      \Comonca C \ar[r,"\delta_C"] \ar[rd,swap,"\Comonw C"] &
      \Excl{\Comonca C} \ar[d,"\Weak{\Comonca C}"]
      \\
      & \Sone
    \end{tikzcd}
    \Treesep
    \begin{tikzcd}
      \Comonca C \ar[r,"\delta_C"]
      \ar[d,swap,"\Comonc C"]
      &[1em]
      \Excl{\Comonca C}
      \ar[d,"\Contr{\Comonca C}"]
      \\
      \Tens{\Comonca C}{\Comonca C}
      \ar[r,"\Tens{\delta_C}{\delta_C}"]
      & \Tens{\Excl{\Comonca C}}{\Excl{\Comonca C}}
    \end{tikzcd}
  \]
  Moreover $(\Comonca C,\delta_C)$ is a $\oc$-coalgebra.
\end{lemma}
\begin{proof}
  The first part of the statement is just a special case of the
  universal property with $X=\Comonca C$ and $f=\Id_X$. For the second
  part we only have to prove
  \[
    \begin{tikzcd}
      \Comonca C
      \ar[r,"\delta_C"]
      \ar[d,swap,"\delta_C"]
      &
      \Excl{\Comonca C}
      \ar[d,"\Excl{\delta_C}"]
      \\
      \Excl{\Comonca C}
      \ar[r,"\Digg{\Coalgca C}"]
      &
      \Excll{\Comonca C}
    \end{tikzcd}
  \]
  Setting %
  $f_1=\Excl{\delta_C}\Compl\delta_C$ and %
  $f_2=\Digg{\Comonca C}\Compl\delta_C$, observe first that %
  $f_1,f_2\in\Cm\cL(C,(\Excll{\Comonca C},\Comonc{\Excl{\Comonca
      C}},\Comonw{\Excl{\Comonca C}}))$ because both are defined by
  composing morphisms in that category. The equation $f_1=f_2$ follows
  by universality, observing that
\[
  \begin{tikzcd}
    \Comonca C
    \ar[r,"f_i"]
    \ar[rd,swap,"\delta_C"]
    &
    \Excl {\Excl{\Comonca C}}
    \ar[d,"\Der {\Excl{\Comonca C}}"]
    \\
    &
    {\Excl{\Comonca C}}
  \end{tikzcd}
\]
for $i=1,2$, which readily results from the naturality of $\Der{}$ and
from the definition of a comonad.
\end{proof}
Here are two important special cases of the above. First, there is
exactly one morphism $\Monz\in\cL(\Sone,\Excl\Sone)$ such that
\[
  \begin{tikzcd}
    \Sone \ar[r,"\Monz"] \ar[rd,swap,"\Id"] & \Excl{\Sone}
    \ar[d,"\Der{\Sone}"]
    \\
    & \Sone
  \end{tikzcd}
  \Treesep
  \begin{tikzcd}
    \Sone \ar[r,"\Monz"] \ar[rd,swap,"\Comonw C"] & \Excl{\Sone}
    \ar[d,"\Id"]
    \\
    & \Sone
  \end{tikzcd}
  \Treesep
  \begin{tikzcd}
    \Sone \ar[r,"\Monz"] \ar[d,swap,"\Invp{\Leftu_\Sone}"] &[1em]
    \Excl{\Sone} \ar[d,"\Contr{\Sone}"]
    \\
    \Tens{\Sone}{\Sone} \ar[r,"\Tens{\Monz}{\Monz}"] &
    \Tens{\Excl{\Sone}}{\Excl{\Sone}}
  \end{tikzcd}
\]
Next there is exactly one morphism %
$\Mont_{X,Y}\in\cL(\Tens{\Excl X}{\Excl Y},\Exclp{\Tens XY})$ such that
\[
  \begin{tikzcd}
    \Tens{\Excl X}{\Excl Y}
    \ar[r,"\Mont_{X,Y}"]
    \ar[rd,swap,"\Tens{\Der X}{\Der Y}"]
    &
    \Exclp{\Tens XY}
    \ar[d,"\Der{\Tens XY}"]
    \\
    &
    \Tens XY
  \end{tikzcd}
  \Treesep
  \begin{tikzcd}
    \Tens{\Excl X}{\Excl Y}
    \ar[r,"\Mont_{X,Y}"]
    \ar[d,swap,"\Tens{\Weak X}{\Weak Y}"]
    &
    \Exclp{\Tens XY}
    \ar[d,"\Weak{\Tens XY}"]
    \\
    \Tens\Sone\Sone
    \ar[r,"\Leftu_\Sone"]
    &
    \Sone
  \end{tikzcd}
\]
\[
  \begin{tikzcd}
    \Tens{\Excl X}{\Excl Y}
    \ar[rr,"\Mont_{X,Y}"]
    \ar[d,swap,"\Tens{\Contr X}{\Contr Y}"]
    &[-0.8em]&[2.6em]
    \Exclp{\Tens XY}
    \ar[d,"\Contr{\Tens XY}"]
    \\
    \Tens{\Tens{\Excl X}{\Excl X}}{\Tens{\Excl Y}{\Excl Y}}
    \ar[r,"\Sym_{2,3}"]
    &
    \Tens{\Tens{\Excl X}{\Excl Y}}{\Tens{\Excl X}{\Excl Y}}
    \ar[r,"\Tens{\Mont_{X,Y}}{\Mont_{X,Y}}"]
    &
    \Tens{\Exclp{\Tens XY}}{\Exclp{\Tens XY}}
  \end{tikzcd}
\]
These two morphisms turn $\Excl\_$ into a lax monoidal comonad on the
SMC $\cL$.

The correspondence %
$C\mapsto(\Comonca C,\delta_C)$ can be turned into a functor %
$\Calgofcmon:\Cm\cL\to\Em\cL$ acting as the identity on morphisms. Let
indeed $f\in\Cm\cL(C,D)$, it suffices to prove that %
$\delta_D\Compl f=\Excl f\Compl\delta_C\in\cL(\Comonca
C,\Excl{\Comonca D})$. Let $f_0=\delta_D\Compl f$ and %
$f_1=\Excl f\Compl\delta_C$. %
By the universal property, it suffices to prove that the three
following diagrams commute for $i=0,1$:
\[
  \begin{tikzcd}
    \Comonca C
    \ar[r,"f_i"]
    \ar[rd,swap,"f"]
    &
    \Excl {\Comonca D}
    \ar[d,"\Der {\Comonca D}"]
    \\
    &
    {\Comonca D}
  \end{tikzcd}
  \Treesep
  \begin{tikzcd}
    \Comonca C
    \ar[r,"f_i"]
    \ar[rd,swap,"\Comonw C"]
    &
    \Excl {\Comonca D}
    \ar[d,"\Weak {\Comonca D}"]
    \\
    &
    \Sone
  \end{tikzcd}
  \Treesep
  \begin{tikzcd}
    \Comonca C
    \ar[r,"f_i"]
    \ar[d,swap,"\Comonc C"]
    &[1em]
    \Excl {\Comonca D}
    \ar[d,"\Contr {\Comonca D}"]
    \\
    \Tens{\Comonca C}{\Comonca C}
    \ar[r,"\Tens{f_i}{f_i}"]
    &
    \Tens{\Excl {\Comonca D}}{\Excl {\Comonca D}}
  \end{tikzcd}
\]
These commutations follow from the commutations satisfied by %
$\delta_C$ and $\delta_D$ and from the fact that $f\in\Cm\cL(C,D)$. As
an example of these computations, we have
\begin{align*}
  \Contr{\Comonca D}\Compl f_0
  &=\Contr{\Comonca D}\Compl\delta_D\Compl f\\
  &=\Tensp{\delta_D}{\delta_D}\Compl\Comonc D\Compl f\\
  &=\Tensp{\delta_D}{\delta_D}\Compl\Tensp{f}{f}\Compl\Comonc C\\
  &=\Tensp{f_0}{f_0}\Compl\Comonc C
\end{align*}
and
\begin{align*}
  \Contr{\Comonca D}\Compl f_1
  &=\Contr{\Comonca D}\Compl\Excl f\Compl\delta_C\\
  &=\Tensp{\Excl f}{\Excl f}\Compl\Contr{\Comonca C}\Compl \delta_C\\
  &=\Tensp{\Excl f}{\Excl f}
    \Compl\Tensp{\delta_C}{\delta_C}\Compl\Comonc C\\
  &=\Tensp{f_1}{f_1}\Compl\Comonc C  \,.
\end{align*}

Conversely given a $\oc$-coalgebra %
$P=(\Coalgca P,\Coalgm P)$ one can define a commutative comonoid
structure on $\Coalgca P$ by the following two morphisms
\[
  \begin{tikzcd}
    \Coalgca P\ar[r,"\Coalgm P"]
    &
    \Excl{\Coalgca P}\ar[r,"\Weak{\Coalgca P}"]
    &[0.6em]
    \Sone
    &[2em]\\[-1em]
    \Coalgca P\ar[r,"\Coalgm P"]
    &
    \Excl{\Coalgca P}\ar[r,"\Contr{\Coalgca P}"]
    &
    \Tens{\Excl{\Coalgca P}}{\Excl{\Coalgca P}}
    \ar[r,"\Tens{\Der{\Coalgca P}}{\Der{\Coalgca P}}"]
    &
    \Tens{\Coalgca P}{\Coalgca P}
  \end{tikzcd}
\]
that we respectively denote as %
$\Coalgw P$ and $\Coalgc P$. This correspondence %
$P\mapsto\Cmonofcalg(P)=(\Coalgca P,\Coalgw P,\Coalgc P)$ can be
turned into a functor %
$\Cmonofcalg:\Em\cL\to\Cm\cL$ acting as the identity on morphisms.

\begin{theorem}\label{th:comon-coalg-lafont}
  For any Lafont SMC $\cL$, the functors $\Calgofcmon$ and
  $\Cmonofcalg$ define an isomorphism of categories between $\Cm\cL$
  and $\Em\cL$.
\end{theorem}
\begin{proof}
  Let $C\in\Cm\cL$ and let $P=\Calgofcmon(C)$ so that %
  $\Coalgca P=\Comonca C$ and $\Coalgm P=\delta_C$. %
  Let $D=\Cmonofcalg(P)$ so that $\Comonca D=\Comonca C$,
  \begin{align*}
    \Comonw D&=\Weak{\Coalgca P}\Compl\Coalgm P
               =\Weak{\Coalgca C}\Compl\delta_C=\Comonw C\\
    \Comonc D&=\Tensp{\Der{\Coalgca P}}{\Der{\Coalgca P}}
               \Compl\Contr{\Coalgca P}
               \Compl\Coalgm P\\
             &=\Tensp{\Der{\Comonca C}}{\Der{\Comonca C}}
               \Compl\Contr{\Comonca C}
               \Compl\delta_C\\
             &=\Tensp{\Der{\Comonca C}}{\Der{\Comonca C}}
               \Tensp{\delta_C}{\delta_C}
               \Compl\Contr{\Comonca C}\\
             &=\Contr{\Comonca C}\,.
  \end{align*}
  Conversely let $P\in\Em\cL$. Let $C=\Cmonofcalg(P)$ so that %
  $\Comonca C=\Coalgca P$, %
  $\Comonw C=\Weak{\Coalgca P}\Compl\Coalgm P$ and
  $\Comonc C=\Tensp{\Der{\Coalgca P}}{\Der{\Coalgca P}}
  \Compl\Contr{\Coalgca P} \Compl\Coalgm P$. Let %
  $Q=\Calgofcmon(C)=(\Comonca P,\delta_C)$. To prove that %
  $\delta_C=\Coalgm P$ it suffices to show that the following diagrams
  commute
  \[
    \begin{tikzcd}
      \Comonca C \ar[r,"\Coalgm P"] \ar[rd,swap,"\Id"]
      & \Excl{\Comonca C}
      \ar[d,"\Der{\Comonca C}"]
      \\
      & \Comonca C
    \end{tikzcd}
    \Treesep
    \begin{tikzcd}
      \Comonca C \ar[r,"\Coalgm P"] \ar[rd,swap,"\Comonw C"] &
      \Excl{\Comonca C} \ar[d,"\Weak{\Comonca C}"]
      \\
      & \Sone
    \end{tikzcd}
    \Treesep
    \begin{tikzcd}
      \Comonca C \ar[r,"\Coalgm P"]
      \ar[d,swap,"\Comonc C"]
      &[1em]
      \Excl{\Comonca C}
      \ar[d,"\Contr{\Comonca C}"]
      \\
      \Tens{\Comonca C}{\Comonca C}
      \ar[r,"\Tens{\Coalgm P}{\Coalgm P}"]
      & \Tens{\Excl{\Comonca C}}{\Excl{\Comonca C}}
    \end{tikzcd}
  \]
  which results from the definition of $C$ and from the fact that $P$
  is a coalgebra. Let us check for instance the last one:
  \begin{align*}
    \Tensp{\Coalgm P}{\Coalgm P}
    \Compl\Comonc C
    &=\Tensp{\Coalgm P}{\Coalgm P}
      \Compl\Tensp{\Der{\Coalgca P}}{\Der{\Coalgca P}}
      \Compl\Contr{\Coalgca P}
      \Compl\Coalgm P\\
    &=\Tensp{\Der{\Excl{\Coalgca P}}}{\Der{\Excl{\Coalgca P}}}
      \Compl\Tensp{\Excl{\Coalgm P}}{\Excl{\Coalgm P}}
      \Compl\Contr{\Coalgca P}
      \Compl\Coalgm P\\
    &=\Tensp{\Der{\Excl{\Coalgca P}}}{\Der{\Excl{\Coalgca P}}}
      \Compl\Contr{\Excl{\Coalgca P}}
      \Compl\Excl{\Coalgm P}
      \Compl\Coalgm P\\
    &=\Tensp{\Der{\Excl{\Coalgca P}}}{\Der{\Excl{\Coalgca P}}}
      \Compl\Contr{\Excl{\Coalgca P}}
      \Compl\Digg{\Coalgca P}
      \Compl\Coalgm P\\
    % &=\Tensp{\Der{\Excl{\Coalgca P}}}{\Der{\Excl{\Coalgca P}}}
    %   \Compl\Contr{\Excl{\Coalgca P}}
    %   \Compl\Digg{\Coalgca P}
    %   \Compl\Coalgm P\\
    &=\Tensp{\Der{\Excl{\Coalgca P}}}{\Der{\Excl{\Coalgca P}}}
      \Tensp{\Digg{\Coalgca P}}{\Digg{\Coalgca P}}
      \Compl\Contr{\Coalgca P}
      \Compl\Coalgm P\\
    &=\Compl\Contr{\Comonca C}
      \Compl\Coalgm P
  \end{align*}
  where we have used in particular the fact that for any $X\in\cL$,
  one has %
  $\Digg X\in\Cm\cL(E(X),E(\Excl X))$ by the fact that the comonad
  $\Excl\_$ is induced by the adjunction $U\Adj E$.

  This shows that $\Cmonofcalg$ and $\Calgofcmon$ define a bijective
  correspondence on objects and since both functors act as the
  identity on morphisms, our contention is proven.
\end{proof}

In that way we retrieve the fact that $\Em\cL$ is cartesian since
$\Cm\cL$ is always cartesian by Theorem~\ref{th:comon-cat-cart} (even
if $\cL$ is not Lafont). Remember that in the general (not necessarily
Lafont) case the fact that $\Em\cL$ is cartesian could be proven under
the additional assumption that $\cL$ is a resource category. Remember
also that a cartesian Lafont SMC is automatically a resource category,
see~\cite{Mellies09}.

\begin{lemma}\label{lemma:lafont-term-prod-coalg}
  Let $C_0,C_1\in\Cm\cL$. Remember that we use $\Tens{C_0}{C_1}$ for
  the cartesian product of $C_0$ and $C_1$ in $\Cm\cL$ (see
  Theorem~\ref{th:comon-cat-cart}). Then we have
  \begin{align*}
    \delta_\Sone=\Monz
    \quad\quad
    \delta_{\Tens{C_0}{C_1}}
    =\Mont_{{\Comonca{C_0}},{\Comonca{C_1}}}
    \Compl\Tensp{\delta_{C_0}}{\delta_{C_1}}
    \in\cL(\Tens{\Comonca{C_0}}{\Comonca{C_1}},
    \Exclp{\Tens{\Comonca{C_0}}{\Comonca{C_1}}})\,.
  \end{align*}
\end{lemma}
\begin{proof}
  One just checks that the right hand morphisms satisfy the three
  diagrams of Lemma~\ref{lemma:comon-coalg}.
\end{proof}

\begin{theorem}\label{th:lafont-weak-contr-coalg-morph}
  Let $\cL$ be a Lafont category and let $C\in\Cm\cL$. Then the
  following diagrams commute
  \[
    \begin{tikzcd}
      \Coalgca C
      \ar[r,"\delta_C"]
      \ar[d,swap,"\Comonw C"]
      &
      \Excl{\Coalgca C}
      \ar[d,"\Excl{\Comonw C}"]
      \\
      \Sone
      \ar[r,"\Monz"]
      &
      \Excl\Sone
    \end{tikzcd}
    \Treesep
    \begin{tikzcd}
      \Comonca C
      \ar[rr,"\delta_C"]
      \ar[d,swap,"\Comonc C"]
      &[0.4em]
      &
      \Excl{\Comonca C}
      \ar[d,"\Excl{\Comonc C}"]
      \\
      \Tens{\Comonca C}{\Comonca C}
      \ar[r,"\Tens{\delta_C}{\delta_C}"]
      &
      \Tens{\Excl{\Comonca C}}{\Excl{\Comonca C}}
      \ar[r,"\Mont_{\Comonca C,\Comonca C}"]
      &
      \Exclp{\Tens{\Comonca C}{\Comonca C}}
    \end{tikzcd}
  \]
\end{theorem}
\begin{proof}
  We deal with the second diagram, the argument for the first one
  being completely similar.
  By Lemma~\ref{ref:weak-contr-comon-morph} we have %
  $\Comonc C\in\Cm\cL(C,\Tens CC)$ and hence %
  (since $\Calgofcmon$ is the identity on morphisms) we have %
  $\Comonc C\in\Em\cL(\Calgofcmon(C),\Calgofcmon(\Tens CC))$ which is
  exactly the diagram under consideration by
  Lemma~\ref{lemma:lafont-term-prod-coalg}.
\end{proof}

\subsubsection{Resource Lafont categories} %
\label{sec:lafont-resource-cat}
A resource Lafont category is a resource category $\cL$ where the
exponential arises in the way explained above; in that case one says
that $\Excl\_$ is the free exponential (it is unique up to unique iso
since it is defined by a universal property). This is equivalent to
requiring that
\begin{itemize}
\item $\cL$ is a Lafont SMC
\item and $\cL$ is cartesian.
\end{itemize}
Indeed when these conditions hold the Seely isomorphisms are uniquely
defined by the universal property of the Lafont SMC $\cL$. The lax
monoidality $(\Monz,\Mont)$ induced by these Seely isomorphisms
coincide with the one which is directly induced by the Lafont property
(again by universality). This is why we used the same notations for
both.

\section{Summable categories}\label{sec:sum-cat}%
Let $\cL$ be a category; composition in $\cL$ is denoted by simple
juxtaposition. We develop a categorical axiomatization of a concept of
finite summability in $\cL$ which will then be a \emph{partially
  additive category}~\cite{ArbibManes80}.
% However
% contrarily to most axiomatization of such categories available in the
% litterature we do not require the category to have products or
% coproducts.
The main idea is to equip $\cL$ with a functor $\Sfun$ which has the
flavor of a monad\footnote{And will actually be shown to have a
  canonical monad structure} and intuitively maps an object $X$ to the
object $\Sfun X$ of all pairs $(x_0,x_1)$ of elements of $X$ whose
sum $x_0+x_1$ is well defined. This is another feature of our approach
which is to give a crucial role to such pairs, which are the values on
which derivatives are computed, very much in the spirit of Clifford's
\emph{dual numbers}. However, contrarily to dual numbers our
structures also axiomatize the actual summation of such pairs.

\begin{Example}
  In order to illustrate the definitions and constructions of the
  paper we will use the category $\COH$ of coherence
  spaces~\cite{Girard87} as a running example. An object of this
  category is a pair $E=(\Web E,\Coh E{}{})$ where $\Web E$ is a set
  (the web of $E$) and $\Coh E{}{}$ is a symmetric and reflexive
  relation on $\Web E$.  The set of cliques of a coherence space $E$
  is
\begin{align*}
  \Cl E=\{x\subseteq\Web E\St\forall a,a'\in x\ \Coh E a{a'}\}\,.
\end{align*}
Equipped with $\subseteq$ as order relation, $\Cl E$ is a cpo. Given
coherence spaces $E$ and $F$, we define the coherence space
$\Limpl EF$ by $\Web{\Limpl EF}=\Web E\times\Web F$ and
\begin{align*}
  \Coh{\Limpl EF}{(a,b)}{(a',b')}
  \text{ if }
  \Coh Ea{a'}\Implies(\Coh Fb{b'}\text{ and }b=b'\Implies a=a')\,.
\end{align*}
\begin{lemma}
  If $s\in\Cl{\Limpl EF}$ and $t\in\Cl{\Limpl FG}$ then $t\Compl s$
% TYPO
% (the compositional relation of $t$ and $s$) belongs to
  (the relational composition of $t$ and $s$) belongs to
  $\Cl{\Limpl EG}$ and the diagonal relation $\Id_E$ belongs to
  $\Cl{\Limpl EE}$.
\end{lemma}
In that way we have turned the class of coherence spaces into a
category $\COH$ with $\COH(E,F)=\Cl{\Limpl EF}$ and $\COH$ is enriched
over pointed sets, with $0=\emptyset$.
This category is cartesian with
$\With{E_0}{E_1}$ given by
$\Web{\With{E_0}{E_1}}=\{0\}\times\Web{E_0}\cup\{1\}\times\Web{E_1}$
% TYPO
% and \(\Proj i=\{((i,a),a)\St a\in\Web E\}\) for $i=0,1$ and, given
and \(\Proj i=\{((i,a),a)\St a\in\Web{E_i}\}\) for $i=0,1$ and, given
$s_i\in\COH(F,E_i)$ (for $i=0,1$),
\begin{align*}
  \Tuple{s_0,s_1}=\{(b,(i,a))\St i\in\{0,1\}\text{ and }(b,a)\in s_i\}\,.
\end{align*}
Given $s\in\COH(E,F)$ and $x\in\Cl E$ one defines
$\Matappa sx\in\Cl F$ by
$\Matappa sx=\{b\in\Web F\St a\in x\text{ and }(a,b)\in s\}$. Given
$x_0,x_1\in\Cl E$ we use $x_0+x_1$ to denote $x_0\cup x_1$ if
$x_0\cup x_1\in\Cl E$ and $x_0\cap x_1=\emptyset$, so that this sum is
not always defined. With these notations observe that
\begin{align*}
  \Matappa s\emptyset=\emptyset
  \text{ and }\Matappa s{(x_0+x_1)}=\Matappa s{x_0}+\Matappa s{x_1}
\end{align*}
(where the right hand side is defined as soon as the left hand side
is) explaining somehow the terminology ``linear maps'' for these
morphisms.
\end{Example}

\begin{definition}
  A \emph{pre-summability structure} on $\cL$ is a tuple
  $(\Sfun,\Sproj0,\Sproj1,\Ssum)$ where $\Sfun:\cL\to\cL$ is a functor
  which preserves the enrichment of $\cL$ (that is $\Sfun 0=0$) and
  $\Sproj0,\Sproj1$ and $\Ssum$ are natural transformation from
  $\Sfun$ to the identity functor such that for any two morphisms
  $f,g\in\cL(Y,\Sfun X)$, if $\Sproj i\Compl f=\Sproj i\Compl g$ for
  $i=0,1$, then $f=g$. In other words, $\Sproj0$ and $\Sproj1$ are
  jointly monic.
\end{definition}

\begin{Example}\label{ex:Sfun-coh-def}
  We give a pre-summability structure on coherence
  spaces. Given a coherence space $E$, the coherence space $\Sfun(E)$
  is defined by $\Web{\Sfun(E)}=\{0,1\}\times\Web E$ and
  $\Coh{\Sfun(E)}{(i,a)}{(i',a')}$ if $i=i'$ and $\Coh Ea{a'}$, or
  % TYPO
  % $i\not=i'$ and $\Scoh Ea{a'}$. Remember that $\Coh Ea{a'}$ means
  $i\not=i'$ and $\Scoh Ea{a'}$. Remember that $\Scoh Ea{a'}$ means
  that $\Coh Ea{a'}$ and $a\not=a'$ (strict coherence relation). %
  Notice that $\Sfun E=\Limplp{\With\Sone\Sone}{E}$ where %
  $\Sone$ is the coherence space whose web is a chosen singleton %
  $\Sonelem$. We shall see in Section~\ref{sec:canonical-sum} that it
  is often possible to define $\Sfun$ in that particular way.
\begin{lemma}
  $\Cl{\Sfun E}$ is isomorphic to the poset of all pairs
  $(x_0,x_1)\in\Cl E^2$ such that $x_0+ x_1$ is defined and belongs to
  $\Cl E$, equipped with the product order.
\end{lemma}
Given $s\in\COH(E,F)$, we define
$\Sfun s\subseteq\Web{\Limpl{\Sfun E}{\Sfun F}}$ by
\begin{align*}
  \Sfun s=\{((i,a),(i,b))\St i\in\{0,1\}\text{ and }(a,b)\in s\}\,.
\end{align*}
Then it is easy to check that $\Sfun s\in\COH(\Sfun E,\Sfun F)$ and
that $\Sfun$ is a functor. This is due to the definition of $s$ which
entails $\Matappa s{(x_0+x_1)}=\Matappa s{x_0}+\Matappa s{x_1}$.

The additional structure is defined as follows:
\begin{align*}
  \Sproj i=\{((i,a),a)\St a\in\Web E\}\text{ and }
  \Ssum=\{((i,a),a)\St i\in\{0,1\}\text{ and }a\in\Web E\}
\end{align*}
which are easily seen to belong to $\COH(\Sfun E,E)$. Notice that
$\Ssum=\Sproj0+\Sproj 1$. Of course
$\Matappa{\Sproj i}{(x_0,x_1)}=x_i$ and
$\Matappa\Ssum{(x_0,x_1)}=x_0+x_1$.
\end{Example}

From now on we assume that we are given such a structure.  We say that
$f_i\in\cL(X,Y)$ (for $i=0,1$) are \emph{summable} if there is a morphism
$g\in\cL(X,\Sfun Y)$ such that
\[
  \begin{tikzcd}
    X\ar[r,"g"]\ar[rd,swap,"f_i"] & \Sfun Y\ar[d,"\Sproj i"]\\
    &Y
  \end{tikzcd}
\]
for $i=0,1$. By definition of a pre-summability structure there is only
one such $g$ if it exists, we denote it as $\Stuple{f_0,f_1}$. When
this is the case we set
$f_0+f_1=\Ssum\Compl\Stuple{f_0,f_1}\in\cL(X,Y)$. We sometimes call
$\Stuple{f_0,f_1}$ the \emph{witness of the summability} of $f_0$ and
$f_1$ and $f_0+f_1$ their \emph{sum}.

\begin{Example}\label{ex:coh-summability-char}
  In the case of coherence spaces, saying that $s_0,s_1\in\COH(E,F)$
  are summable simply means that $s_0\cap s_1=\emptyset$ and
  $s_0\cup s_1\in\COH(E,F)$. This property is equivalent to
  \begin{align*}
    \forall x\in\Cl X\quad (\Matappa{s_0}x,\Matappa{s_1}x)\in\Cl{\Sfun E}
  \end{align*}
  and in that case the witness is defined exactly in the same way as
  $\Tuple{s_0,s_1}\in\COH(E,\With FF)$.
\end{Example}

\begin{lemma}\label{lemma:compl-summable}
  Assume that $f_0,f_1\in\cL(X,Y)$ are summable and that
  $g\in\cL(U,X)$ and $h\in\cL(Y,Z)$. Then $h\Compl f_0\Compl g$ and
  $h\Compl f_1\Compl g$ are summable with witness
  $(\Sfun h)\Compl\Stuple{f_0,f_1}\Compl g\in\cL(U,\Sfun Z)$ and sum
  % TYPO
  % $(\Sfun h)\Compl(f_0+f_1)\Compl g\in\cL(U,Z)$.
  $h\Compl(f_0+f_1)\Compl g\in\cL(U,Z)$.
\end{lemma}
The proof boils down to the naturality of $\Sproj i$ and $\Ssum$. An
easy consequence is that the application of $\Sfun$ to a morphism can
be written as a witness.
\begin{lemma}
  If $f\in\cL(X,Y)$ then
  $f\Compl\Sproj0,f\Compl\Sproj1\in\cL(\Sfun X,Y)$ are summable with
  witness $\Sfun f$ and sum $f$. That is
  $\Sfun f=\Stuple{f\Compl\Sproj0,f\Compl\Sproj1}$.
\end{lemma}

Now using this notion of pre-summability structure we start
introducing additional conditions to define a summability
structure.

Notice that by definition $\Sproj 0$ and $\Sproj 1$ are summable with
$\Id$ as witness and $\Ssum$ as sum. Here is our first condition:

\begin{Axicond}{\Saxcom}
  $\Sproj 1$ and $\Sproj 0$ are summable and the
  witness $\Stuple{\Sproj 1,\Sproj 0}\in\cL(\Sfun X,\Sfun X)$
  satisfies $\Ssum\Compl\Stuple{\Sproj 1,\Sproj 0}=\Ssum$.
\end{Axicond}

Notice that this witness is an involutive iso since
$\Sproj i\Compl\Stuple{\Sproj 1,\Sproj 0}\Compl\Stuple{\Sproj 1,\Sproj
  0}=\Sproj i$ for $i=0,1$.

\begin{lemma}\label{lemma:ssum-com}
  If $f_0,f_1\in\cL(X,Y)$ are summable then $f_1,f_0$ are summable
  with witness $\Stuple{\Sproj 1,\Sproj 0}\Compl\Stuple{f_0,f_1}$ and
  we have $f_0+f_1=f_1+f_0$.
\end{lemma}

\begin{Axicond}{\Saxzero}
  For any $f\in\cL(X,Y)$, the morphisms $f$ and $0\in\cL(X,Y)$ are
  summable and their sum is $f$, that is $\Ssum\Compl\Tuple{f,0}=f$.
\end{Axicond}

By \Saxcom{} this implies that $0$ and $f$ are summable
with $0+f=f$.

Notice that we have four morphisms
$\Sproj 0\Sproj 0,\Sproj 1\Sproj 1,\Sproj 0\Sproj 1,\Sproj 1\Sproj
0\in\cL(\Sfun^2X,X)$.
\begin{lemma}\label{lemma:stuple-sproj-mono2}
  If $f,f'\in\cL(X,\Sfun^2Y)$ satisfy
  $\Sproj i\Compl\Sproj j\Compl f=\Sproj i\Compl\Sproj j\Compl f'$ for
  all $i,j\in\{0,1\}$ then $f=f'$, that is, the
  $\Sproj i\Compl\Sproj j$ are jointly monic.
\end{lemma}
This is an easy consequence of the fact that $\Sproj0,\Sproj1$ are
jointly monic.

\begin{Axicond}{\Saxwit}
  Let $f_{00},f_{01},f_{10},f_{11}\in\cL(X,Y)$ be morphisms such that
  $(f_{00},f_{01})$ and $(f_{10},f_{11})$ are summable, and moreover
  $(f_{00}+f_{01},f_{10}+f_{11})$ is summable. Then the witnesses
  $\Stuple{f_{00},f_{01}},\Stuple{f_{10},f_{11}}\in\cL(X,\Sfun
  X)$ are summable.
\end{Axicond}

The last axiom requires a little preparation.  By
Lemma~\ref{lemma:compl-summable} the pairs of morphisms
$(\Sproj 0\Sproj 0,\Sproj 0\Sproj 1)$ and
$(\Sproj 1\Sproj 0,\Sproj 1\Sproj 1)$ are summable with sums
$\Sproj0\Compl\Ssum$ and $\Sproj1\Compl\Ssum$ respectively. By the
same lemma these two morphisms are summable (with sum
$\Ssum\Compl\Ssum\in\cL(\Sfun^2X,X)$). By Axiom~\Saxwit{} it follows
that the witnesses
$\Stuple{\Sproj 0\Sproj 0,\Sproj 0\Sproj 1},\Stuple{\Sproj 1\Sproj
  0,\Sproj 1\Sproj 1}\in\cL(\Sfun^2X,\Sfun X)$ are summable, let
$\Sflip=\Stuple{\Stuple{\Sproj 0\Sproj 0,\Sproj 0\Sproj
    1},\Stuple{\Sproj 1\Sproj 0,\Sproj 1\Sproj
    1}}\in\cL(\Sfun^2X,\Sfun^2X)$ be the corresponding witness which
is easily seen to be an involutive natural iso using
Lemma~\ref{lemma:stuple-sproj-mono2}. Notice that $\Sflip$ (which is
similar to the flip of a tangent bundle functor) is characterized by
\begin{align*}
  \forall i,j\in\Eset{0,1}\quad \Sproj i\Compl\Sproj j\Compl\Sflip
  =\Sproj j\Compl\Sproj i\,.
\end{align*}
We can now state our last axiom.

\begin{Axicond}{\Saxass}
  The following diagram commutes.
  \[
    \begin{tikzcd}
      \Sfun^2X\ar[r,"\Sflip"]\ar[dr,swap,"\Ssum_{\Sfun X}"]
      &\Sfun^2X\ar[d,"\Sfun\Ssum_X"]\\
      &\Sfun X
    \end{tikzcd}
  \]
\end{Axicond}

  Let us see what this condition has to do with associativity of summation.
\begin{lemma}\label{lemma:Ssum-assoc4}
  Let $f_{00},f_{01},f_{10},f_{11}\in\cL(X,Y)$ be morphisms such that
  $(f_{00},f_{01})$ and $(f_{10},f_{11})$ are summable, and moreover
  $(f_{00}+f_{01},f_{10}+f_{11})$ is summable. Then $(f_{00},f_{10})$
  and $(f_{01},f_{11})$ are summable, 
  $(f_{00}+f_{10},f_{01}+f_{11})$ is summable and moreover
  \begin{align*}
    (f_{00}+f_{01})+(f_{10}+f_{11})=(f_{00}+f_{10})+(f_{01}+f_{11})\,.
  \end{align*}
\end{lemma}
\begin{proof}
  By Axiom~\Saxwit{} we have a ``global witness''
  $g=\Stuple{\Stuple{f_{00},f_{01}},
    \Stuple{f_{10},f_{11}}}\in\cL(X,\Sfun^2Y)$. Let
  $g'=\Sflip\Compl g\in\cL(X,\Sfun^2Y)$. We have
  $\Sproj0\Sproj0\Compl g'=f_{00}$ and
  $\Sproj1\Sproj0\Compl g'=f_{10}$ which shows that $f_{00}$ and
  $f_{10}$ are summable with witness
  $\Stuple{f_{00},f_{10}}=\Sproj0\Compl g'\in\cL(X,\Sfun Y)$.
  Similarly $f_{01}$ and $f_{11}\in\cL(X,\Sfun Y)$ are summable with
  witness $\Stuple{f_{01},f_{11}}=\Sproj1\Compl g'$. Since $\Sproj0$
  and $\Sproj1$ are summable, it results from
  Lemma~\ref{lemma:compl-summable} that $\Stuple{f_{00},f_{10}}$ and
  $\Stuple{f_{01},f_{11}}$ are summable with witness
  $\Stuple{\Stuple{f_{00},f_{10}},\Stuple{f_{01},f_{11}}}=g'$. We have
  \begin{align*}
    \Sfun\Ssum_X\Compl\Stuple{\Stuple{f_{00},f_{10}},\Stuple{f_{01},f_{11}}}
    &=\Stuple{\Ssum_X\Compl\Stuple{f_{00},f_{10}},\Ssum\Compl\Stuple{f_{01},f_{11}}}\text{ by Lemma~\ref{lemma:compl-summable}}\\
    &=\Stuple{f_{00}+f_{10},f_{01}+f_{11}}
  \end{align*}
  On the other hand, by Axiom~\Saxass{} and by definition of $g'$ we have
  \begin{align*}
    \Sfun\Ssum_X\Compl\Stuple{\Stuple{f_{00},f_{10}},\Stuple{f_{01},f_{11}}}
    &=\Ssum_{\Sfun X}\Compl\Stuple{\Stuple{f_{00},f_{01}},
      \Stuple{f_{10},f_{11}}}\\
    &=\Stuple{f_{00},f_{01}}+\Stuple{f_{10},f_{11}}
  \end{align*}
  so we have shown that
  \begin{align*}
    \Stuple{f_{00},f_{01}}+\Stuple{f_{10},f_{11}}=\Stuple{f_{00}+f_{10},f_{01}+f_{11}}
  \end{align*}
  that is, the summation of summable pairs is performed componentwise.

  % ligne blanche à enlever
  % !!!!!!
  Next we have that
  $\Ssum_X\Compl\Stuple{f_{00}+f_{10},f_{01}+f_{11}}
  =(f_{00}+f_{10})+(f_{01}+f_{11})$ and, by
  Lemma~\ref{lemma:compl-summable} we know that
  $\Ssum_X\Compl\Stuple{f_{00},f_{01}}=f_{00}+f_{01}$ and
  $\Ssum_X\Compl\Stuple{f_{10},f_{11}}=f_{10}+f_{11}$ are summable
  with sum equal to
  $\Ssum_X\Compl(\Stuple{f_{00},f_{01}}+\Stuple{f_{10},f_{11}})$. This
  shows that
  $(f_{00}+f_{10})+(f_{01}+f_{11})=(f_{00}+f_{01})+(f_{10}+f_{11})$ as
  contended.
\end{proof}

\begin{lemma}\label{lemma:Ssum-assoc3}
  Let $f_0,f_1,f_2\in\cL(X,Y)$ be such that $(f_0,f_1)$ is summable
  and $(f_0+f_1,f_2)$ is summable. Then $(f_1,f_2)$ is summable and
  $(f_0,f_1+f_2)$ is summable and we have
  $(f_0+f_1)+f_2=f_0+(f_1+f_2)$.
\end{lemma}
\begin{proof}
  It suffices to apply Lemma~\ref{lemma:Ssum-assoc4} to
  $f_0,f_1,0,f_2$, using \Saxzero{} for making sure that $(0,f_2)$ is
  summable, with sum $=f_2$.
\end{proof}

\begin{Example}
  All these properties are easy to check in coherence spaces and boil
  down to the standard algebraic properties of set unions.
\end{Example}

\begin{definition}
  A summability structure on $\cL$ is a pre-summability structure
  which satisfies axioms \Saxcom{}, \Saxzero{}, \Saxwit{} and
  \Saxass. We call \emph{summable category} a tuple
  $(\cL,\Sfun,\Sproj0,\Sproj1,\Ssum)$ consisting of a
  category $\cL$ equipped with a summability structure.
\end{definition}

% We define a general notion of summable family of morphisms
% $(f_i)_{i=1}^n$ in $\cL(X,Y)$ by induction on $n$:
% \begin{itemize}
% \item if $n=0$ then $(f_i)_{i=1}^n$ if summable with sum $0$
% \item if $n=1$ then $(f_i)_{i=1}^n$ if summable with sum $f_1$
% \item and $(f_i)_{i=1}^{n+2}$ is summable if $f_{n+1},f_{n+2}$ are
%   summable and $(f_1,\dots,f_{n},f_{n+1}+f_{n+2})$ is summable, and
%   then $f_1+\cdots+f_{n+2}=(f_1+\cdots+f_n)+(f_{n+1}+f_{n+2})$.
% \end{itemize}

We define a general notion of summable family of morphisms
$(f_i)_{i=1}^n$ in $\cL(X,Y)$ together with its sum $f_1+\cdots+f_n$
by induction on $n$:
\begin{itemize}
\item if $n=0$ then $(f_i)_{i=1}^n$ if summable with sum $0$;
\item if $n>0$ then $(f_i)_{i=1}^n$ is summable if %
  $(f_i)_{i=1}^{n-1}$ is summable and %
  $f_1+\cdots+f_{n-1},f_n$ is summable, and then %
  $f_1+\cdots+f_n=(f_1+\cdots+f_{n-1})+f_n$.
\end{itemize}
Of course we use the standard notation %
$\sum_{i=1}^nf_i$ for $f_1+\cdots+f_n$.

\begin{lemma}\label{lemma:left-summable-family}
  If $(f_i)_{i=1}^n$ is summable with $n>0$ then %
  $(f_i)_{i=2}^n$ is summable and %
  $f_1,\sum_{i=2}^nf_i$ are summable and %
  $f_1+\sum_{i=2}^nf_i=\sum_{i=1}^nf_i$.
\end{lemma}
\begin{proof}
  By induction on $n$. If $n=0$ there is nothing to prove so assume
  $n>0$. If $n=1$ the statement results from \Saxzero{} so we assume
  that $n\geq 2$.
  By definition we know that $\List f1{n-1}$ is summable and %
  $\sum_{i=1}^{n-1}f_i+f_n=\sum_{i=1}^nf_i$.
  So by inductive hypothesis $\List f2{n-1}$ is summable, %
  $f_1,\sum_{i=2}^{n-1}f_i$ are summable and
  $f_1+\sum_{i=2}^{n-1}f_i=\sum_{i=1}^{n-1}f_i$. So we can apply %
  Lemma~\ref{lemma:Ssum-assoc3} to %
  $f_1,\sum_{i=2}^{n-1}f_i,f_n$ and hence %
  $\sum_{i=2}^{n-1}f_i,f_n$ are summable which by definition means that %
  $\List f2n$ is summable and %
  $\sum_{i=2}^nf_i=\sum_{i=2}^{n-1}f_i+f_n$, and moreover %
  $f_1,\sum_{i=2}^nf_i$ are summable and %
  $f_1+\sum_{i=2}^nf_i=\sum_{i=1}^{n-1}f_i+f_n=\sum_{i=1}^nf_i$ %
  as contended.
\end{proof}

Now we prove that summability is invariant by permutations. For this
we consider first a circular permutation and then a transposition.
\begin{lemma}\label{lemma:summability-circular}
  If $\List f1n$ are summable then $\List f2{n},f_1$ is summable and %
  $\sum_{i=1}^nf_i=f_2+\cdots+f_{n}+f_1$.
\end{lemma}
\begin{proof}
  This is obvious if $n\leq 1$ so we can assume $n\geq 2$.
  By Lemma~\ref{lemma:left-summable-family} %
  $\List f2n$ are summable and %
  $f_1,\sum_{i=2}^nf_i$ are summable with %
  $f_1+\sum_{i=2}^nf_i=\sum_{i=1}^nf_i$. %
  So $\sum_{i=2}^nf_i,f_1$ are summable by Lemma~\ref{lemma:ssum-com} and
  hence %
  $f_2,\dots,f_n,f_1$ is summable (by definition) with sum equal to %
  $\sum_{i=1}^nf_i$.
\end{proof}

\begin{lemma}\label{lemma:summability-transp}
  If the family $\List f1n$ is summable, with $n\geq 2$, then %
  $\List f1{n-2},f_n,f_{n-1}$ is summable with the same sum.
\end{lemma}
\begin{proof}
  By our assumption, $\List f1{n-2}$ is summable (let us call $g$ its
  sum), $g,f_{n-1}$ are summable and %
  $g+f_{n-1},f_n$ are summable. Moreover %
  $(g+f_{n-1})+f_n=\sum_{i=1}^nf_i$. %
  It follows by Lemma~\ref{lemma:Ssum-assoc3} that %
  $f_{n-1},f_n$ are summable and hence $f_n,f_{n-1}$ are summable with
  $f_n+f_{n-1}=f_{n-1}+f_n$ by Lemma~\ref{lemma:ssum-com}. %
  So we know by Lemma~\ref{lemma:Ssum-assoc3} that %
  $g,f_n+f_{n-1}$ are summable and hence by the same lemma that %
  $g,f_n$ are summable and that %
  $g+f_n,f_{n-1}$ are summable with %
  $(g+f_n)+f_{n-1}=g+(f_n+f_{n-1})=\sum_{i=1}^nf_i$. %
  By definition it follows that
  $\List f1{n-2},f_n$ is a summable family whose sum is %
  $g+f_n$, and then that %
  $\List f1{n-2},f_n,f_{n-1}$ is a summable family whose sum is %
  $\sum_{i=1}^nf_i$, as announced.
\end{proof}

\begin{proposition}\label{prop:gen-summab}
  For any $p\in\Symgrp n$ (the symmetric group) and any family of
  morphisms $(f_i)_{i=1}^n$, the family $(f_i)_{i=1}^n$ is summable
  iff the family $(f_{p(i)})_{i=1}^n$ is summable and then
  $\sum_{i\in I}f_i=\sum_{i\in I}f_{p(i)}$.
\end{proposition}
% From appendix
%
\begin{proof}
  % It suffices to prove that if $(f_1,\dots,f_n)$ is summable then the
  % cyclic permutation $(f_n,f_1,\dots,f_{n-1})$ is summable with the
  % same sum. For $n=0$ and $n=1$, this is trivial. For $n=2$ the
  % property results from Lemma~\ref{lemma:ssum-com}. For $n=3$ it
  % results from Lemma~\ref{lemma:Ssum-assoc3}
  % and~\ref{lemma:ssum-com}. Let us deal with the case $n=4$, the
  % general case being similar. We know that $(f_1,f_2)$ is summable,
  % $(f_3,f_4)$ is summable and $(f_1+f_2,f_3+f_4)$ is summable. It
  % follows that $(f_4,f_3)$ is summable by Lemma~\ref{lemma:ssum-com} and
  % hence, by Lemma~\ref{lemma:Ssum-assoc4} we have that $(f_1,f_4)$ and
  % $(f_2,f_3)$ is summable, so that $(f_4,f_1)$ is summable so that
  % finally $(f_4,f_1,f_2,f_3)$ is summable. The identity of summations
  % results from the corresponding identities in the lemmas we have
  % used.
  Remember that %
  $\Symgrp n$ is generated by the permutations %
  $(1,\dots,n-2,n,n-1)$ (transposition) and %
  $(2,\dots,n,1)$ (circular permutation) and apply %
  Lemmas~\ref{lemma:summability-transp} %
  and~\ref{lemma:summability-circular}.
\end{proof}

So we define an unordered finite family $(f_i)_{i\in I}$ to be
summable if any of its enumerations $(f_{i_1},\dots,f_{i_n})$ is
summable and then we set $\sum_{i\in I}f_i=\sum_{k=1}^nf_{i_k}$.

\begin{theorem}\label{th:summability-subsets}
  A family of morphisms $(f_i)_{i\in I}$ in $\cL(X,Y)$ is summable iff
  for any family of pairwise disjoint sets $(I_j)_{j\in J}$ such that
  $\cup_{j\in J}I_j=I$:
  \begin{itemize}
  \item for each $j\in J$ the restricted family $(f_i)_{i\in I_j}$ is
    summable with sum $\sum_{i\in I_j}f_i\in\cL(X,Y)$
  \item the family $(\sum_{i\in I_j}f_i)_{j\in J}$ is summable
  \end{itemize}
  and then we have $\sum_{i\in I}f_i=\sum_{j\in J}\sum_{i\in I_j}f_i$.
\end{theorem}
\begin{proof}
  By induction on $k=\Card J\geq 1$. If $k=1$ the property trivially
  holds so assume $k>1$. Upon choosing enumerations we can assume
  that %
  $I=\Eset{1,\dots,n}$ and $J=\Eset{1,\dots,k}$, %
  with $n,k\in\Nat$. %
  Thanks to Proposition~\ref{prop:gen-summab} we can choose these
  enumerations in such a way that %
  $I_k=\{l+1,\dots,n\}$ for some $l\in\Eset{1,\dots,n}$. Then by an
  iterated application of the definition of summability and of
  Lemma~\ref{lemma:Ssum-assoc3} we know that the families %
  $\List f1l$ and $\List{f}{l+1}k$ are summable and that %
  $(\sum_{i=1}^lf_i)+(\sum_{j=l+1}^kf_i)=\sum_{i=1}^nl_i$.
  We conclude the proof by applying the inductive hypothesis to %
  $(I_j)_{j=1}^{k-1}$ which satisfies %
  $\Union_{j=1}^{k-1}I_j=\Eset{1,\dots,l}$.
\end{proof}

\begin{remark}
  These properties strongly suggest to consider summability as an
  $n$-ary notion, axiomatized in an operadic way. However in the
  sequel we shall see that the differential operations use $\Sfun X$
  as a space of pairs, and there it is not clear that such an
  operadic approach would be so convenient. This is why we stick (at
  least for the time being) to this ``binary'' axiomatization.
\end{remark}

Another interesting consequence of \Saxass{} is that $\Sfun$ preserves
summability.
\begin{theorem}\label{th:Sfun-preserves-summ}
  Let $f_0,f_1\in\cL(X,Y)$ be summable. Then
  $\Sfun f_0,\Sfun f_1\in\cL(\Sfun X,\Sfun Y)$ are summable, with
  witness $\Stuple{\Sfun f_0,\Sfun f_1}\in\cL(\Sfun X,\Sfun^2Y)$ given
  by $\Stuple{\Sfun f_0,\Sfun f_1}=\Sflip\Compl\Sfun\Stuple{f_0,f_1}$.
  % TYPO: ceci est un oubli idiot!
  And one has $\Sfun f_0+\Sfun f_1=\Sfun(f_0+f_1)$.
\end{theorem}
\begin{proof}
  We must prove that
  $\Sproj i\Compl\Sflip\Compl\Sfun\Stuple{f_0,f_1}=\Sfun f_i$. For
  this we use the fact that $\Sproj0,\Sproj1\in\cL(\Sfun Y,Y)$ are
  jointly monic. We have
  \begin{align*}
    \Sproj j\Compl\Sproj i\Compl\Sflip\Compl\Compl\Sfun\Stuple{f_0,f_1}
    &=\Sproj i\Compl\Sproj j\Compl\Sfun\Stuple{f_0,f_1}\\
    &=\Sproj i\Compl\Stuple{f_0,f_1}\Compl\Sproj j\text{\quad by naturality}\\
    &=f_i\Compl\Sproj j=\Sproj j\Compl\Sfun f_i\text{\quad by naturality.}
  \end{align*}
  % TYPO: preuve du bout qui manquait
  And we have
  \begin{align*}
    \Sfun f_0+\Sfun f_1
    &=\Ssum_{\Sfun Y}\Compl\Stuple{\Sfun f_0,\Sfun f_1}
      \text{\quad by definition}\\
    &=\Ssum_{\Sfun Y}\Compl\Sflip\Compl\Sfun\Stuple{f_0,f_1}\\
    &=\Sfun\Ssum_Y
      \Compl\Sflip^2
      \Compl\Sfun\Stuple{f_0,f_1}
      \text{\quad by \Saxass{}}\\
    &=\Sfun\Ssum_Y
      \Compl\Sfun\Stuple{f_0,f_1}
      \text{\quad since }\Sflip\text{ is involutive}\\
    &=\Sfun(\Ssum_Y\Compl\Stuple{f_0,f_1})
      \text{\quad by functoriality}\\
    &=\Sfun(f_0+f_1)\,.
  \end{align*}
\end{proof}

We will use the notations $\Sin 0=\Stuple{\Id,0}\in\cL(X,\Sfun X)$ and
$\Sin 1=\Stuple{0,\Id}\in\cL(X,\Sfun X)$.
\begin{lemma}
  The morphisms $\Sin0,\Sin1\in\cL(X,\Sfun X)$ are natural in $X$.
\end{lemma}
\begin{proof}
  Let $f\in\cL(X,Y)$. For $i=0,1$ we have %
  $\Sproj i\Compl\Sfun f\Compl\Stuple{\Id,0} =f\Compl\Sproj
  i\Compl\Stuple{\Id,0}$ which is equal to %
  $f$ if $i=0$ and to $0$ if $i=1$ since $f\Compl 0=0$. %
  On the other hand $\Sproj i\Compl\Stuple{\Id,0}\Compl f$ is equal
  to %
  $f$ is $i=0$ and to $0$ if $i=1$ since $0\Compl f=0$. %
  The naturality follows by the fact that $\Sproj0,\Sproj1$ are jointly
  monic.
\end{proof}

Notice that if $\cL$ has products $\With XY$ and coproducts $\Plus XY$
then we have
\[
  \begin{tikzcd}
    \Plus XX\ar[r,"\Cotuple{\Sin0,\Sin1}"] &\Sfun
    X\ar[r,"\Tuple{\Sproj0,\Sproj 1}"] &\With XX
  \end{tikzcd}
\]
where $\Cotuple{\Sin0,\Sin1}$ is the co-pairing of %
$\Sin0$ and $\Sin1$, locating $\Sfun X$ somewhere in between the
coproduct and the product of $X$ with itself. Notice that, in the case
of coherence spaces, $\Sfun X$ is neither the product $\With XX$ nor
the coproduct $\Plus XX$ in general.

In contrast, if $\cL$ has biproducts, then we necessarily have
$\Sfun X=\With XX=\Plus XX$ with obvious structural morphisms, and
$\cL$ is additive. Of course this is not the situation we are
primarily interested in!

% \begin{remark}
%   We have seen that any summable category
%   $(\cL,\Sfun,\Sproj0,\Sproj1,\Ssum)$ is actually enriched over
%   partial commutative monoids which are structures $(M,0,\mathord+)$
%   where $\mathord+:M\times M\to M$ is a partial function such that
%   $\{0\}\times M\subseteq\Fdom(\mathord +)$ and $0+m=m$, if
%   $(m_0,m_1)\in\Fdom(\mathord+)$ then $(m_1,m_0)\in\Fdom(\mathord+)$
%   and $m_0+m_1=m_1+m_0$ and last, for any $m_0,m_1,m_2\in M$,
%   $(m_0,m_1)\in\Fdom(\mathord+)$ and
%   $(m_0+m_1,m_2)\in\Fdom(\mathord +)$ iff
%   $(m_1,m_2)\in\Fdom(\mathord+)$ and
%   $(m_0,m_1+m_2)\in\Fdom(\mathord +)$ and, when these equivalent
%   conditions hold one has $(m_0+m_1)+m_2=m_0+(m_1+m_2)$. That is,
%   roughly speaking, $\cL$ is a ``partially additive category'' in the
%   sense of various authors~\cite{AbibManes,MascariPedicini,Haghverdi}.
%   Indeed $\cL(X,Y)$ is a partial commutative monoid with $f_0+f_1$
%   defined if $f_0,f_1$ is summable. This is however only one part of
%   the story because the functor $\Sfun$ has the main feature of
%   \emph{internalizing} pairs of summable morphisms as morphisms of the
%   category $\cL$ itself, our witnesses of summability, and summability
%   is meaningful also on these witnesses. This ``functorialization'' of
%   summability is essential in our approach to differentiability as we
%   shall see.
% \end{remark}

\subsection{A monad structure on $\Sfun$}\label{sec:Sfun-monad}
We already noticed that there is a natural transformation
$\Sin 0\in\cL(X,\Sfun X)$. As also mentioned the morphisms
$\Sproj i\Compl\Sproj j\in\cL(\Sfun^2 X,X)$ (for all
$i,j\in\Eset{0,1}$) are summable, so that the morphisms
% TYPO
% $\Sproj0\Compl\Sproj0,\Sproj1\Compl\Sproj0+\Sproj0\Compl\Sproj0\in\cL(\Sfun^2
% X,\Sfun X)$ are summable, let
$\Sproj0\Compl\Sproj0,\Sproj1\Compl\Sproj0+\Sproj0\Compl\Sproj1\in\cL(\Sfun^2
X,\Sfun X)$ are summable by Theorem~\ref{th:summability-subsets}, let
$\Sfunadd=
\Stuple{\Sproj0\Compl\Sproj0,\Sproj1
  \Compl\Sproj0+\Sproj0
  \Compl\Sproj1}\in\cL(\Sfun^2X,\Sfun X)$ %
be the witness of this summability.

\begin{theorem}\label{th:sfun-monade}
  The tuple $(\Sfun,\Sin0,\Sfunadd)$ is a monad on $\cL$ and we have
  $\Sfunadd\Compl\Sflip=\Sfunadd$.
\end{theorem}
\begin{proof}
  The proof is easy and uses the fact that $\Sproj0,\Sproj1$ are
  jointly monic. Let us prove that $\Sfunadd$ is natural so let
  $f\in\cL(X,Y)$, we have
  $\Sproj0\Compl(\Sfun f)\Compl\Sfunadd_{X} =
  f\Compl\Sproj0\Sfunadd_X$ by naturality of $\Sproj0$ and hence
  $\Sproj0\Compl(\Sfun
  f)\Compl\Sfunadd_{X}=f\Compl\Sproj0\Compl\Sproj0$, and
  $\Sproj0\Compl\Sfunadd_Y\Compl(\Sfun^2f)=
  \Sproj0\Compl\Sproj0\Compl(\Sfun^2f)=f\Compl\Sproj0\Compl\Sproj0$ by
  naturality of $\Sproj0$.

  Similarly, using the naturality of $\Sproj1$, we have
  $\Sproj1\Compl(\Sfun f)\Compl\Sfunadd_{X}
  =f\Compl\Sproj1\Sfunadd_X
  =f\Compl(\Sproj0\Compl\Sproj1+\Sproj1\Compl\Sproj0)
  =f\Compl\Sproj0\Compl\Sproj1+f\Compl\Sproj1\Compl\Sproj0$
  and
  $\Sproj1\Compl\Sfunadd_Y\Compl(\Sfun^2f)
    =(\Sproj0\Compl\Sproj1+\Sproj1\Compl\Sproj0)\Compl(\Sfun^2f)
    =\Sproj0\Compl\Sproj1\Compl\Sfun^2f+\Sproj1\Compl\Sproj0\Compl(\Sfun^2f)
    =f\Compl\Sproj0\Compl\Sproj1+f\Compl\Sproj1\Compl\Sproj0$.

    One proves
    $\Sfunadd_X\Compl\Sfunadd_{\Sfun
      X}=\Sfunadd_X\Compl\Sfun\Sfunadd_X$ by showing in the same
    manner that
    $\Sproj0\Compl\Sfunadd_X\Compl\Sfunadd_{\Sfun
      X}=\Sproj0\Compl\Sproj0\Compl\Sproj0
    =\Sproj0\Compl\Sfunadd_X\Compl\Sfun\Sfunadd_X $ and that
    $\Sproj1\Compl\Sfunadd_X\Compl\Sfunadd_{\Sfun
      X}=\Sproj0\Compl\Sproj0\Compl\Sproj1
    +\Sproj0\Compl\Sproj1\Compl\Sproj0
    +\Sproj1\Compl\Sproj0\Compl\Sproj0
    =\Sproj1\Compl\Sfunadd_X\Compl\Sfun\Sfunadd_X$.  The
    commutations involving $\Sfunadd$ and $\Sin0$ are proved in the
    same way. The last equation results from
    % TYPO
    % $\Sflip\Compl\Sproj i\Compl\Sproj j=\Sproj j\Compl\Sproj i$
    $\Sproj i\Compl\Sproj j\Compl\Sflip=\Sproj j\Compl\Sproj i$
\end{proof}
\begin{Example}
  In our coherence space running example, we have
  % TYPO
  % $\Matappa{\Sin0}{(x,u)}=x$ and
  $\Matappa{\Sin0}x=(x,\emptyset)$ and
  $\Matappa{\Sfunadd}{((x,u),(y,v))}=(x,u+y)$; notice indeed that
  since $((x,u),(y,v))\in\Cl{\Sfun^2 E}$ we have $x+u+y+v\in\Cl E$.
\end{Example}

Just as in tangent categories, this monad structure will be crucial
for expressing that the differential (Jacobian) is a linear morphism.

\section{Differentiation in a summable symmetric monoidal category}
\label{sec:sum-moncat} %
Let $\cL$ be a symmetric monoidal
category (SMC), with monoidal product $\ITens$, unit $\Sone$ and
isomorphisms $\Rightu_X\in\cL(\Tens X\Sone,X)$,
$\Leftu_X\in\cL(\Tens\Sone ,X)$,
$\Assoc_{X_0,X_1,X_2}\in\cL(\Tens{\Tensp{X_0}{X_1}}{X_2},
\Tens{X_0}{\Tensp{X_1}{X_2}})$
and $\Sym_{X_0,X_1}\in\cL(\Tens{X_0}{X_1},\Tens{X_1}{X_0})$. Most
often these isos will be kept implicit to simplify the presentation.
Concerning the compatibility of the summability structure with the
monoidal structure our axiom stipulates distributivity.

Assume that $\cL$ is also equipped with a summability structure. We say
that $\cL$ is a summable SMC if the following property holds, which
expresses that the tensor distributes over the (partially defined)
sum.

\begin{Axicond}{\Saxdist}
  If $(f_{00},f_{01})$ is a summable pair of morphisms in
  $\cL(X_0,Y_0)$ and $f_1\in\cL(X_1,Y_1)$ then
  $(\Tens{f_{00}}{f_1},\Tens{f_{01}}{f_1})$ is a summable pair of
  morphisms in $\cL(\Tens{X_0}{X_1},\Tens{Y_0}{Y_1})$, and moreover
\begin{align*}
  \Tens{f_{00}}{f_1}+\Tens{f_{01}}{f_1}=\Tens{(f_{00}+f_{01})}{f_1}
\end{align*}  
\end{Axicond}

As a consequence, using the symmetry of $\ITens$, if $(f_{00},f_{01})$
is summable in $\cL(X_0,Y_0)$ and $(f_{10},f_{11})$ is summable in
$\cL(X_1,Y_1)$, the family
$(\Tens{f_{00}}{f_{10}},\Tens{f_{00}}{f_{11}},\Tens{f_{01}}{f_{10}},\Tens{f_{01}}{f_{11}})$
is summable in $\cL(\Tens{X_0}{X_1},\Tens{Y_0}{Y_1})$ and we have
\begin{align*}
  \Tens{(f_{00}+f_{01})}{(f_{10}+f_{11})}
  =\Tens{f_{00}}{f_{10}}+\Tens{f_{00}}{f_{11}}+\Tens{f_{01}}{f_{10}}+\Tens{f_{01}}{f_{11}}\,.
\end{align*}

We can define a natural transformation
$\Sstr_{X_0,X_1}\in\cL(\Tens{X_0}{\Sfun X_1},\Sfun\Tensp{X_0}{X_1})$
by setting $\Sstr_{X_0,X_1}=\Stuple{\Tens{X_0}{\Sproj0},\Tens{X_0}{\Sproj1}}$
which is well defined by \Saxdist. We use
$\Sstrs_{X_0,X_1}\in\cL(\Tens{\Sfun X_0}{X_1},\Sfun\Tensp{X_0}{X_1})$
for the natural transformation defined from $\Sstr$ using the symmetry
isomorphism of the SMC, that is %
$\Sstrs_{X_0,X_1}=\Sstr_{X_1,X_0}\Compl\Sym
=\Stuple{\Tens{\Sproj0}{X_1},\Tens{\Sproj1}{X_1}}
\in\cL(\Tens{\Sfun X_0}{X_1},\Sfun\Tensp{X_0}{X_1})$.
% as the following composition of natural transformation
% \[
%   \begin{tikzcd}
%     \Tens{\Sfun X_0}{X_1}\ar[r,"\Tens{\Sfun X_0}{\Sin0}"]
%     &[1em]\Tens{\Sfun X_0}{\Sfun X_1}\ar[r,"\Smont"]
%     &\Sfun\Tensp{X_0}{X_1}
%   \end{tikzcd}
% \]
% Notice that this morphism is completely characterized by the two
% following commutations
% \begin{align*}
%   \Sproj0\Compl\Sstr=\Tens{\Sproj0}{X_0}
%   \text{ and }
%   \Sproj1\Compl\Sstr=\Tens{\Sproj1}{X_0}  
% \end{align*}
% the second equation being due to the fact that $\Sproj1\Compl\Sin0=0$.

\begin{lemma}\label{lemma:Sstr-sum}
  % TYPO
  % $\Ssum\Compl\Sstr_{\Tens{X_0}{X_1}}=\Tens{X_0}{\Ssum_{X_1}}$.
  $\Ssum\Compl\Sstr_{{X_0},{X_1}}=\Tens{X_0}{\Ssum_{X_1}}$.
\end{lemma}
\begin{proof}
  We have
  % TYPO
  % $\Ssum\Compl\Sstr_{\Tens{X_0}{X_1}}=\Tens{X_0}{\Sproj
  $\Ssum\Compl\Sstr_{{X_0},{X_1}}=\Tens{X_0}{\Sproj
    0}+\Tens{X_0}{\Sproj 1}=\Tens{X_0}{(\Sproj0+\Sproj1)}$ by
  \Saxdist{} and we have $\Sproj0+\Sproj1=\Ssum_{X_1}$.
\end{proof}

\begin{theorem}\label{th:Sstr-com-strength}
  The natural transformation $\Sstr$ is a strength for the monad
  $(\Sfun,\Sin0,\Sfunadd)$ and equipped with $\Sstr$ this monad is
  commutative.
\end{theorem}
% From the appendix
%
\begin{proof}
  The fact that %
  $\Sstr$ is a strength means that the following two diagrams commute:
  \[
    \begin{tikzcd}
      \Tens{X_0}{X_1}\ar[d,swap,"\Tens{X_0}{\Sin0}"]
      \ar[dr,"\Sin0"]&\\
      \Tens{X_0}{\Sfun X_1}\ar[r,"\Sstr"]
      &\Sfun\Tensp{X_0}{X_1}
    \end{tikzcd}
    \quad\quad
    \begin{tikzcd}
      \Tens{X_0}{\Sfun^2X_1}
      \ar[r,"\Sstr"]\ar[d,"\Tens{X_0}{\Sfunadd}"]
      &\Sfun\Tensp{X_0}{\Sfun X_1}\ar[r,"\Sfun\Sstr"]
      &\Sfun^2\Tensp{X_0}{X_1}\ar[d,"\Sfunadd"]\\
      \Tensp{X_0}{\Sfun X_1}\ar[rr,"\Sstr"]
      &&\Sfun\Tensp{X_0}{X_1}
    \end{tikzcd}
  \]
  Let us prove for instance the second one. We have
  \begin{align*}
    \Sfunadd\Compl(\Sfun\Sstr)\Compl\Sstr
    &=\Stuple{\Sproj0\Compl\Sproj0,\Sproj1\Compl\Sproj0+\Sproj0\Compl\Sproj1}
    \Compl\Stuple{\Sstr\Compl\Sproj0,\Sstr\Compl\Sproj1}\Compl\Sstr\\
    &=\Stuple{\Sproj0\Compl\Sstr\Compl\Sproj0,\Sproj1\Compl\Sstr\Sproj0
      +\Sproj0\Sstr\Sproj1}\Compl\Sstr\\
    &=\Stuple{\Tensp{X_0}{\Sproj0}\Compl\Sproj0\Compl\Sstr,
      \Tensp{X_0}{\Sproj1}\Compl\Sproj0\Compl\Sstr
      +\Tensp{X_0}{\Sproj0}\Compl\Sproj1\Compl\Sstr}\\
    &=\Stuple{\Tensp{X_0}{\Sproj0}\Compl\Tensp{X_0}{\Sproj0},
      \Tensp{X_0}{\Sproj1}\Compl\Tensp{X_0}{\Sproj0}
      +\Tensp{X_0}{\Sproj0}\Compl\Tensp{X_0}{\Sproj1}}
      =\Tens{X_0}{\Sfunadd}\,.
  \end{align*}
  The fact that $(\Sfun,\Sin0,\Sfunadd,\Sstr)$ is a commutative monad
  means that, moreover, the following diagram commutes:
  \[
    \begin{tikzcd}
      \Tens{\Sfun X_0}{\Sfun X_1}\ar[r,"\Sstr_{\Sfun X_0,X_1}"]
      \ar[d,swap,"\Sstrs_{X_0,\Sfun X_1}"]
      &[1em] \Sfun\Tensp{\Sfun X_0}{X_1}\ar[r,"\Sfun\Sstrs_{X_0,X_1}"]
      &[1em] \Sfun^2\Tensp{X_0}{X_1}\ar[d,"\Sfunadd"]\\
      \Sfun\Tensp{X_0}{\Sfun X_1}\ar[r,"\Sfun\Sstr_{X_0,X_1}"]
      & \Sfun^2\Tensp{X_0}{X_1}\ar[r,"\Sfunadd"]
      & \Sfun\Tensp{X_0}{X_1}
    \end{tikzcd}
  \]
  which results from a stronger property, namely that the following
  diagram commutes
  \[
    \begin{tikzcd}
      \Tens{\Sfun X_0}{\Sfun X_1}\ar[rr,"\Sstr_{\Sfun X_0,X_1}"]
      \ar[d,swap,"\Sstrs_{X_0,\Sfun X_1}"]
      &[1em]&[1em] \Sfun\Tensp{\Sfun X_0}{X_1}\ar[d,"\Sfun\Sstrs_{X_0,X_1}"]\\
      \Sfun\Tensp{X_0}{\Sfun X_1}\ar[r,"\Sfun\Sstr_{X_0,X_1}"]
      & \Sfun^2\Tensp{X_0}{X_1}\ar[r,"\Sflip_{\Tens{X_0}{X_1}}"]
      & \Sfun^2\Tensp{X_0}{X_1}
    \end{tikzcd}
  \]
  and from Theorem~\ref{th:sfun-monade}. The last commutation is
  proved as follows:
  \begin{align*}
    \Sproj i\Compl\Sproj j\Compl(\Sfun\Sstrs_{X_0,X_1})
    \Compl\Sstr_{\Sfun X_0,X_1}
    &= \Sproj i
      \Compl\Sstrs_{X_0,X_1}
      \Compl\Sproj j
      \Compl\Sstr_{\Sfun X_0,X_1}\\
    &=\Tensp{\Sproj i}{X_1}
      \Compl\Tensp{\Sfun X_0}{\Sproj j}\\
    &=\Tens{\Sproj i}{\Sproj j}\\
    \Sproj i
    \Compl\Sproj j
    \Compl\Sflip
    \Compl(\Sfun\Sstr_{X_0,X_1})
    \Compl\Sstrs_{X_0,\Sfun X_1}
    &=\Sproj j\Compl\Sproj i
      \Compl(\Sfun\Sstr_{X_0,X_1})
      \Compl\Sstrs_{X_0,\Sfun X_1}\\
    &=\Sproj j
      \Compl\Sstr_{X_0,X_1}
      \Compl\Sproj i
      \Compl\Sstrs_{X_0,\Sfun X_1}\\
    &=\Tensp{X_0}{\Sproj j}\Compl\Tensp{\Sproj i}{\Sfun X_1}\\
    &=\Tens{\Sproj i}{\Sproj j}\,.
  \end{align*}
\end{proof}

We set
$\Smont_{X_0,X_1}=\Sfunadd\Compl(\Sfun\Sstrs_{X_0,X_1})\Compl\Sstr_{\Sfun
  X_0,X_1}=\Sfunadd\Compl(\Sfun\Sstr_{X_0,X_1})\Compl\Sstrs_{X_0,\Sfun
  X_1}=\Stuple{\Tens{\Sproj0}{\Sproj0},
  \Tens{\Sproj1}{\Sproj0}+\Tens{\Sproj0}{\Sproj1}}$; it is well known
that in such a commutative monad situation, the associated tuple
$(\Sfun,\Sin0,\Sfunadd,\Smont)$ is a symmetric monoidal monad on the
SMC $\cL$.

\begin{definition}
  When the summability structure of the SMC $\cL$ satisfies \Saxdist{}
  we say that $\cL$ is a \emph{summable SMC}.
\end{definition}

\subsection{Differential structure}\label{sec:diff-struct-res-cat}

We say that a resource category $\cL$ (see
Section~\ref{sec:resource-cat}) is summable if it is summable as an
SMC and satisfies the following additional condition of compatibility
with the cartesian product.

\begin{Axicond}{\Saxprod}
  The functor $\Sfun$ preserves all finite cartesian products.  In
  other words $0\in\cL(\Sfun\Top,\Top)$ and
  $\Tuple{\Sfun{\Proj0},\Sfun{\Proj1}}\in\cL(\Sfun(\With{X_0}{X_1}),
  \With{\Sfun{X_0}}{\Sfun{X_1}})$ are isos.
\end{Axicond}

A \emph{differential structure} on a summable resource category $\cL$
consists of a natural transformation
$\Sdiff_X\in\cL(\Excl{\Sfun X},\Sfun\Excl X)$ which satisfies the
following conditions.

\begin{Axicond}{\Daxlocal}
  \(
    \begin{tikzcd}
      \Excl{\Sfun X}\ar[r,"\Sdiff_X"]\ar[dr,swap,"\Excl{\Sproj0}"]
      &\Sfun\Excl X\ar[d,"\Sproj0"]\\
      &\Excl X
    \end{tikzcd}
  \)
\end{Axicond}

\begin{Axicond}{\Daxlin}
  \(
    \begin{tikzcd}
      \Excl X\ar[d,swap,"\Excl{\Sin 0}"]\ar[rd,"\Sin0"]
      &\\
      \Excl{\Sfun X}\ar[r,"\Sdiff_X"]
      &\Sfun\Excl X
    \end{tikzcd}
    \quad\quad
    \begin{tikzcd}
      \Excl{\Sfun^2X}\ar[r,"\Sdiff_{\Sfun X}"]\ar[d,swap,"\Excl{\Sfunadd}"]
      &\Sfun{\Excl{\Sfun X}}\ar[r,"\Sfun\Sdiff_X"]
      &\Sfun^2\Excl X\ar[d,"\Sfunadd"]\\
      \Excl{\Sfun X}\ar[rr,"\Sdiff_X"]
      &&
      \Sfun\Excl X
    \end{tikzcd}
  \)
\end{Axicond}
This first condition allows to extend the functor $\Excl\_$ to the
Kleisli category $\cL_{\Sfun}$ of the monad $\Sfun$. In this Kleisli
category, a morphism $X\to Y$ can be seen as a pair $(f_0,f_1)$ of two
summable morphisms in $\cL(X,Y)$, and composition is defined by
$g\Comp f=(g_0\Compl f_0,g_1\Compl f_0+g_0\Compl f_1)$, a definition
which is very reminiscent of the multiplication of dual numbers.

\begin{Axicond}{\Daxchain}
  \(
    \begin{tikzcd}
      \Excl{\Sfun X}\ar[r,"\Sdiff_X"]\ar[rd,swap,"\Der{\Sfun X}"]
      &\Sfun\Excl X\ar[d,"\Sfun\Der X"]\\
      &\Sfun X
    \end{tikzcd}
    \quad
    \begin{tikzcd}
      \Excl{\Sfun X}\ar[rr,"\Sdiff_X"]\ar[d,swap,"\Digg{\Sfun X}"] &
      &\Sfun\Excl X\ar[d,"\Sfun\Digg X"]\\
      \Excll{\Sfun X}\ar[r,"\Excl{\Sdiff_X}"] &\Excl{\Sfun\Excl
        X}\ar[r,"\Sdiff_{\Excl X}"] &\Sfun\Excll X
    \end{tikzcd}
  \)
\end{Axicond}
This second condition allows to extend the functor $\Sfun$ to the
Kleisli category $\Kl\cL$. We obtain in that way the functor
$\Sdfun:\Kl\cL\to\Kl\cL$ defined as follows: on objects, we set
$\Sdfun X=\Sfun X$. Next, given $f\in\Kl\cL(X,Y)=\cL(\Excl X,Y)$, the
morphism
$\Sdfun f\in\Kl\cL(\Sfun X,\Sfun Y)=\cL(\Excl{\Sfun X},\Sfun Y)$ is
defined by $\Sdfun f=(\Sfun f)\Compl\Sdiff_X$. The purpose of the two
commutations is precisely to make this operation functorial and this
functoriality is a categorical version of the chain rule of calculus,
exactly as in tangent categories since, as we shall see, this functor
$\Sdfun$ essentially computes the derivative of $f$.

\begin{remark}
  It is very likely that the natural transformation $\Sdiff_X$ can be
  seen as one of the six kinds distributive law between the monad
  $\Sfun$ and the comonad $\Excl\_$ described
  in~\cite{PowerWatanabe02}, Section~8. % If this is actually the case
  % there is no doubt that the two liftings above to Kleisli categories
  % can be adapted from the results presented in that paper.
\end{remark}

\begin{Axicond}{\Daxwith}
\(
  \begin{tikzcd}
    \Excl{\Sfun\Top}\ar[rr,"\Sdiff_\Top"]\ar[d,swap,"\Excl 0"]
    &&\Sfun{\Excl\Top}\ar[d,"\Sfun\Inv{(\Seelyz)}"]\\
    \Excl\Top\ar[r,"\Inv{(\Seelyz)}"]&\Sone\ar[r,"\Sin0"]&\Sfun\Sone
  \end{tikzcd}
  \quad
  \begin{tikzcd}
    \Excl{\Sfun{\Withp{X_0}{X_1}}}\ar[r,"\Sdiff_{\With{X_0}{X_1}}"]
    \ar[d,swap,"\Excl{\Tuple{\Sfun{\Proj0},\Sfun{\Proj1}}}"]
    &\Sfun\Excl{\Withp{X_0}{X_1}}\ar[r,"\Sfun\Inv{(\Seelyt)}"]
    &[1.2em]\Sfun\Tensp{\Excl{X_0}}{\Excl{X_1}}\\
    \Excl{\Withp{\Sfun X_0}{\Sfun X_1}}\ar[r,"\Inv{(\Seelyt)}"]
    &\Tens{\Excl{\Sfun X_0}}{\Excl{\Sfun X_1}}
    \ar[r,"\Tens{\Sdiff_{X_0}}{\Sdiff_{X_1}}"]
    &\Tens{\Sfun{\Excl{X_0}}}{\Sfun{\Excl{X_1}}}\ar[u,swap,"\Smont_{\Excl{X_0},\Excl{X_1}}"]
  \end{tikzcd}
\)  
\end{Axicond}

% The first of these two last diagrams expresses that the derivative of
% a constant function is $0$ and the second one is a version of the
% Leibniz Law of calculus.
\begin{theorem}[Leibniz rule]\label{th:leibniz-cat}
  If \Daxwith{} holds then the following diagrams commute.
  \[
    \begin{tikzcd}
      \Excl{\Sfun X}\ar[r,"\Sdiff_X"]\ar[d,swap,"\Weak{\Sfun X}"]
      &\Sfun\Excl X\ar[d,"\Sfun\Weak X"]\\
      \Sone\ar[r,"\Sin0"] & \Sfun\Sone
    \end{tikzcd}
    \quad\quad
    \begin{tikzcd}
      \Excl{\Sfun X}\ar[rr,"\Sdiff_X"]\ar[d,swap,"\Contr{\Excl{\Sfun X}}"]
      &&\Sfun\Excl X\ar[d,"\Sfun\Contr X"]\\
      \Tens{\Excl{\Sfun X}}{\Excl{\Sfun X}}\ar[r,"\Tens{\Sdiff_X}{\Sdiff_X}"]
      &\Tens{\Sfun\Excl X}{\Sfun\Excl X}\ar[r,"\Smont_{\Excl X,\Excl X}"]
      &\Sfun\Tensp{\Excl X}{\Excl X}
    \end{tikzcd}
  \]
\end{theorem}
\begin{proof}
  This is an easy consequence of the naturality of $\Sdiff$ and of the
  definition of $\Weak X$ and $\Contr X$ which is based on the
  cartesian products and on the Seely isomorphisms.
\end{proof}

\begin{Axicond}{\Daxschwarz}
  \(
    \begin{tikzcd}
      \Excl{\Sfun^2X}\ar[r,"\Sdiff_{\Sfun X}"]\ar[d,swap,"\Excl\Sflip"]
      &\Sfun\Excl{\Sfun X}\ar[r,"\Sfun\Sdiff_X"]
      &\Sfun^2\Excl X\ar[d,"\Sflip"]\\
      \Excl{\Sfun^2X}\ar[r,"\Sdiff_{\Sfun X}"]
      &\Sfun\Excl{\Sfun X}\ar[r,"\Sfun\Sdiff_X"]
      &\Sfun^2\Excl X
    \end{tikzcd}
    \)
\end{Axicond}
This diagram, involves the canonical flip $\Sflip$ introduced before
the statement of \Saxass{} and expresses a kind of commutativity of
the second derivative.

\begin{definition}\label{def:Sfun-diff}
  A \emph{differentiation} in a summable resource category $\cL$ is a
  natural transformation $\Sdiff_X\in\cL(\Excl{\Sfun X},\Sfun\Excl X)$
  which satisfies \Daxlocal, \Daxlin, \Daxchain, \Daxwith{} and
  \Daxschwarz. A summable resource category given together with a
  differentiation is a \emph{differential summable resource category}.
\end{definition}

\subsection{Derivatives and partial derivatives in the Kleisli category}%
\label{sec:Kleisli-derivatives}
The Kleisli category $\Kl\cL$ of the comonad
$(\Excl{},\Der{},\Digg{})$ is well known to be cartesian. In general
it is not a differential cartesian category in the sense
of~\cite{Alvarez-PicalloLemay20} because it is not required to be
additive\footnote{We postpone the precise axiomatization of this kind
  of partially additive differential category to further work. Of
  course it will be based on the concept of summability
  structure.}. Our running example of coherence spaces is an example
of such a category which is not a differential category.
% En fait c'est plutôt proche d'une catégorie cartésienne tangente.

There is an inclusion functor $\Kllin:\cL\to\Kl\cL$ which maps $X$ to
$X$ and $f\in\cL(X,Y)$ to $f\Compl\Der X\in\Kl\cL(X,Y)$, it is
faithful but not full in general and allows to see any morphism of
$\cL$ as a ``linear morphism'' of $\Kl\cL$.

We have already mentioned the functor $\Sdfun:\Kl\cL\to\Kl\cL$,
remember that $\Sdfun X=\Sfun X$ and
$\Sdfun f=(\Sfun f)\Compl\Sdiff_X$ when $f\in\cL(X,Y)$. Then we have
$\Sdfun\Comp\Kllin=\Kllin\Comp\Sfun$ which allows to extend simply the
monad structure of $\Sfun$ to $\Sdfun$ by setting
$\Sdfunit_X=\Kllin{\Sin0}\in\Kl\cL(X,\Sdfun X)$ and
$\Sdfmult_X=\Kllin{\Sfunadd}\in\Kl\cL(\Sdfun^2 X,\Sdfun X)$.

\begin{theorem}\label{th:Sdfunst-monad}
  The morphisms $\Sdfunit_X\in\Kl\cL(X,\Sdfun X)$ and %
  $\Sdfmult_X\in\Kl\cL(\Sdfun^2X,\Sdfun X)$ are natural and turn the
  functor $\Sdfun$ into a monad on $\Kl\cL$.
\end{theorem}
\begin{proof}
  The only non obvious property is naturality, the monadic diagram
  commutations resulting from those of $(\Sfun,\Sin,\Ssum)$ on $\cL$
  and of the functoriality of $\Kllin$. The proof can certainly be
  adapted from~\cite{PowerWatanabe02}, we provide it for convenience.
  Let $f\in\Kl\cL(X,Y)$, that is $f\in\cL(\Excl X,Y)$. We must first
  prove that $\Sdfun f\Comp\Sdfunit_X=\Sdfunit_Y\Comp f$. We have
  \begin{align*}
    \Sdfun f\Comp\Sdfunit_X
    &=(\Sfun f)\Compl\Sdiff_X\Compl\Excl{\Sdfunit_X}\Compl\Digg X\\
    &=(\Sfun f)
      \Compl\Sdiff_X
      \Compl\Excl{\Sin0}
      \Compl\Excl{\Der X}
      \Compl\Digg X\\
    &=(\Sfun f)
      \Compl\Sdiff_X
      \Compl\Excl{\Sin0}\\
    &=(\Sfun f)\Compl\Sin0\text{\quad by \Daxlin}\\
    &=\Sin0\Compl f\text{\quad by naturality}\\
    &=\Sdfunit_Y\Comp f\,.
  \end{align*}
  Similarly
  \begin{align*}
    \Sdfun f\Comp\Sdfmult_X
    &=(\Sfun f)
      \Compl\Sdiff_X
      \Compl\Excl{\Sdfmult_X}
      \Compl\Digg X\\
    &=(\Sfun f)
      \Compl\Sdiff_X
      \Compl\Excl{\Sfunadd_X}
      \Compl\Excl{\Der X}
      \Compl\Digg X\\
    &=(\Sfun f)
      \Compl\Sfunadd_{\Excl X}
      \Compl(\Sfun\Sdiff_X)
      \Compl\Sdiff_{\Sfun X}\text{\quad by \Daxlin}\\
    &=\Sfunadd_Y
      \Compl(\Sfun\Sfun f)
      \Compl(\Sfun\Sdiff_X)
      \Compl\Sdiff_{\Sfun X}\text{\quad by naturality}\\
    &=\Sdfmult_Y\Comp\Sdfun^2f
  \end{align*}
\end{proof}

Since $\Sfun$ preserves cartesian products, we can equip easily this
monad %
$(\Sdfun,\Sdfunit,\Sdfmult)$ on $\Kl\cL$ with a commutative strength %
$\Sdfstr_{X_0,X_1}\in\Kl\cL(\With{X_0}{\Sdfun{X_1}},\Sdfun{\Withp{X_0}{X_1}})$
which is the following composition in $\cL$
\[
  \begin{tikzcd}
    \Excl{\Withp{X_0}{\Sfun X_1}}\ar[r,"\Der{}"]
    &\With{X_0}{\Sfun X_1}\ar[r,"\With{\Sin0}{\Sfun X_1}"]
    &[1em]\With{\Sfun X_0}{\Sfun X_1}\ar[r,"\eta"]
    &[-1em]\Sfun \Withp{X_0}{X_1}
  \end{tikzcd}
\]
where $\eta=\Inv{\Tuple{\Sfun{\Proj0},\Sfun{\Proj1}}}$ is the
canonical iso of \Saxprod. It is possible to prove the following
commutation in $\cL$, relating the strength of $\Sfun$ (wrt.~$\ITens$)
with the strength of $\Sdfun$ (wrt.~$\IWith$) through the Seely
isomorphisms
\[
  \begin{tikzcd}
    \Excl{\Withp{X_0}{\Sfun X_1}}\ar[r,"\Der{}"]
    &\With{X_0}{\Sfun X_1}\ar[r,"\With{\Sin0}{\Sfun X_1}"]
    &\With{\Sfun X_0}{\Sfun X_1}\ar[r,"\eta"]
    &\Sfun{\Withp{X_0}{X_1}}\\
    \Tens{\Excl{X_0}}{\Excl{\Sfun X_1}}\ar[u,"\Seelyt"]
    \ar[r,"\Tens{\Excl{X_0}}{\Sdiff_{X_1}}"]
    &\Tens{\Excl{X_0}}{\Sfun \Excl{X_1}}\ar[r,"\Sstr_{\Excl{X_0},\Excl{X_1}}"]
    &\Sfun{\Tensp{\Excl{X_0}}{\Excl{X_1}}}\ar[r,"\Sfun\Seelyt"]
    &\Sfun{\Excl{\Withp{X_0}{X_1}}}\ar[u,swap,"\Sfun\Der{}"]
  \end{tikzcd}
\]

Given $f\in\Kl\cL(\With{X_0}{X_1},Y)$, we can define the partial
derivatives $\Sdfun_0f\in\Kl\cL(\With{\Sdfun X_0}{X_1},\Sdfun Y)$ and
$\Sdfun_1f\in\Kl\cL(\With{X_0}{\Sdfun X_1},\Sdfun Y)$ as
$\Sdfun f\Comp\Sdfstrs$ and $\Sdfun f\Comp\Sdfstr$ where we use
$\Sdfstrs$ for the strength
$\With{\Sdfun X_0}{X_1}\to\Sdfun\Withp{X_0}{X_1}$ defined from
$\Sdfstr$ using the symmetry of $\IWith$.

\subsection{Deciphering the diagrams}
After this rather terse list list of categorical axioms, it is fair to
provide the reader with intuitions about their intuitive meaning; this
is the purpose of this section.

One should think of the objects of $\cL$ as partial commutative
monoids (with additional structures depending on the considered
category), and $\Sfun X$ as the object of pairs $(x,u)$ of elements
$x,u\in X$ such that $x+u\in X$ is defined. The morphisms in $\cL$ are
linear in the sense that they preserve $0$ and this partially defined
sums whereas the morphisms of $\Kl\cL$ should be thought of as
functions which are not linear but admit a ``derivative''. More
precisely $f\in\Kl\cL(X,Y)$ can be seen as a function $X\to Y$ and,
given $(x,u)\in\Sfun X$ we have
\begin{align*}
  \Sdfun f(x,u)=(f(x),\Derd{f(x)}xu)\in\Sfun Y\,,
\end{align*}
where $\Derd{f(x)}xu$ is just a notation for the second component of
the pair $\Sdfun f(x,u)$ which, by construction, is such that the sum
$f(x)+\Derd{f(x)}xu$ is a well defined element of $Y$. Now we assume
that this derivative $\Derd{f(x)}xu$ obeys the standard rules of
differential calculus and we shall see that the above axioms about
$\Sdiff$ correspond to these rules.

\begin{remark}
  The equations we are using in this section as intuitive
  justifications for the diagrams of
  Section~\ref{sec:diff-struct-res-cat} refer to the standard laws and
  properties of the differential calculus that we assume the reader to
  be acquainted with. They do hold exactly as written here in the
  model $\PCOH$ where derivatives are computed exactly as in Calculus
  as we will show in a forthcoming paper.
\end{remark}

\begin{remark}
  We use the well established notation $\Derd{f(x)}xu$ which must be
  understood properly: in particular the expression $\Derd{f(x)}xu$ is
  a function of $x$ (the point where the derivative is computed) and
  of $u$ (the linear parameter of the derivative). When required we
  use $\Derdev{f(x)}x{x_0}u$ for the evaluation of this derivative at
  point $x_0\in X$.
\end{remark}

% With this intuition in mind we can interpret the axioms above, keeping
% in mind that $\Sdfun f$ is defined using $\Sdiff$ by
% $\Sdfun f=(\Sfun f)\Compl\Sdiff_X$ when $f\in\cL(\Excl X,Y)$.

\begin{itemize}
\item \Daxlocal{} %
  means that the first component of $\Sdfun f(x,u)$ is $f(x)$,
  justifying our intuitive notation
  \begin{align*}
    \Sdfun f(x,u)=(f(x),\Derd{f(x)}xu)\in\Sfun Y\,.
  \end{align*}
\item The first diagram of \Daxchain{} means that if $f\in\Kl\cL(X,Y)$
  is linear%
  \footnote{This notion of linearity implies commutation with the
    partial algebraic structure introduced by $\Sfun$ as shown by
    Lemma~\ref{lemma:compl-summable}.} %
  in the sense that there is $g\in\cL(X,Y)$ such that
  $f=g\Compl\Der X$, then $\Derd{f(x)}xu=f(u)$. Notice that it
  prevents differentiation from being trivial by setting
  $\Derd{f(x)}xu=0$ for all $f$ and all $x,u$. Consider now
  $f\in\Kl\cL(X,Y)$ and $g\in\Kl\cL(\Excl Y,Z)$; the second diagram
  means that $\Sdfun(g\Comp f)=\Sdfun g\Comp\Sdfun f$, which amounts
  to
  \begin{align*}
    \Derd{g(f(x))}xu=\Derdev{g(y)}{y}{f(x)}{(\Derd{f(x)}{x}{u})}
  \end{align*}
  which is exactly the chain rule.
\item The ``second derivative''
  $\Sdfun^2f\in\Kl\cL(\Sfun^2X,\Sfun^2Y)$ of $f\in\Kl\cL(X,Y)$ is
  $(\Sfun^2f)\Compl(\Sfun\Sdiff_X)\Compl\Sdiff_{\Sfun X}$. Remember
  that $\Sdfun f(x,u)=(f(x),\Derd{f(x)}xu)$, therefore applying the
  standard rules of differential calculus we have
  \begin{align*}
    \Sdfun^2f((x,u),(x',u'))
    &=(\Sdfun f(x,u),\Derd{\Sdfun f(x,u)}{(x,u)}{(x',u')})\\
    &=((f(x),\Derd{f(x)}{x}{u}),\Derp{(f(x),
      \Derd{f(x)}{x}{u})}{x}{x'}+\Derp{(f(x),
      \Derd{f(x)}{x}{u})}{u}{u'})\\
    &=((f(x),\Derd{f(x)}{x}{u}),(\Derd{f(x)}{x}{x'},
      \Derdn2{f(x)}{x}{(u,x')}+\Derd{f(x)}{x}{u'}))
  \end{align*}
  where we have used the fact that $f(x)$ does not depend on $u$ and
  that $\Derd{f(x)}{x}{u}$ is linear in $u$). We have used \Daxlin{}
  to prove Theorem~\ref{th:Sdfunst-monad} whose main content is the
  naturality of $\Sdfunit$ and $\Sdfmult$. This second naturality
  means that %
  $\Sdfun f\Comp\Sdfmult_X=\Sdfmult_Y\Comp\Sdfun^2f$, that is, by the
  computation above %
  $\Derd{f(x)}{x}{(u+x')}=\Derd{f(x)}{x}{u}+\Derd{f(x)}{x}{x'}$ %
  since, intuitively, %
  $\Sdfmult_X((x,u),(x',u'))=(x,u+x')$. %
  Similarly the naturality of $\Sdfunit$ means that %
  $\Derd{f(x)}{x}{0}=0$. So the condition \Daxlin{} means that the
  derivative is a function which is linear with respect to its second
  parameter.
\item We have assumed that $\cL$ is cartesian and hence $\Kl\cL$ is
  also cartesian. Intuitively $\With{X_0}{X_1}$ is the space of pairs
  $(x_0,x_1)$ with $x_i\in X_i$, and our assumption \Saxprod{} means
  that $\Sfun\Withp{X_0}{X_1}$ is the space of pairs
  $((x_0,x_1),(u_0,u_1))$ such that $(x_i,u_i)\in\Sfun X_i$, and the
  sum of such a pair is $(x_0+u_0,x_1+u_1)\in \With{X_0}{X_1}$. Then,
  given $f\in\Kl\cL(\With{X_0}{X_1},Y)$ the second diagram of
  \Daxwith{} means that
  \begin{align*}
    \Derd{f(x_0,x_1)}{(x_0,x_1)}{(u_0,u_1)}
    =\Derp{f(x_0,x_1)}{x_0}{u_0}+\Derp{f(x_0,x_1)}{x_1}{u_1}
  \end{align*}
  which can be seen by the following computation of
  $\Sproj1\Compl\Sdfun f$ using that diagram
  \begin{align*}
    \Sproj 1\Compl\Sdfun f
    &=\Sproj 1\Compl(\Sfun f)\Compl\Sdiff_{\With{X_0}{X_1}}\\
    &=f\Compl\Seelyt\Compl\Sproj1
      \Compl\Smont_{\Excl{X_0},\Excl{X_1}}
      \Compl\Tensp{\Sdiff_{X_0}}{\Sdiff_{X_1}}
      \Compl\Inv{(\Seelyt)}\Compl\Excl{\Tuple{\Sfun\Proj0,\Sfun\Proj1}}\\
    &=f\Compl\Seelyt\Compl(\Tens{\Sproj1}{\Sproj0}+\Tens{\Sproj0}{\Sproj1})
      \Compl\Tensp{\Sdiff_{X_0}}{\Sdiff_{X_1}}
      \Compl\Inv{(\Seelyt)}\Compl\Excl{\Tuple{\Sfun\Proj0,\Sfun\Proj1}}\\
    &=f\Compl\Seelyt\Compl(\Tens{\Sproj1\Compl\Sdiff_{X_0}}{\Excl{\Sproj0}})
      \Compl\Inv{(\Seelyt)}\Compl\Excl{\Tuple{\Sfun\Proj0,\Sfun\Proj1}}
      +f\Compl\Seelyt\Compl(\Tens{\Excl{\Sproj0}}{\Sproj1\Compl\Sdiff_{X_0}})
      \Compl\Inv{(\Seelyt)}\Compl\Excl{\Tuple{\Sfun\Proj0,\Sfun\Proj1}}\\
    &=\Sproj1\Compl(\Sdfun_0f)
      \Compl\Excl{\Tuple{\Proj0\Compl\Sproj0,\Sfun\Proj1}}
      +\Sproj1\Compl(\Sdfun_1f)
      \Compl\Excl{\Tuple{\Sfun\Proj0,\Proj1\Compl\Sproj0}}
  \end{align*}
  the two components of these sums corresponding to the two partial
  derivatives, see Section~\ref{sec:Kleisli-derivatives}.

  Then Theorem~\ref{th:leibniz-cat} means that
  $\Derd{f(x,x)}xu
  =\Diffpev{f(x_0,x_1)}{x_0}{x,x}u+\Diffpev{f(x_0,x_1)}{x_1}{x,x}u$
  which is the essence of the Leibniz rule of Calculus.
\item The object $\Sfun^2 X$ consists of pairs $((x,u),(x',u'))$ such
  that $x$, $u$, $x'$ and $u'$ are globally summable. Then
  $\Sflip\in\cL(\Sfun^2X,\Sfun^2X)$ maps $((x,u),(x',u'))$ to
  $((x,x'),(u,u'))$.  Therefore, using the same computation of
  $\Sdfun^2f((x,u),(x',u'))$ as in the case of \Daxlin{}, we see that
  \Daxschwarz{} expresses that
  $\Derdn2{f(x)}{x}{(u,x')}=\Derdn2{f(x)}{x}{(x',u)}$ (upon taking
  $u'=0$). So this diagram means that the second derivative
  (aka.~Hessian) is a symmetric bilinear function, a property of
  sufficiently regular differentiable functions often refereed to as
  Schwarz Theorem.
\end{itemize}

\subsection{A differentiation in coherence spaces} %
\label{sec:coh-diff}
Now we exhibit such a differentiation in $\COH$.
% Again, we consider this possibility as one of
% the main contributions of this paper.
We define $\Excl E$ as follows: $\Web{\Excl E}$ is the set
of finite multisets\footnote{There is also a definition using finite
  sets instead of finite multisets, and this is the one considered by
  Girard in~\cite{Girard87}, but it does not seem to be compatible with
  differentiation, see Remark~\ref{rk:diff-coh-multi}.} $m$ of
elements of $\Web E$ such that $\Supp m\in\Cl E$ (such an $m$ is
called a finite multiclique). Given $m_0,m_1\in\Web{\Excl E}$, we have
$\Coh{\Excl E}{m_0}{m_1}$ if $m_0+ m_1\in\Web{\Excl E}$. This
operation is a functor $\COH\to\COH$: given $s\in\COH(E,F)$ one sets
\begin{align*}
  \Excl s=
  \{(\Mset{\List a1n},\Mset{\List b1n})\St
  n\in\Nat,\ (a_i,b_i)\in s
  \text{ for }i=1,\dots,n\text{ and }\Mset{\List a1n}\in\Web{\Excl E}\}
\end{align*}
which actually belongs to $\Cl{\Limpl{\Excl E}{\Excl F}}$ because
$s\in\Cl{\Limpl EF}$.  The comonad structure of this functor and the
associated commutative comonoid structure are given by
\begin{itemize}
\item $\Der E=\{(\Mset a,a)\St a\in\Web E\}$
\item
  $\Digg E=\{(m,\Mset{m_1,\dots,m_n})\in\Web{\Limpl{\Excl E}{\Excll
      E}}\St m=m_1+\cdots+ m_n\}$
\item $\Weak E=\{(\Emptymset,\Sonelem)\}$
\item and
  $\Contr E=\{(m,(m_1,m_2))\in\Web{\Limpl{\Excl E}{\Tensp{\Excl
        E}{\Excl E}}}\St m=m_1+ m_2\}$.
\end{itemize}
Composition in $\Kl\COH$ can be described directly as follows: let %
$s\in\Cl{\Limpl{\Excl E}{F}}$ and %
$t\in\Cl{\Limpl{\Excl F}{G}}$, then
$t\Comp s\in\cL{\Limplp{\Excl E}{G}}$ is %
$\{(m_1+\cdots+m_n,c)\St\exists \List b1n\in\Web F\ (\Mset{\List
  b1n},c)\in t\text{, }(m_i,b_i)\in s\text{ for }i=1,\dots n \text{
  and }m_1+\cdots+m_n\in\Web{\Excl E}\}$. %
A morphism %
$s\in\Kl\COH(E,F)$ induces a function
$\Fun s:\Cl E\to\Cl F$ by
$\Fun s(x)=\{b\St\exists m\in\Mfin x\ (m,b)\in s\}$. The functions
$f:\Cl E\to\Cl F$ definable in that way are exactly the \emph{stable
  functions}: $f$ is stable if for any $x\in\Cl E$ and any $b\in f(x)$
there is exactly one minimal subset $x_0$ of $x$ such that
$b\in f(x_0)$, and moreover this $x_0$ is finite. When moreover this
$x_0$ is always a singleton $f$ is said \emph{linear} and such linear
functions are in bijection with $\COH(E,F)$ (given $t\in\COH(E,F)$,
the associated linear function $\Cl E\to\Cl F$ is the map
$x\mapsto\Matappa tx$).

Notice that for a given stable function $f:\Cl E\to\Cl F$ there can be
infinitely many $s\in\Kl\COH(E,F)$ such that $f=\Fun s$ since
the definition of $\Fun s$ does not take into account the
multiplicities in the multisets $m$ such that $(m,b)\in s$. For
instance, if $a\in\Web E$ and $b\in\Web F$ then $\{(\Mset{a},b)\}$ and
$\{(\Mset{a,a},b)\}$ define exactly the same stable (actually linear)
function.

Up to trivial iso we have
$\Web{\Excl{\Sfun E}}=\{(m_0,m_1)\in\Web{\Excl E}\St \Supp{m_0}\cap
\Supp{m_1}=\emptyset\text{ and }m_0+ m_1\in\Web{\Excl E}\}$ and
$\Coh{\Excl{\Sfun E}}{(m_{00},m_{01})}{(m_{10},m_{11})}$ if
$m_{00}+ m_{01}+ m_{10}+ m_{11}\in\Web{\Excl X}$ and
$\Supp{m_{00}+ m_{10}}\cap\Supp{m_{01}+ m_{11}}=\emptyset$. With this
identification we define
$\Sdiff_E\subseteq\Web{\Limpl{\Excl{\Sfun E}}{\Sfun{\Excl E}}}$ as
follows:
\begin{multline}\label{eq:def-sdiff-coh}
  \Sdiff_E=\{((m_0,\Emptymset),(0,m_0))\St m_0\in\Web{\Excl E}\}\\
  \cup\{((m_0,\Mset a),(1,m_0+\Mset a))
  \St m_0+\Mset a\in\Web{\Excl E}\text{ and }a\notin \Supp{m_0}\}\,.
\end{multline}
We think useful to check directly that %
$\Sdiff_E\in\COH(\Excl{\Sfun E},\Sfun{\Excl E})$ although this
checking is not necessary since we shall see in
Section~\ref{sec:coalg-to-diff} that this property results from a much
simpler one.
Let $((m_{j0},m_{j1}),(i_j,m_j))\in\Sdiff_E$ for $j=0,1$ and assume
that
%
% \inlineeq[eq:coh-sdiff-hyp]{\Coh{\Excl{\Sfun
%       E}}{(m_{00},m_{01})}{(m_{10},m_{11})}}.
%
%
\begin{align}\label{eq:coh-sdiff-hyp}
  \Coh{\Excl{\Sfun E}}{(m_{00},m_{01})}{(m_{10},m_{11})}\,.
\end{align}
By symmetry, there are 3 cases to consider.
\begin{itemize}
\item If $i_0=i_1=0$ then we have $m_{j1}=\Emptymset$ and $m_{j0}=m_j$
  for $j=0,1$. Then we have $\Coh{\Sfun{\Excl E}}{(0,m_0)}{(0,m_1)}$ by
  our assumption~\eqref{eq:coh-sdiff-hyp}, and if $(0,m_0)=(0,m_1)$ then
  $(m_{00},m_{01})=(m_{10},m_{11})$.
\item Assume now that $i_0=i_1=1$. We have $m_{j1}=\Mset{a_j}$ for
  $a_j\in\Web E$, with $a_j\notin \Supp{m_{j0}}$ and
  $m_{j}=m_{j0}+\Mset{a_j}$. Our assumption~\eqref{eq:coh-sdiff-hyp}
  means that $m_{00}+ m_{10}+\Mset{a_0,a_1}\in\Web{\Excl E}$ and
  $\Supp{m_{00}+ m_{10}}\cap\{a_0,a_1\}=\emptyset$. Therefore
  $m_0+ m_1\in\Web{\Excl E}$ and hence
  $\Coh{\Sfun{\Excl E}}{(1,m_0)}{(1,m_1)}$. Assume moreover that
  $m_0=m_1$, that is $m_{00}+\Mset{a_0}=m_{10}+\Mset{a_1}$. This
  implies $m_{00}=m_{10}$ and $a_0=a_1$ since we know that
  $a_1\notin \Supp{m_{00}}$ and $a_0\notin \Supp{m_{10}}$.
\item Last assume that $i_0=1$ and $i_1=0$. So we have
  $m_{01}=\Mset{a}$ with $a\notin \Supp{m_{00}}$ and
  $m_0=m_{00}+\Mset{a}$; $m_{11}=\Emptymset$ and
  $m_1=m_{10}$. By~\eqref{eq:coh-sdiff-hyp} we know that
  $\Supp{m_0+ m_1}\in\Cl{\Excl E}$. Coming back to the definition of
  the coherence in $\Sfun F$ (for a coherence space $F$), we must also
  prove that $m_0\not=m_1$: this results from~\eqref{eq:coh-sdiff-hyp}
  which entails that $a\notin \Supp{m_1}=m_{10}$ whereas we know that
  $a\in \Supp{m_0}$.
\end{itemize}
We postpone the proofs of the other commutations as they will be
reduced in Section~\ref{sec:coalg-to-diff} to much simpler properties because $\COH$ id .
Given $x\in\Cl E$, we can define a coherence space $E_x$ (the local
sub-coherence space at $x$) as follows:
$\Web{E_x}=\{a\in\Web E\setminus x\St x\cup\{a\}\in\Cl X\}$ and
$\Coh{E_x}{a_0}{a_1}$ if $\Coh E{a_0}{a_1}$. Then, given
$s\in\Kl\COH(E,F)$, we can define the \emph{differential} of $s$ at
$x$ as
\begin{align*}
  \Diffrac{s(x)}{x}
  =\{(a,b)\in\Web{E_x}\times\Web F\St\exists m\in\Web{\Excl E}\
  (m+\Mset a,b)\in s\text{ and } \Supp m\subseteq x\}
  \subseteq\Web{\Limpl{E_x}{Y}}\,.
\end{align*}
% that is, for each $u\in\Cl{E_x}$,
% $\Matappa{\Diffrac{s(x)}{x}}u=\cup_{a\in u}(\Fun s(x\cup\Eset
% a)\setminus\Fun s(x))$ as easily checked.

  % \begin{lemma}
  %   $\Diffst_x f$ is a linear map $\Cl{E_x}\to\Cl{F_{f(x)}}$. Its
  %   trace is given by
  %   \begin{align*}
  %     \Ltrace(\Diffst_xf)=\{(a,b)\in\Web{E_x}\times\Web F\St
  %     \exists x_0\in\Web{\Excl E}\ x_0\subseteq x\text{ and }(x_0\cup\{a\},b)\in\Trace f\}
  %   \end{align*}
  % \end{lemma}
  % \begin{proof}
  %   The fact that $\Diffst_xf(u)\in\Cl{F_{f(x)}}$ results from
  %   $\Diffst f(u)\subseteq f(x\cup u)\setminus f(x)$ (notice that if
  %   $u\in\Cl{E_x}$ then $x\cup u\in\Cl E$). The fact that $\Diffst_xf$
  %   commutes with arbitrary unions results from its definition.
  % \end{proof}

\begin{theorem}\label{th:coh-diff-clique}
  Let $s:\COH(E,F)$. Then $\Sdfun s\in\Kl\COH(\Sfun E,\Sfun F)$ satisfies
  \begin{align*}
    \forall (x,u)\in\Cl{\Sfun E}\quad
    \Fun{\Sdfun s}(x,u)=(\Fun s(x),\Matappa{\Diffrac{s(x)}x}u)
  \end{align*}
\end{theorem}
\begin{remark}\label{rk:diff-coh-multi}
  The definition of $\Sdfun s$ depends on $s$ and not only on
  $\Fun s$: for instance if $s=\Eset{(\Mset a,b)}$ then
  $\Sdfun s=\{((\Mset a,\Emptymset),(0,b)),((\Emptymset,\Mset
  a),(1,b))\}$ and if $s'=\Eset{(\Mset{a,a},b)}$ then
  $\Sdfun s'=\{((\Mset{a,a},\Emptymset),(0,b))\}$; in that case the
  derivative vanishes whereas $\Fun s=\Fun{s'}$ are the same function. 
\end{remark}

\begin{proof}
  Let $(x,u)\in\Cl{\Sfun E}$ and $(i,b)\in\Web{\Sfun F}$ with
  $i\in\{0,1\}$ and $b\in\Web F$. We have
  $(i,b)\in\Fun{\Sdfun s}(x,u)$ iff there is
  $(m_0,m_1)\in\Web{\Excl{\Sfun E}}$ such that $\Supp{m_0}\subseteq x$,
  $\Supp{m_1}\subseteq u$ and
  $((m_0,m_1),(i,b))\in\Sdfun s=\Sdiff_E\Compl\Sfun s$.
  This latter condition holds iff either $i=0$, $m_1=\Emptymset$, and
  $(m_0,b)\in s$, or $i=1$, $m_1=\Mset a$ for some
  $a\in\Web E\setminus\Supp{m_0}$ such that $m_0+\Mset a\in\Cl E$, and
  $(m_0+\Mset{a},b)\in s$.

  Assume first that $(i,b)\in \Fun{\Sdfun s}(x,u)$ and let $(m_0,m_1)$
  be as above. If $i=0$ we have $(m_0,b)\in s$ and $\Supp m_0\subseteq x$
  and hence $b\in\Fun s(x)$, that is
  $(i,b)\in(\Fun s(x),\Matappa{\Diffrac{s(x)}x}u)$. If $i=1$ let
  $a\in\Web E\setminus\Supp{m_0}$ be such that $m_1=\Mset{a}$,
  $m_0+\Mset{a}\in\Web{\Excl E}$, $(m_0+\Mset{a},b)\in s$ and
  $\Supp{m_0,\Mset a}\subseteq(x,u)$ (remember that we consider the
  elements of $\Cl{\Sfun E}$ as pairs of cliques), that is
  $\Supp{m_0}\subseteq x$ and $a\in u$. Then we know that
  $a\in\Web{E_x}$ since $x\cup u\in\Cl E$ and $x\cap
  u=\emptyset$. Therefore
  $(i,b)\in(\Fun s(x),\Matappa{\Diffrac{s(x)}x}u)$.

  We have proven
  $\Fun{\Sdfun s}(x,u)\subseteq(\Fun
  s(x),\Matappa{\Diffrac{s(x)}x}u)$, we prove the converse inclusion.
  Let $(i,b)\in (\Fun s(x),\Matappa{\Diffrac{s(x)}x}u)$. If $i=0$, we
  have $b\in \Fun s(x)$ and hence there is a uniquely defined
  $m_0\in\Web{\Excl E}$ such that $\Supp{m_0}\subseteq x$ and
  $(m_0,b)\in s$. It follows that
  $((m_0,\Emptymset),(0,b))\in\Sdiff_E\Compl\Sfun s$ and hence
  $(i,b)\in \Fun{\Sdfun s}(x,u)$. Assume now that $i=1$ so that
  $b\in\Matappa{\Diffrac{f(x)}x}u$ and hence there is $a\in u$ (which
  implies $a\notin x$) such that $(a,b)\in\Diffrac{s(x)}x$. So there
  is $m_0\in\Web{\Excl E}$ such that $\Supp{m_0}\subseteq x$ and
  $(m_0+\Mset a,b)\in s$ (notice that $a\notin\Supp{m_0}$ since
  $\Supp{m_0}\subseteq x$ and $a\notin x$). It follows that
  $((m_0,\Mset a),(1,m_0+\Mset a))\in\Sdiff_E$ and hence
  $((m_0,\Mset a),(1,b))\in (\Sfun s)\Compl\Sdiff_E$ so that
  $(1,b)\in\Fun{\Sdfun s}(x,u)$.
\end{proof}

\begin{remark}\label{rk:diff-coh-uniform}
  This shows in particular that
  $\Diffrac{f(x)}x\in\COH(E_x,F_{\Fun s(x)})$ since
  $\Diffrac{f(x)}x=\Sproj 1\Comp g\Comp\Sin 1$ and also that this
  derivative is stable with respect to the point $x$ where it is
  computed and thus differentiation of stable functions can be
  iterated. However Remark~\ref{rk:diff-coh-multi} indicates a
  peculiarity of this derivative which has as consequence that the
  morphisms in $\Kl\COH$ do not coincide with their Taylor expansion
  that one can define by iterating this derivative (the expansion of
  $s$ is $s$ whereas the expansion of $s'$ is $\emptyset$). This is an
  effect of the \emph{uniformity} of the construction $\Excl E$, that
  is, of the fact that for $m\in\Mfin{\Web E}$ to be in
  $\Web{\Excl E}$, it is required that $\Supp m$ be a clique. This can
  be remedied, without breaking the main feature of our construction,
  namely that it is compatible with the determinism%
  \footnote{Remember that by this we mean that, in the type of
    booleans $\Plus\Sone\Sone$ for instance, the only cliques are
    $\emptyset$, $\Eset\True$ and $\Eset\False$.} %
  of the model, by using non-uniform coherence spaces instead, where
  $\Web{\Excl E}=\Mfin E$~\cite{BucciarelliEhrhard99,Boudes11}, see
  Section~\ref{sec:nucs-diff}. In some sense, stable functions on
  Girard's coherence spaces are smooth but not analytic.
\end{remark}

\section{Canonically summable categories}\label{sec:canonical-sum}
The concept of summable category applies typically to models of Linear
Logic in the sense of Seely (see~\cite{Mellies09}): such a model is
based on an SMC $\cL$ whose morphisms are intuitively considered as
linear, and the summability structure makes this linearity more
explicit. In the models we want to apply primarily our theory to
---~typically (probabilistic) coherence spaces~--- the summability
structure boils down to a more basic structure which is always present
in such a model: the functor $\Sfun X$ is defined on objects by
$\Sfun X=\Limplp{\With\Sone\Sone}X$, and similarly for
morphisms. \emph{A priori}, given a categorical model of LL $\cL$,
this functor does not necessarily define a summability structure. The
purpose of this section is to examine under which conditions this is
the case, and to express the differential structure introduced above
in this particular and important setting.

Let $\cL$ be a cartesian\footnote{Actually we don't need all cartesian
  products, only all $n$-ary products of $1$.} SMC where the object
$\Into=\With\Sone\Sone$ is exponentiable, that is, the functor
$\Tinto:X\mapsto\Tens{X}{\Into}$ has a right
adjoint
% \footnote{Interestingly this adjunction induces on $\cL$ the
%   standard linear state monad associated with $\Into$: the functor
%   $X\mapsto(\Limpl\Into{\Tens X\Into})$, a fact which has certainly a
%   computational interpretation related with differentiation.}
$\Scfun:X\mapsto\Limplp{\Into}{X}$.
We use
$\Evlin\in\cL(\Tens{(\Limpl{\Into}{X})}{\Into},X)$ for the
% TYPO
% corresponding evaluation morphism.  We denote this functor as
% $\Scfun$, notice that, being a right adjoint, it preserves all limits
corresponding evaluation morphism and, given %
$f\in\cL(\Tens Y\Into, X)$ we use $\Curlin f$ for the associated %
\emph{Curry transpose} of $f$ which satisfies %
$\Curlin f\in\cL(Y,\Limpl\Into X)$.  Being a right adjoint, $\Scfun$
preserves all limits existing in $\cL$ (and in particular the
cartesian product).

We shall use the construction provided by the following lemma.

\begin{lemma}\label{lemma:into-point-tnat}
  Let $\phi\in\cL(\Sone,\Into)$. %
  For any object $X$ of $\cL$ let %
  $\Scfunnt \phi_X\in\cL(\Limpl\Into X,X)$ be the following composition of
  morphisms
  \[
    \begin{tikzcd}
    \Limplp\Into X\ar[r,"\Inv\Rightu_{\Limpl\Into X}"]
    &
    \Tens{\Limplp\Into X}{\Sone}\ar[r,"\Tens{\Limplp\Into X}{\phi}"]
    &[2em]
    \Tens{\Limplp\Into X}{\Into}\ar[r,"\Evlin"]
    &[-1em]
    X
    \end{tikzcd}
  \]
  Then $(\Scfunnt \phi_X)_{X\in\cL}$ is a natural transformation.

  Let %
  $f\in\cL(\Tens Y\Into,X)$, so that %
  $\Curlin f\in\cL(Y,\Limpl\Into X)$. Then one has
  \begin{align*}
    \Scfunnt\phi_X\Compl(\Curlin f)=f\Compl\Tensp Y\phi\Compl\Inv{\Rightu_Y}
    \in\cL(Y,X)\,.
  \end{align*}
\end{lemma}
\begin{proof}
  Naturality results from the naturality of $\Rightu$ and
  functoriality of $\Limpl\Into\_$. Let us prove the second part of
  the lemma, we have:
  \begin{align*}
    \Scfunnt\phi_X\Compl(\Curlin f)
    &=\Evlin
      \Compl\Tensp{\Limplp\Into X}\phi
      \Compl\Invp{\Rightu_{\Limpl\Into X}}
      \Compl(\Curlin f)\\
    &=\Evlin
      \Compl\Tensp{\Limplp\Into X}\phi
      \Compl\Tensp{(\Curlin f)}\Sone
      \Compl\Inv{\Rightu_Y}\\
    &=\Evlin
      \Compl\Tensp{(\Curlin f)}\Into
      \Compl\Tensp Y\phi
      \Compl\Inv{\Rightu_Y}\\
    &=f\Compl\Tensp Y\phi\Compl\Inv{\Rightu_Y}\,.
  \end{align*}
\end{proof}

For $i=0,1$ we have a morphism $\Win i\in\cL(\Sone,\Into)$ given by
$\Win 0=\Tuple{\Id_\Sone,0}$ and $\Win 1=\Tuple{0,\Id_\Sone}$. %
We also have a diagonal morphism
$\Wdiag=\Tuple{\Id_\Sone,\Id_\Sone}\in\cL(\Sone,\Into)$. %
Using these we define the following natural transformations
$\Scfun X\to X$:
\begin{align*}
  \Sproj i&=\Scfunnt{\Win i}\text{\quad for }i=0,1\\
  \Ssum&=\Scfunnt{\Wdiag}\,.
\end{align*}

\begin{definition}
  The category $\cL$ is \emph{canonically summable} if
  $(\Scfun,\Sproj0,\Sproj1,\Ssum)$ is a summability structure.
\end{definition}

\begin{remark}
  Canonical summability is a \emph{property} of $\cL$ and not an
  additional structure, which is however defined in a rather implicit
  manner. We exhibit three elementary conditions that are necessary
  and sufficient for guaranteeing canonical summability.
\end{remark}

% Remember that $\Inv{\Rightu_X}$ is the iso $X\to\Tens X\Sone$ provided by
% the monoidal structure of $\cL$.

\begin{lemma}\label{lemma:can-epic-monic}
  The following conditions are equivalent
  \begin{itemize}
  \item for any $X\in\cL$, the morphisms %
    $\Tens X{\Win0},\Tens X{\Win1}$ are jointly epic
  \item $(\Scfun,\Sproj0,\Sproj1,\Ssum)$ is a pre-summability
    structure on $\cL$.
  \end{itemize}
\end{lemma}
\begin{proof}
  Assume that $\Tens X{\Win0},\Tens X{\Win1}$ are jointly epic and
  let %
  $f_j\in\cL(X,\Scfun Y)$ for $j=0,1$ be such that %
  $\Sproj i\Compl f_0=\Sproj i\Compl f_1$ for $i=0,1$.
  Let %
  $f'_j=\Evlin\Compl\Tensp{f_j}{\Into}\in\cL(\Tens X\Into,Y)$ %
  so that $f_j=\Curlin{f'_j}$, for $j=0,1$. We have
  \begin{align*}
    \Sproj i\Compl f_j
    &=\Scfunnt{\Win i}\Compl(\Curlin{f'_j})\\
    &=f'_j
      \Compl\Tensp{X}{\Win j}
      \Compl\Inv{\Rightu_X}\text{\quad by Lemma~\ref{lemma:into-point-tnat}.}
  \end{align*}
  So we have $f'_0=f'_1$ by our assumption on the $\Win j$'s and hence
  $f_0=f_1$.

  Assume conversely that $\Sproj0,\Sproj1$ are jointly monic and let %
  $f_0,f_1\in\cL(\Tens X\Into,Y)$ be such that %
  $f_0\Compl\Tensp X{\Win i}=f_1\Compl\Tensp X{\Win i}$ for
  $i=0,1$. By Lemma~\ref{lemma:into-point-tnat} again we have %
  $f_j\Compl\Tensp X{\Win i}=\Sproj
  i\Compl(\Curlin{f_j})\Compl\Rightu_X$ and hence %
  $\Curlin{f_0}=\Curlin{f_1}$ and hence $f_0=f_1$ which proves that %
  $\Tens X{\Win0},\Tens X{\Win1}$ are jointly epic.
\end{proof}

\begin{theorem}\label{th:can-summable}
  Let $\cL$ be a cartesian SMC where the object %
  $\Into=\With\Sone\Sone$ is exponentiable.
  Setting %
  $\Sproj i=\Scfunnt{\Win i}$ for $i=0,1$ and %
  $\Ssum=\Scfunnt\Wdiag$, the two following statements are equivalent.
  \begin{enumerate}
  \item\label{it:cond1-th-can-summable} For any $X\in\cL$, the
    morphisms %
    $\Tens X{\Win0},\Tens X{\Win1}$ are jointly epic %
    (we call \Csaxepi{} this condition) and %
    $(\Scfun,\Sproj0,\Sproj1,\Ssum)$ satisfies \Saxwit{}, see
    Section~\ref{sec:sum-cat}.
  \item\label{it:cond2-th-can-summable}
    $(\Scfun,\Sproj0,\Sproj1,\Ssum)$ is a summable category that is,
    $\cL$ is canonically summable.
  \end{enumerate}
\end{theorem}
\begin{proof}
  The fact that
  (\ref{it:cond2-th-can-summable})~$\Implies$~(\ref{it:cond1-th-can-summable})
  results immediately from Lemma~\ref{lemma:can-epic-monic} so let us
  prove the converse. %
  We assume that (\ref{it:cond1-th-can-summable}) holds.
  By Lemma~\ref{lemma:can-epic-monic} we know that %
  $\Sproj0,\Sproj1$ are jointly monic, so we are left with proving %
  \Saxcom, \Saxzero, \Saxass{} and \Saxdist.

  \Proofcase
  \Saxcom{}. Let %
  $f=\Curlin g\in\cL(\Scfun X,\Scfun X)$ where $g$ is the following
  composition of morphisms
  \[
    \begin{tikzcd}
      \Tens{\Limplp\Into X}{\Into}\ar[r,"\Tens{\Id}{\Tuple{\Proj1,\Proj0}}"]
      &[2.6em]\Tens{\Limplp\Into X}{\Into}\ar[r,"\Evlin"]
      &[-1em]X
    \end{tikzcd}
  \]
  We have 
  \begin{align*}
    \Sproj i\Compl f
    &=g\Compl\Tensp{\Limplp\Into X}{\Win i}\Compl\Inv\Rightu
    \text{\quad by Lemma~\ref{lemma:into-point-tnat}}\\
    &=\Evlin\Compl\Tensp{\Limplp\Into X}{\Win{1-i}}\Compl\Inv\Rightu
    \text{\quad by definition of $g$}\\
    &=\Sproj{1-i}
  \end{align*}
  and similarly
  \begin{align*}
    \Ssum\Compl f
    =g\Compl\Tensp{\Limplp\Into X}{\Wdiag}\Compl\Inv\Rightu
    =\Evlin\Compl\Tensp{\Limplp\Into X}{\Wdiag}\Compl\Inv\Rightu=\Ssum\,.
  \end{align*}

  \Proofcase
  \Saxzero. Let $f\in\cL(X,Y)$. %
  Let %
  $h=\Curlinp{f\Compl\Rightu_X\Compl\Tensp X{\Proj 0}}\in\cL(X,\Scfun Y)$. %
  We have
  \begin{align*}
    \Sproj i\Compl h
    &=f\Compl\Rightu_X\Compl\Tensp X{\Proj 0}
      \Compl\Tensp{X}{\Win i}
      \Compl\Inv{\Rightu_X}
    \text{\quad by Lemma~\ref{lemma:into-point-tnat}}\\
    &=
      \begin{cases}
        f & \text{if }i=0\\
        0 & \text{otherwise}
      \end{cases}
  \end{align*}
  which shows that $f,0$ are summable with %
  $\Stuple{f,0}=h$.
  Moreover %
  \begin{align*}
    \Ssum\Compl h=f\Compl\Rightu_X\Compl\Tensp X{\Proj 0}
      \Compl\Tensp{X}{\Wdiag}
      \Compl\Inv{\Rightu_X}=f\,.
  \end{align*}

  \Proofcase
  \Saxass{}. We define %
  $\Sflip_X\in\cL(\Scfun^2 X,\Scfun^2X)$ by %
  $\Sflip_X=\Curlin(\Curlin(\Evlin\Compl\Tensp\Evlin\Into
  \Compl(\Scfun^2X\ITens\Scflip)))$ where the transposed morphism is
  typed as follows.
  \[
    \begin{tikzcd}
      \Scfun^2X\ITens\Into\ITens\Into
      \ar[r,"\Id\ITens\Scflip"]
      &\Scfun^2X\ITens\Into\ITens\Into
      \ar[r,"\Evlin\ITens\Into"]
      &\Scfun X\ITens\Into\ar[r,"\Evlin"]
      &X
    \end{tikzcd}
  \]
  A computation similar to the previous ones shows that
  $\Sproj i\Compl\Sproj j\Compl\Sflip=\Sproj j\Compl\Sproj i$ as
  required. We have moreover
  \begin{align*}
    \Ssum_{\Scfun X}\Compl\Sflip_X
    &=\Curlin(\Evlin\Compl\Tensp\Evlin\Into\Compl(\Scfun^2X\ITens\Scflip))
      \Compl((\Scfun^2X)\ITens\Wdiag)\Compl\Inv{\Rightu}
      \text{\quad by Lemma~\ref{lemma:into-point-tnat}
      and definition of }\Ssum\\
    &=\Curlin(\Evlin\Compl\Tensp\Evlin\Into
      \Compl(\Scfun^2X\ITens\Scflip)\Compl(\Scfun^2X\ITens\Wdiag\ITens\Into))
      \Compl\Inv{\Rightu}\\
    &=\Curlin(\Evlin\Compl\Tensp\Evlin\Into
      \Compl(\Scfun^2X\ITens(\Scflip\Compl\Tensp{\Wdiag}{\Into})))
      \Compl\Inv{\Rightu}\\
    &=\Curlin(\Evlin\Compl\Tensp\Evlin\Into
      \Compl(\Scfun^2X\ITens(\Tensp{\Into}{\Wdiag}\Compl\Sym_{\Sone,\Into})))
      \Compl\Inv{\Rightu}\text{\quad by definition of }\Scflip\\
    &=\Curlin(\Evlin\Compl\Tensp\Evlin\Into
      \Compl(\Scfun^2X\ITens\Into\ITens\Wdiag)
      \Compl(\Tens{\Scfun^2X}\Sym_{\Sone,\Into}))
      \Compl\Inv{\Rightu}\\
    &=\Curlin(\Evlin\Compl(\Tens{\Scfun X}{\Wdiag})\Tensp\Evlin\Sone
      \Compl(\Scfun^2X\ITens\Sym_{\Sone,\Into}))
      \Compl\Inv{\Rightu}
    \text{\quad by functoriality of }\ITens\\
    &=\Curlin(\Ssum_X\Compl\Rightu\Compl(\Evlin\ITens\Sone)
      \Compl(\Scfun^2X\ITens\Sym_{\Sone,\Into}))\Compl\Inv{\Rightu}
    \text{\quad by definition of }\Ssum_X\\
    &=\Curlin(\Ssum_X\Compl\Evlin\Tensp\Rightu\Into)
      \Compl\Inv{\Rightu}\text{\quad by nat.~of }\Rightu
    \text{ and standard SMC commutations}\\
    &=\Curlin(\Ssum_X\Compl\Evlin)=\Scfun\Ssum\,.
  \end{align*}

  \Proofcase
  \Saxdist{}. Let $(f_{00},f_{01})$ be a summable pair
  of morphisms in $\cL(X_0,Y_0)$ so that we have the witness %
  $\Stuple{f_{00},f_{01}}\in\cL(X_0,\Scfun{Y_0})$, and let %
  $f_1\in\cL(X_1,Y_1)$. %
  Let %
  $h=\Curlin{h'}\in\cL(\Tens{X_0}{X_1},\Scfun\Tensp{Y_0}{Y_1})$ where
  $h'$ is the following composition of morphisms:
  \[
    \begin{tikzcd}
      X_0\ITens X_1\ITens\Into\ar[r,"\Tens{\Stuple{f_{00},f_{01}}}{\Sym}"]
      &[2.8em]\Limplp{\Into}{Y_0}\ITens\Into\ITens X_1\ar[r,"\Tens\Evlin{f_1}"]
      &Y_0\ITens Y_1\,.
    \end{tikzcd}
  \]
  We have
  \begin{align*}
    \Sproj i\Compl h
    &=\Tensp{\Evlin}{f_1}
      \Compl\Tensp{\Stuple{f_{00},f_{01}}}{\Sym_{X_1,\Into}}
      \Compl\Tensp{\Tens{X_0}{X_1}}{\Win i}
      \Compl\Inv{\Rightu_{\Tens{X_0}{X_1}}}
    \text{\quad by Lemma~\ref{lemma:into-point-tnat}}\\
    &=\Tensp{\Evlin}{f_1}
      \Compl\Tensp{\Stuple{f_{00},f_{01}}}{\Tens{\Win i}{X_1}}
      \Compl\Tensp{X_0}{\Sym_{X_1,\Sone}}
      \Compl\Inv{\Rightu_{\Tens{X_0}{X_1}}}\\
    &=\Tensp{(\Evlin\Compl\Tensp{\Stuple{f_{00},f_{01}}}{\Win i})}{f_1}
      \Compl\Tensp{X_0}{\Sym_{X_1,\Sone}}
      \Compl\Inv{\Rightu_{\Tens{X_0}{X_1}}}\\
    &=\Tensp{f_{0i}}{f_1}
      \Compl\Tensp{\Rightu_{X_0}}{X_1}
      \Compl\Tensp{X_0}{\Sym}
      \Compl\Inv{\Rightu_{\Tens{X_0}{X_1}}}\\
    &=\Tens{f_{0i}}{f_1}\,.
  \end{align*}
  which shows that %
  $\Tens{f_{00}}{f_1},\Tens{f_{01}}{f_1}$ are summable with
  \begin{align*}
    \Stuple{\Tens{f_{00}}{f_1},\Tens{f_{01}}{f_1}}=h\,.
  \end{align*}
  We have by a similar computation
  \begin{align*}
    \Ssum\Compl h
    &=\Tensp{(\Evlin\Compl\Tensp{\Stuple{f_{00},f_{01}}}{\Wdiag})}{f_1}
      \Compl\Tensp{X_0}{\Sym_{X_1,\Sone}}
      \Compl\Inv{\Rightu_{\Tens{X_0}{X_1}}}\\
    &=\Tensp{(f_{00}+f_{01})}{f_1}
      \Compl\Tensp{\Rightu_{X_0}}{X_1}
      \Compl\Tensp{X_0}{\Sym}
      \Compl\Inv{\Rightu_{\Tens{X_0}{X_1}}}\\
    &=\Tens{(f_{00}+f_{01})}{f_1}\,.    
  \end{align*}

\end{proof}

There are cartesian SMC where $\Into$ is exponentiable and which are
not canonically summable. The category $\PSET$ provides probably the
simplest example of that situation.
\begin{Example}
  % Let $\PSET$ be the category of pointed sets. We use $0_X$ for the
  % singled out point of the object $X$. A morphism $f\in\PSET(X,Y)$ is
  % a function $f:X\to Y$ such that $f(0_X)=0_Y$.  The terminal object
  % is the singleton $\{0\}$. The cartesian product $\With XY$ is the
  % ordinary cartesian product, with $0_{\With XY}=(0_X,0_Y)$. The
  % tensor product $\Tens XY$ is defined as
  % \begin{align*}
  %   \Tens XY=\{(x,y)\in X\times Y\St x=0\Equiv y=0\}
  % \end{align*}
  % with $0_{\Tens XY}=(0_X,0_Y)$. The unit of the tensor product is the
  % object $\Sone=\{0,\Sonelem\}$ of $\PSET$. This category is enriched
  % over itself, the singled out point of $\PSET(X,Y)$ being the
  % constantly $0_Y$ function. Actually, it is monoidal closed with
  % $\Limpl XY=\PSET(X,Y)$ and $0_{\Limpl XY}$ defined by
  % $0_{\Limpl XY}(x)=0_Y$ for all $x\in X$. A mono in $\PSET$ is a
  % morphism of $\PSET$ which is injective as a function.
  %
  We refer to Section~\ref{sec:pointed-sets}.  We have the functor
  $\Scfun:\PSET\to\PSET$ defined by $\Scfun X=(\Limpl{\Into}{X})$. An
  element of $\Scfun X$ is a function $z:\{0,\Sonelem\}^2\to X$ such
  that $z(0,0)=0$. %
  The projections $\Sproj i:\Scfun X\to X$ are characterized by
  $\Sproj 0(z)=z(\Sonelem,0)$ and $\Sproj 1(z)=z(0,\Sonelem)$, %
  so $\Tuple{\Sproj 0,\Sproj 1}$ is not injective since
  $\Tuple{\Sproj 0,\Sproj 1}(z)=(z(\Sonelem,0),z(0,\Sonelem))$ does
  not depend on $z(\Sonelem,\Sonelem)$ which can take any value. %
  So $(\Scfun,\Sproj0,\Sproj1,\Ssum)$ is not even a pre-summability
  structure in $\PSET$. This failure of injectivity is due to the fact
  that $\Into$ lacks an addition which would satisfy
  $(\Sonelem,0)+(0,\Sonelem)=(\Sonelem,\Sonelem)$ and, preserved by
  $z$, would enforce injectivity.
\end{Example}

There are also cartesian SMC where $\Into$ is exponentiable, where
\Csaxepi{} holds but where $(\Scfun,\Sproj0,\Sproj1,\Ssum)$ does not
satisfy \Saxwit{}.
\begin{Example}
  Let $\cB$ be the category whose objects are the finite dimensional
  real Banach space. By this we mean pairs $(V,\Norm\__V)$ where $V$
  is a finite dimensional real vector space and $\Norm\__V$ is a norm
  on $V$. In $\cB$, a morphism $V\to W$ is a linear map such that %
  $\forall v\in V\ \Norm{f(v)}_W\leq\Norm v_V$. This category is a
  cartesian symmetric monoidal closed category with $\Limpl UV$
  defined as the space of \emph{all} linear maps $f:U\to V$ and
  \[
    \Norm f_{\Limpl UV}=\sup\{\Norm{f(u)}_V\St u\in U\text{ and }\Norm
    u_U\leq 1\}\,.
  \]
  Indeed since we consider only finite dimensional spaces, all linear
  maps are continuous (for the product topology induced by any choice
  of basis, which is the same as the one induce by the norm) and hence
  bounded.  The tensor product classifies bilinear maps (with norm
  defined by sups as for linear maps) and satisfies
  $\Norm{\Tens uv}_{\Tens UV}=\Norm u_U\Norm v_V$ for all $u\in U$ and
  $v\in V$. %
  The unit of this tensor product is $\Sone=\Real$ with
  $\Norm u_\Sone=\Absval u$. %
  The cartesian product is the standard direct product of vector
  spaces with $\Norm{(u,v)}_{\With UV}=\max(\Norm u_V,\Norm
  v_V)$. Notice that there is also a coproduct $\Plus UV$, with the
  same underlying vector space and
  $\Norm{(u,v)}_{\Plus UV}=\Norm u_U+\Norm v_V$. So $\With UV$ and
  $\Plus UV$ \emph{are not isomorphic} in $\cB$ which is not an
  additive category.
  
  The functor $\Scfun:\cB\to\cB$ maps $U$ to %
  $V=\Scfun U$ where %
  $V=\{(u_0,u_1)\in U\St \Norm{u_0+u_1}_U\leq U\}$ and %
  \[
    \Norm{(u_0,u_1)}_V=\sup\{\Norm{a_0u_0+a_1u_1}_U\St
    (a_0,a_1)\in\Intercc{-1}1\times\Intercc{-1}1\}
  \]
  The natural transformations $\Sproj i$ are the obvious projections
  and $\Ssum(u_0,u_1)=u_0+u_1$.

  Then, in $\Sone=\Real$:
  \begin{itemize}
  \item $-1/2$ and $1/2$ are summable because %
    $\Absval{-\frac a2+\frac b2}\leq 1$ for all $a,b\in\Intercc{-1}{1}$
  \item $-1/2+1/2=0$ and $1$ are summable in $\Sone$
  \item but $1/2$ and $1$ are not summable in $\Sone$.
  \end{itemize}
  So $\cB$ is not canonically summable.
\end{Example}
This example shows that the condition \Saxwit{} cannot be disposed of
and speaks not only of associativity of partial sum, but also of some
kind of ``positivity'' of morphisms in $\cL$.

\subsection{The comonoid structure of $\Into$}
We assume that $\cL$ is a canonically summable cartesian SMC.  The
morphisms $\Win0,\Win1\in\cL(\Sone,\Into)$ are summable with
$\Win0+\Win1=\Wdiag$, with witness $\Id\in\cL(\Into,\Into)$. As a
consequence of \Saxdist{} the morphisms
$\Tensp{\Win0}{\Win0}\Compl\Inv\Rightu$,
$\Tensp{\Win0}{\Win1}\Compl\Inv\Rightu$ and
$\Tensp{\Win1}{\Win0}\Compl\Inv\Rightu$ are summable in
$\cL(\Sone,\Tens\Into\Into)$. Therefore
$\Tensp{\Win0}{\Win0}\Compl\Inv\Rightu$ and
$\Tensp{\Win0}{\Win1}\Compl\Inv\Rightu+\Tensp{\Win1}{\Win0}\Compl\Inv\Rightu$
are summable in $\cL(\Sone,\Tens\Into\Into)$ so there is a uniquely
defined $\Scmont\in\cL(\Into,\Tens\Into\Into)$ such that
%\begin{align*}
  \(
  \Scmont\Compl\Win0=\Tensp{\Win0}{\Win0}\Compl\Inv\Rightu
  \text{ and }
  \Scmont\Compl\Win1
  =\Tensp{\Win0}{\Win1}\Compl\Inv\Rightu
  +\Tensp{\Win1}{\Win0}\Compl\Inv\Rightu\,.
  \)
%\end{align*}
% And the
% morphisms $\Win0\Compl\Proj0,\Win1\Compl\Proj1\in\cL(\Into,\Into)$ are
% summable with witness
% $\Tuple{\Tens{\Proj0}{\Proj0},\Tens{\Proj1}{\Proj1}}
% \in\cL(\Tens\Into\Into,\Into)$: indeed we have
% $\Tuple{\Tens{\Proj0}{\Proj0},
%   \Tens{\Proj1}{\Proj1}}\Compl\Tensp\Into{\Win0}\Compl\Inv\Rightu=
% \Tuple{\Proj0,0}=\Tuple{\Id_\Sone,0}\Compl\Proj0 =\Win0\Compl\Proj0$
% and similarly for $\Win1\Compl\Proj1$. Moreover one checks easily that
% $\Win0\Compl\Proj0+\Win1\Compl\Proj1=\Id_\Into$ using \Csaxepi{} and
% the fact that $\Proj i\Compl\Win i$ is equal to $\Id_\Into$ if $i=j$
% and to $0$ otherwise. By \Saxdist{} the morphisms
% $\Tens{\Win0\Compl\Proj0}{\Win0},\Tens{\Win1\Compl\Proj1}{\Win0},
% \Tens{\Win0\Compl\Proj0}{\Win1}\in\cL(\Tens\Into\Sone,\Tens\Into\Into)$
% are summable and we set
% \begin{align*}
%   \Scmont=(\Tens{\Win0\Compl\Proj0}{\Win0}+\Tens{\Win1\Compl\Proj1}{\Win0}+
%   \Tens{\Win0\Compl\Proj0}{\Win1})
%   \Compl\Inv\Rightu\in\cL(\Into,\Tens\Into\Into)\,.
% \end{align*}
\begin{theorem}\label{th:can-comonoid-Into}
  Equipped with $\Proj0\in\cL(\Into,\Sone)$ as counit and
  $\Scmont\in\cL(\Into,\Tens\Into\Into)$ as comultiplication,
  $\Into$ is a cocommutative comonoid in the SMC $\cL$.
\end{theorem}
\begin{proof}
  To prove the required commutations, we use \Csaxepi. Here are two
  examples of these computations.
  \begin{align*}
    \Rightu\Compl\Tensp\Into{\Proj0}\Compl\Scmont\Compl\Win0
    &=\Rightu\Compl\Tensp\Into{\Proj0}
      \Compl\Tensp{\Win0}{\Win0}
      \Compl{\Inv\Rightu}
      =\Rightu\Compl\Tensp{\Win0}{\Sone}\Compl\Inv\Rightu=\Win0
  \end{align*}
and
  \begin{align*}
    \Rightu\Compl\Tensp\Into{\Proj0}\Compl\Scmont\Compl\Win1
    &=\Rightu\Compl\Tensp\Into{\Proj0}
      \Compl(\Tens{\Win0}{\Win1}+\Tens{\Win1}{\Win0})
      \Compl{\Inv\Rightu}
      =\Rightu\Compl\Tensp{\Win1}{\Sone}\Compl\Inv\Rightu=\Win1
  \end{align*}
  since $\Proj0\Compl\Win i$ is equal to $\Id_\Sone$ if $i=0$ and
  to $0$ otherwise.  Hence
  $\Rightu\Compl\Tensp\Into{\Proj0}\Compl\Scmont=\Into$. Next
  \begin{align*}
    \Tensp\Into{\Scmont}\Compl\Scmont\Compl\Win0
    &=\Tensp{\Into}{\Smont}\Compl\Tensp{\Win0}{\Win0}\Compl\Inv\Rightu
      =\Tensp{\Win0}{\Tensp{\Win0}{\Win0}}\Compl\Tensp{\Into}{\Inv\Rightu}\Compl\Inv\Rightu
  \end{align*}
  and
  \begin{align*}
    \Tensp\Into{\Scmont}\Compl\Scmont\Compl\Win1
    &=\Tensp{\Into}{\Smont}
      \Compl(\Tens{\Win0}{\Win1}+\Tens{\Win1}{\Win0})\Compl\Inv\Rightu\\
    &=(\Tens{\Win0}{\Tensp{\Win0}{\Win1}}
      +\Tens{\Win0}{\Tensp{\Win1}{\Win0}}
      +\Tens{\Win1}{\Tensp{\Win0}{\Win0}})
      \Compl\Tensp{\Into}{\Inv\Rightu}\Compl\Inv\Rightu\,.
  \end{align*}
  Similar computations show that 
  % \begin{align*}
  \(
    \Tensp{\Scmont}\Into\Compl\Scmont\Compl\Win0
    =\Tensp{\Tensp{\Win0}{\Win0}}{\Win0}
    \Compl\Tensp{\Inv\Rightu}{\Into}\Compl\Inv\Rightu
    \)
  % \end{align*}
  and
  % \begin{align*}
  \(
    \Tensp{\Scmont}\Into\Compl\Scmont\Compl\Win1
    =(\Tens{\Tensp{\Win0}{\Win0}}{\Win1}
      +\Tens{\Tensp{\Win0}{\Win1}}{\Win0}
      +\Tens{\Tensp{\Win1}{\Win0}}{\Win0})
      \Compl\Tensp{\Inv\Rightu}{\Into}\Compl\Inv\Rightu\,.
      \)
  %\end{align*}
  Therefore
  $\Assoc\Compl\Tensp{\Scmont}\Into\Compl\Scmont\Compl\Win i
  =\Tensp\Into{\Scmont}\Compl\Scmont\Compl\Win i$ for $i=0,1$ and
  $\Scmont$ is coassociative. Cocommutativity is proven similarly.
\end{proof}

\subsection{Strong monad structure of $\Scfun$}\label{sec:strength-can}
Therefore $\Tinto$ has a canonical comonad structure given by
$\Rightu\Compl\Tensp X{\Proj 0}\in\cL(\Tinto X,X)$ and
$\Assoc\Compl\Tensp{X}{\Scmont}\in\cL(\Tinto X,\Tinto^2X)$. Through
the adjunction $\Tinto\Adj\Scfun$ the functor $\Scfun$ inherits a
monad structure which is exactly the same as the monad structure of
Section~\ref{sec:Sfun-monad}. This monad structure
$(\Sin0,\Sfunadd)$ can be described as the Curry transpose of the
following morphisms (the monoidality isos are implicit)
\[
  \begin{tikzcd}
    \Tens X\Into\ar[r,"\Tens X{\Proj 0}"]
    &X
  \end{tikzcd}
  \quad
  \begin{tikzcd}
    \Limplp\Into{\Limplp\Into X}\ITens\Into\ar[r,"\Tens\Id\Scmont"]
    &[-0.6em]
    \Limplp\Into{\Limplp\Into X}\ITens\Into\ITens\Into\ar[r,"\Tens\Evlin\Into"]
    &[-0.6em]
    \Limplp\Into X\ITens\Into\ar[r,"\Evlin"]
    &[-1em]
    X
  \end{tikzcd}\,.
\]
Similarly the trivial costrength
$\Assoc\in\cL(\Tinto\Tensp XY,\Tens X{\Tinto Y})$ induces the strength
$\Sstr\in\cL(\Tens X{\Scfun Y},\Scfun\Tensp XY)$ of $\Scfun$ (the same
as the one defined in the general setting of
Section~\ref{sec:sum-moncat}). We have seen in
Section~\ref{sec:sum-moncat} that equipped with this strength $\Scfun$
is a commutative monad and recalled that there is therefore an
associated lax monoidality
$\Smont_{X_0,X_1}\in\cL(\Tens{\Scfun X_0}{\Scfun
  X_1},\Scfun\Tensp{X_0}{X_1})$ which can be seen as arising from
$\Scmont$ by transposing the following morphism (again we keep the
monoidal isos implicit)
\[
  \begin{tikzcd}
    \Limplp\Into{X_0}\ITens\Limplp\Into{X_1}\ITens\Into
    \ar[r,"\Tens\Id\Scmont"]
    &
    \Limplp\Into{X_0}\ITens\Limplp\Into{X_1}\ITens\Into\ITens\Into
    \ar[r,"\Tens\Evlin\Evlin"]
    &
    \Tens{X_0}{X_1}
  \end{tikzcd}\,.
\]

\subsection{Canonically summable SMCC}\label{sec:can-sum-SMCC} %
In a SMCC, the conditions of Theorem~\ref{th:can-summable} admit a slightly
simpler formulation.

\begin{theorem}\label{th:can-sum-smcc}
  A cartesian SMCC is canonically summable if and only if the condition
  \begin{Axicond}{\Ccsaxepi}
    $\Win0$ and $\Win1$ are jointly epic
  \end{Axicond}
  \noindent
  holds and $(\Scfun,\Sproj0,\Sproj1,\Ssum)$ satisfies \Saxwit.
\end{theorem}

\begin{Example}
  The SMCC $\COH$ is canonically summable, actually the summability
  structure we have considered on this category is exactly its
  canonical summability structure. Let us check the three conditions.
  
  The coherence space $\Into=\With\Sone\Sone$ is given by
  $\Web\Into=\Eset{0,1}$ with $\Scoh\Into 01$. Then
  $\Win i=\{(\Sonelem,i)\}$ and
  $\Wdiag=\Eset{(\Sonelem,0),(\Sonelem,1)}$. If
  $s\in\COH(\Into,F)$ then
  $(i,b)\in s\Equiv(\Sonelem,b)\in s\Compl\Win i$ for
  $i=0,1$ and hence $\Win 0,\Win 1$ are jointly epic so
  $\COH$ satisfies \Ccsaxepi.

  The functor $\Scfun$ defined by $\Scfun E=\Limplp\Into E$ (and
  similarly for morphisms) coincides exactly with the functor $\Sfun$
  described in Example~\ref{ex:Sfun-coh-def}.  Therefore the
  associated summability is the one described in
  Example~\ref{ex:coh-summability-char}.

  Let $s_i\in\COH(\Into,E)$ for $i=0,1$. Let
  $t_i=s_i\Compl\Wdiag=\{(\Sonelem,a)\in\Web{\Limpl \Sone
    E}\St((0,a)\in s_i\text{ or }(1,a)\in s_i\}$. Assume that $t_0$
  and $t_1$ are summable, that is $t_0\cap t_1=\emptyset$ and
  $t_0\cup t_1\in\COH(\Sone,E)$, we must prove that
  $s_0\cap s_1=\emptyset$ and $s_0\cup s_1\in\COH(\Into, E)$. Let
  $(j_i,a_i)\in s_i$ for $i=0,1$. We have $(\Sonelem,a_i)\in t_i$ and
  hence $a_0\not=a_1$ from which it follows that
  $(j_0,a_0)\not=(j_1,a_1)$. Since $\Coh\Into{j_0}{j_1}$ and
  $(j_0,a_0),(j_1,a_1)\in s_0\cup s_1\in\COH(\Into,E)$, we have
  $\Coh{E}{a_0}{a_1}$. Hence $s_0$ and $s_1$ are summable.
  %
  % We define
  % $\Scflip=\Eset{((i,j),(j,i))\St
  %   i,j\in\Eset{0,1}}\in\COH(\Tens\Into\Into,\Tens\Into\Into)$. Let us
  % check the diagram concerning $\Wdiag$: %
  % for $i,j,k\in\Eset{0,1}$ we have
  % $((\Sonelem,i),(j,k))\in\Scflip\Compl\Tensp{\Wdiag}{\Into}\Equiv
  % i=j$ and %
  % % $((\Sonelem,i),(j,k))\in\Scflip\Compl\Tensp{\Into}{\Wdiag}\Compl\Sym
  % % \Equiv((i,\Sonelem),(j,k))
  % % \in\Scflip\Compl\Tensp{\Into}{\Wdiag}\Compl\Sym\Equiv i=k$.
  % $((\Sonelem,i),(j,k))\in\Tensp{\Into}{\Wdiag}\Compl\Sym
  % \Equiv((i,\Sonelem),(j,k)) \in\Compl\Tensp{\Into}{\Wdiag}\Equiv
  % i=j$.
\end{Example}

\subsection{Differentiation in a canonically summable category}
Let $\cL$ be a resource category (see the beginning of
Section~\ref{sec:diff-struct-res-cat}) which is canonically summable.
Doubtlessly the following lemma is a piece of categorical folklore, it
relies only on the adjunction $\Tinto\Adj\Scfun$ and on the
functoriality of $\Excl\_$.
Let $\Ftunit_X\in\cL(X,\Scfun\Tinto X)$ and
$\Ftcounit_X\in\cL(\Tinto\Scfun X,X)$ be the unit and counit of this
adjunction. Let $\phi_X:\cL(\Excl{\Scfun X},\Scfun{\Excl X})$ be a natural
transformation, then we define a natural transformation
$\Stofdiff\phi_X\in\Tinto\Excl X\to\Excl{\Tinto X}$ as the following
composition of morphisms
\[
  \begin{tikzcd}
    \Tinto{\Excl X}\ar[r,"\Tinto\Excl{\Ftunit_X}"]
    &\Tinto\Excl{\Scfun\Tinto X}\ar[r,"\Tinto\phi_{\Tinto X}"]
    &\Tinto\Scfun\Excl{\Tinto X}\ar[r,"\Ftcounit_{\Excl{\Tinto X}}"]
    &\Excl{\Tinto X}\,.
  \end{tikzcd}
\]
Conversely given a natural transformation
$\psi_X\in\cL(\Tinto\Excl X,\Excl{\Tinto X})$ we define a natural
transformation $\Diffofst\psi_X\in\cL(\Excl{\Scfun X},\Scfun{\Excl X})$
as the following composition of morphisms
\[
  \begin{tikzcd}
    \Excl{\Scfun X}\ar[r,"\Ftunit_{\Excl{\Scfun X}}"]
    &\Scfun\Tinto\Excl{\Scfun X}\ar[r,"\Scfun\psi_{\Scfun X}"]
    &\Scfun\Excl{\Tinto\Scfun X}\ar[r,"\Scfun\Excl{\Ftcounit_X}"]
    &\Scfun\Excl X\,.
  \end{tikzcd}
\]

\begin{lemma}
  With the notations above, $\Diffofst{\Stofdiff\phi}=\phi$ and
  $\Stofdiff{\Diffofst\psi}=\psi$.
\end{lemma}
\begin{proof}
  Simple computation using the basic properties of adjunctions and the
  naturality of the various morphisms involved.
\end{proof}

\begin{lemma}\label{lemma:nat-Into-strength}
  Let $\Sdiffst_X\in\cL(\Tinto\Excl X,\Excl{\Tinto X})$
  be a natural transformation. The associated natural transformation
  $\Diffofst{\Sdiffst}_X\in\cL(\Excl{\Scfun X},\Scfun{\Excl X})$ satisfies
  \Daxchain{} iff the two following diagrams commute
  \[
    \begin{tikzcd}
      \Tinto{\Excl X}\ar[r,"\Sdiffst_X"]
      \ar[dr,swap,"\Tinto{\Der X}"]
      &\Excl{\Tinto X}\ar[d,"\Der{\Tinto X}"]\\
      &\Tinto X
    \end{tikzcd}
    \quad\quad
    \begin{tikzcd}
      \Tinto{\Excl X}\ar[rr,"\Sdiffst_X"]
      \ar[d,swap,"\Tinto{\Digg X}"]
      &&\Excl{\Tinto X}\ar[d,"\Digg{\Tinto X}"]\\
      \Tinto{\Excll X}\ar[r,"\Sdiffst_{\Excl X}"]
      &\Excl{\Tinto{\Excl X}}\ar[r,"\Excl{\Sdiffst_X}"]
      &\Excll{\Tinto X}
    \end{tikzcd}
  \]
  in other words $\Sdiffst_X$ is a co-distributive law
  $\Tinto\Excl X\to\Excl{\Tinto X}$. These conditions will be
  called~\Daxcchain.
\end{lemma}
% from appendix
%
\begin{proof}
  Consists of computations using naturality and adjunction
  properties. As an example, assume the second commutation and let us
  prove the second diagram of \Daxchain:
  \[
    \begin{tikzcd}
      \Excl{\Scfun X}\ar[rr,"\Diffofst{\Sdiffst_X}"]
      \ar[d,swap,"\Digg{\Scfun X}"]
      &&
      \Scfun{\Excl X}\ar[d,"\Scfun{\Digg X}"]\\
      \Excll{\Scfun X}\ar[r,"\Excl{\Diffofst{\Sdiffst_X}}"]
      &\Excl{\Scfun{\Excl X}}\ar[r,"\Diffofst{\Sdiffst_{\Excl X}}"]
      &\Scfun{\Excll X}
    \end{tikzcd}
  \]
  % \[
  %   \begin{tikzcd}
  %     \Excl{\Limplp\Into X}\ar[rr,"\Diffofst{\Sdiffst_X}"]
  %     \ar[d,swap,"\Digg{\Limpl\Into X}"]
  %     &&
  %     \Limpl\Into{\Excl X}\ar[d,"\Limpl\Into{\Digg X}"]\\
  %     \Excll{\Limplp\Into X}\ar[r,"\Excl{\Diffofst{\Sdiffst_X}}"]
  %     &\Excl{\Limplp\Into{\Excl X}}\ar[r,"\Diffofst{\Sdiffst_{\Excl X}}"]
  %     &\Limpl\Into{\Excll X}
  %   \end{tikzcd}
  % \]
  We have
  \begin{align*}
    (\Scfun\Digg X)\Compl\Diffofst{\Sdiffst_X}
    &=(\Scfun\Digg X)\Compl(\Scfun\Excl{\Ftcounit_X)}
      \Compl(\Scfun\Sdiffst_{\Scfun X})\Compl\Ftunit_{\Excl{\Scfun X}}\\
    &=(\Scfun\Excll{\Ftcounit_X})\Compl(\Scfun\Digg{\Tinto\Scfun X})
      \Compl(\Scfun\Sdiffst_{\Scfun X})\Compl\Ftunit_{\Excl{\Scfun X}}\\
    &=(\Scfun\Excll{\Ftcounit_X})\Compl(\Scfun\Excl{\Sdiffst_{\Scfun X}})
      \Compl(\Scfun\Sdiffst_{\Excl{\Scfun X}})
      \Compl(\Scfun\Tinto\Digg{\Scfun X})
      \Compl\Ftunit_{\Excl{\Scfun X}}\\
    &=(\Scfun\Excll{\Ftcounit_X})\Compl(\Scfun\Excl{\Sdiffst_{\Scfun X}})
      \Compl(\Scfun\Sdiffst_{\Excl{\Scfun X}})
      \Compl\Ftunit_{\Excll{\Scfun X}}
      \Compl\Digg{\Scfun X}\\
    &=(\Scfun\Excll{\Ftcounit_X})
      \Compl(\Scfun\Excl{\Sdiffst_{\Scfun X}})
      \Compl(\Scfun\Excl{\Ftcounit_{\Tinto\Excl{\Scfun X}}})
      \Compl(\Scfun\Excl{\Tinto\Ftunit_{\Excl{\Scfun X}}})
      \Compl(\Scfun\Sdiffst_{\Excl{\Scfun X}})
      \Compl\Ftunit_{\Excll{\Scfun X}}
      \Compl\Digg{\Scfun X}\\
    &=(\Scfun\Excll{\Ftcounit_X})
      \Compl(\Scfun\Excl{\Ftcounit_{\Excl{\Tinto\Scfun X}}})
      \Compl(\Scfun\Excl{\Tinto\Scfun\Sdiffst_{\Scfun X}})
      \Compl(\Scfun\Sdiffst_{\Scfun\Tinto\Excl{\Scfun X}})
      \Compl(\Scfun\Tinto\Excl{\Ftunit_{\Excl{\Scfun X}}})
      \Compl\Ftunit_{\Excll{\Scfun X}}
      \Compl\Digg{\Scfun X}\\
    &=(\Scfun\Excl{\Ftcounit_{\Excl{\Scfun X}}})
      \Compl(\Scfun\Excl{\Tinto\Scfun\Excl{\Ftcounit_X}})
      \Compl(\Scfun\Excl{\Tinto\Scfun\Sdiffst_{\Scfun X}})
      \Compl(\Scfun\Sdiffst_{\Scfun\Tinto\Excl{\Scfun X}})
      \Compl\Ftunit_{\Excl{\Scfun\Tinto\Excl{\Scfun X}}}
      \Compl\Excl{\Ftunit_{\Excl X}}\Compl\Digg{\Scfun X}\\
    &=(\Scfun\Excl{\Ftcounit_{\Excl{\Scfun X}}})
      \Compl(\Scfun\Excl{\Tinto\Scfun\Excl{\Ftcounit_X}})
      \Compl(\Scfun\Sdiffst_{\Scfun\Excl{\Tinto\Scfun X}})
      \Compl(\Scfun\Tinto\Excl{\Scfun\Sdiffst_{\Scfun X}})
      \Compl\Ftunit_{\Excl{\Scfun\Tinto\Excl{\Scfun X}}}
      \Compl\Excl{\Ftunit_{\Excl{\Scfun X}}}
      \Compl\Digg{\Scfun X}\\
    &=(\Scfun\Excl{\Ftcounit_{\Excl{\Scfun X}}})
      \Compl(\Scfun\Excl{\Tinto\Scfun\Excl{\Ftcounit_X}})
      \Compl(\Scfun\Sdiffst_{\Scfun\Excl{\Tinto\Scfun X}})
      \Compl(\Scfun\Tinto\Excl{\Scfun\Sdiffst_{\Scfun X}})
      \Compl\Ftunit_{\Excl{\Scfun\Tinto\Excl{\Scfun X}}}
      \Compl\Excl{\Ftunit_{\Excl{\Scfun X}}}
      \Compl\Digg{\Scfun X}\\
    &=(\Scfun\Excl{\Ftcounit_{\Excl{\Scfun X}}})
      \Compl(\Scfun\Sdiffst_{\Scfun\Excl X})
      \Compl(\Scfun\Tinto\Excl{\Scfun\Excl{\Ftcounit_X}})
      \Compl\Ftunit_{\Excl{\Scfun\Excl{\Tinto\Scfun X}}}
      \Compl(\Excl{}\Scfun\Sdiffst_{\Scfun X})
      \Compl\Excl{\Ftunit_{\Excl{\Scfun X}}}
      \Compl\Digg{\Scfun X}\\
    &=(\Scfun\Excl{\Ftcounit_{\Excl{\Scfun X}}})
      \Compl(\Scfun\Sdiffst_{\Scfun\Excl X})
      \Compl\Ftunit_{\Excl{\Scfun\Excl X}}
      \Compl(\Excl{\Scfun\Excl{\Ftcounit_X}})
      \Compl(\Excl{}\Scfun\Sdiffst_{\Scfun X})
      \Compl\Excl{\Ftunit_{\Excl{\Scfun X}}}
      \Compl\Digg{\Scfun X}\\
    &=\Diffofst\Sdiffst_{\Excl X}
      \Compl\Excl{\Diffofst\Sdiffst_X}
      \Compl\Digg{\Scfun X}
\end{align*}
  % \begin{align*}
  %   \Evlin\Compl\Tensp{\Limplp{\Into}{\Digg X}}{\Into}
  %   \Compl\Tensp{\Diffofst{\Sdiffst_X}}{\Into}
  %   &=\Digg X\Compl\Evlin\Compl\Tensp{\Diffofst{\Sdiffst_X}}{\Into}\\
  %   &=\Digg X\Compl\Excl\Evlin\Compl\Sdiffst_{\Limpl\Into X}\text{\quad by def.~of }\Diffofst{\Sdiffst_X}\\
  %   &=\Excll\Evlin\Compl\Digg{\Tens{\Limplp\Into X}{\Into}}
  %     \Compl\Sdiffst_{\Limpl\Into X}\\
  %   &=\Excll\Evlin\Compl\Excl{\Sdiffst_{\Limpl\Into X}}\Compl\Sdiffst_{\Excl{\Limplp\Into X}}\Compl\Tensp{\Digg{\Limpl\Into X}}\Into\\
  %   &=\Excl\Evlin\Compl\Excl{\Tensp{\Diffofst{\Sdiffst_X}}{\Into}}\Compl\Sdiffst_{\Excl{\Limplp\Into X}}\Compl\Tensp{\Digg{\Limpl\Into X}}\Into
  %     \text{\quad by def.~of }\Diffofst{\Sdiffst_X}\\
  %   &=\Excl\Evlin\Compl\Sdiffst_{\Limpl\Into{\Excl X}}\Compl\Tensp{\Excl{\Diffofst{\Sdiffst_X}}}{\Into}\Compl\Tensp{\Digg{\Limpl\Into X}}\Into
  %     \text{\quad by nat.~of }\Sdiffst\\
  %   &=\Evlin\Compl\Tensp{\Diffofst{\Sdiffst_{\Excl X}}}{\Into}\Compl\Tensp{\Excl{\Diffofst{\Sdiffst_X}}}{\Into}\Compl\Tensp{\Digg{\Limpl\Into X}}\Into
  %     \text{\quad by def.~of }\Diffofst{\Sdiffst_X}
  % \end{align*}
  % whence the announced commutation.
  The other computations are similar.
\end{proof}

% We know from Theorem~\ref{th:Sfun-mont} that $\Scfun$ is equipped with
% exactly one lax symmetric monoidal monoidal structure $(\Sin0,\Smont)$
% such that $\Sproj0\Compl\Smont=\Tens{\Sproj0}{\Sproj0}$ and
% $\Sproj1\Compl\Smont=\Tens{\Sproj0}{\Sproj1}+\Tens{\Sproj1}{\Sproj0}$. An
% easy verification shows that $\Smont$ is induced by $\Scmont$, more
% precisely, the morphism
% $\Smont_{X_0,X_1}\in\cL(\Tens{\Scfun X_0}{\Scfun
%   X_1},\Scfun(\Tens{X_0}{X_1}))$ is the curry transpose of the
% following morphism
% \[
%   \begin{tikzcd}
%     \Limplp\Into{X_0}\ITens\Limplp\Into{X_1}
%     \ITens\Into\ar[d,"\Id\ITens\Scmont"]\\[-0.4em]
%     \Limplp\Into{X_0}\ITens\Limplp\Into{X_1}\ITens\Into\ITens\Into
%     \ar[d,"\sigma"]\\[-0.4em]
%     \Limplp\Into{X_0}\ITens\Into\ITens\Limplp\Into{X_1}\ITens\Into
%     \ar[d,"\Tens\Evlin\Evlin"]\\[-0.4em]
%     \Tens{X_0}{X_1}
%   \end{tikzcd}
% \]
% where $\sigma$ is an iso uniquely defined using the SMC structure of $\cL$.

Let $\Sdiffst_X\in\cL(\Tens{\Excl X}{\Into},\Excl{\Tensp X\Into})$
satisfying~\Daxcchain. We introduce additional conditions. We keep
implicit some of the monoidal isos associated with $\ITens$ to
increase readability.

\begin{Axicond}{\Daxclocal}
  \[
    \begin{tikzcd}
      \Tens{\Excl X}{\Into}\ar[rr,"\Sdiffst_X"]
      &&\Excl{\Tensp X\Into}\\
      \Tens{\Excl X}\Sone
      \ar[u,"\Tens{\Excl X}{\Win0}"]
      \ar[r,"\Rightu_{\Excl X}"]
      &
      \Excl X\ar[r,"\Excl{\Inv{\Rightu_X}}"]
      &
      \Excl{\Tensp X\Sone}\ar[u,swap,"\Excl{\Tensp X{\Win0}}"]
    \end{tikzcd}
  \]
\end{Axicond}

\begin{Axicond}{\Daxclin}
  \[
    \begin{tikzcd}
      \Tens{\Excl X}\Into\ar[rr,"\Sdiffst_X"]
      \ar[d,swap,"\Tens{\Excl X}{\Proj0}"]
      &&
      \Excl{\Tensp X\Into}\ar[d,"\Excl{\Tensp X{\Proj0}}"]\\
      \Tens{\Excl X}\Sone
      \ar[r,"\Rightu_{\Excl X}"]
      &\Excl X\ar[r,"\Excl{\Inv{\Rightu_X}}"]
      &\Excl{\Tensp X\Sone}
    \end{tikzcd}
    \Treesep
    \begin{tikzcd}
      \Excl X\ITens\Into\ar[r,"\Tens{\Excl X}{\Scmont}"]
      \ar[d,swap,"\Sdiffst_X"]
      &
      \Excl X\ITens\Into\ITens\Into\ar[r,"\Tens{\Sdiffst_X}\Into"]
      &
      \Excl{\Tensp X\Into}\ITens\Into\ar[d,"\Sdiffst_{\Tens X\Into}"]\\
      \Excl{\Tensp X\Into}\ar[rr,"\Exclp{\Tens X{\Scmont}}"]
      &&
      \Exclp{X\ITens\Into\ITens\Into}
    \end{tikzcd}
  \]
\end{Axicond}

\begin{Axicond}{\Daxcwith}
  \footnotesize{
    \[
    \begin{tikzcd}
      \Tens\Sone\Into\ar[rr,"\Tens\Sone{\Proj 0}"]
      \ar[d,swap,"\Tens{\Seelyz}\Into"]
      &[-0.8em]&[-0.8em]
      \Sone\ar[d,"\Seelyz"]\\
      \Tens{\Excl\Top}{\Into}\ar[r,"\Sdiffst_\Top"]
      &
      \Exclp{\Tens\Top\Into}\ar[r,"\Excl 0"]
      &\Excl\Top
    \end{tikzcd}
    \Treesep
    \begin{tikzcd}
      \Excl{X_0}\ITens\Excl{X_1}\ITens\Into\ar[r,"\Tens\Id{\Scmont}"]
      \ar[d,swap,"\Tens\Seelyt\Into"]
      &[1.4em]
      \Excl{X_0}\ITens\Excl{X_1}\ITens\Into\ITens\Into
      \ar[r,"\Tens{\Sdiffst_{X_0}}{\Sdiffst_{X_1}}"]
      &[2.8em]
      \Exclp{\Tens{X_0}\Into}\ITens\Exclp{\Tens{X_1}\Into}
      \ar[d,"\Seelyt"]
      \\
      \Exclp{\With{X_0}{X_1}}\ITens\Into\ar[r,"\Sdiffst_{\With{X_0}{X_1}}"]
      &\Exclp{\Tens{\Withp{X_0}{X_1}}{\Into}}
      \ar[r,"\Excl{\Tuple{\Tens{\Proj0}{\Into},\Tens{\Proj1}{\Into}}}"]
      &\Exclp{\With{\Tensp{X_0}\Into}{\Tensp{X_0}\Into}}
    \end{tikzcd}
  \]}
\end{Axicond}

\begin{Axicond}{\Daxcschwarz}
  \[
    \begin{tikzcd}
      \Excl X\ITens\Into\ITens\Into\ar[r,"\Sdiffst_X\ITens\Into"]
      \ar[d,swap,"\Tens{\Excl X}{\Scflip}"]
      &\Exclp{\Tens X\Into}\ITens\Into\ar[r,"\Sdiffst_{\Tens X\Into}"]
      &\Exclp{X\ITens\Into\ITens\Into}
      \ar[d,"\Exclp{\Tens X\Scflip}"]\\
      \Excl X\ITens\Into\ITens\Into\ar[r,"\Sdiffst_X\ITens\Into"]
      &\Exclp{\Tens X\Into}\ITens\Into\ar[r,"\Sdiffst_{\Tens X\Into}"]
      &\Exclp{X\ITens\Into\ITens\Into}
    \end{tikzcd}
  \]
\end{Axicond}

\begin{theorem} %
  \label{th:sdiffst-sdiff}
  Let $\Sdiffst_X\in\cL(\Tens{\Excl X}{\Into},\Excl{\Tensp X\Into})$
  be a natural transformation. The two following conditions are equivalent.
  \begin{itemize}
  \item $\Sdiffst$ satisfies \Daxcchain, \Daxclocal, \Daxclin{},
    \Daxcwith{} and \Daxcschwarz.
  \item $\Diffofst\Sdiffst$ is a differentiation in $(\cL,\Scfun)$ (in
    the sense of Definition~\ref{def:Sfun-diff}).
  \end{itemize}
\end{theorem}
\begin{proof}
  Simple categorical computations: there is a simple direct
  correspondence between the conditions %
  \Daxcchain, \Daxclocal, \Daxclin{}, \Daxcwith{} and \Daxcschwarz{}
  on %
  $\Sdiffst$ and the conditions the conditions %
  \Daxchain, \Daxlocal, \Daxlin{}, \Daxwith{} and \Daxschwarz{} on %
  $\Diffofst\Sdiffst$ through the adjunction $\Tinto\Adj\Scfun$.
  % TODO ??????? Trivial et très chiant à écrire...
  %
  % As an example, let us prove the equivalence between the second
  % diagrams of the properties %
  % \Daxwith{} and \Daxcwith{}. Let us first assume that the second
  % diagram of \Daxcwith{} holds for all $X_0,X_1$. We have
  % \begin{align*}
  %   \Diffofst\Sdiffst_{\With{X_0}{X_1}}
  %   &=(\Scfun\Excl{\Ftcounit_{\With{X_0}{X_1}})}
  %     \Compl(\Scfun\Sdiffst_{\Scfun{\Withp{X_0}{X_1}}})
  %     \Compl\Ftunit_{\Excl{\Scfun{\Withp{X_0}{X_1}}}}
  % \end{align*}
  %
\end{proof}

\begin{definition}\label{def:diff-can-sum-cat}
  A \emph{differential canonically summable resource category} is a
  canonically summable resource category $\cL$ equipped with a natural
  transformation
  $\Sdiffst_X\in\cL(\Tens{\Excl X}{\Into},\Excl{\Tensp X\Into})$
  satisfying \Daxcchain, \Daxclocal, \Daxclin{}, \Daxcwith{} and
  \Daxcschwarz. Then we set $\Sdiff=\Diffofst\Sdiffst$.
\end{definition}

We show now that this differential structure boils down to a much
simpler one.

\subsection{A $\oc$-coalgebra structure on $\Into$ induced by a
  canonical differential structure}
Let %
$\Sdiffst_X\in\cL(\Tens{\Excl X}{\Into},\Excl{\Tensp X\Into})$ be a
natural transformation which satisfies the conditions of
Definition~\ref{def:diff-can-sum-cat}.

\begin{lemma}\label{lemma:sdiffst-seelyt}
  Given objects $X_0,X_1$ of $\cL$, the following diagram commutes
  \[
    \begin{tikzcd}
      \Excl{X_0}\ITens\Excl{X_1}\ITens\Into
      \ar[rr,"\Tens{\Excl{X_0}}{\Sdiffst_{X_1}}"]
      \ar[d,swap,"\Tens{\Seelyt_{X_0,X_1}}{\Into}"]
      &[1em]&
      \Tens{\Excl{X_0}}{\Excl{\Tensp{X_1}{\Into}}}
      \ar[d,"\Seelyt_{X_0,\Tens{X_1}{\Into}}"]
      \\
      \Tens{\Exclp{\With{X_0}{X_1}}}{\Into}
      \ar[r,"\Sdiffst_{\With{X_0}{X_1}}"]
      &
      \Excl{\Tensp{\Withp{X_0}{X_1}}{\Into}}
      \ar[r,"\Excl{q_0}"]
      &
      \Exclp{\With{X_0}{\Tensp{X_1}{\Into}}}
    \end{tikzcd}
  \]
  where %
  $q_0=\Tuple{\Rightu_{X_0}\Tensp{\Proj0}{\Proj0},\Tens{\Proj1}{\Into}}
  \in\cL(\Tens{\Withp{X_0}{X_1}}{\Into},\With{X_0}{\Tensp{X_1}{\Into}})$.
\end{lemma}
\begin{proof}
  Observe first that
  \(
    q_0=\Withp{\Rightu_{X_0}\Compl\Tensp{X_0}{\Proj0}}{\Tensp{X_1}{\Into}}
    \Compl
    \Tuple{\Tens{\Proj0}{\Into},{\Tens{\Proj0}{\Into}}}
  \).
  We have
  \begin{align*}
    \Excl{q_0}
    \Compl
    \Sdiffst_{\With{X_0}{X_1}}
    \Compl
    \Tensp{\Seelyt_{X_0,X_1}}{\Into}
    % &=\Excl{q_0}
    %   \Compl
    %   \Sdiffst_{\With{X_0}{X_1}}
    %   \Compl
    %   \Tensp{\Seelyt_{X_0,X_1}}{\Into}\\
    &=\Exclp{\With{\Rightu_{X_0}\Compl\Tensp{X_0}{\Proj0}}{\Tensp{X_1}{\Into}}}
      \Compl
      \Excl{\Tuple{\Tens{\Proj0}{\Into},{\Tens{\Proj0}{\Into}}}}
      \Compl
      \Sdiffst_{\With{X_0}{X_1}}
      \Compl
      \Tensp{\Seelyt_{X_0,X_1}}{\Into}\\
    &=\Exclp{\With{\Rightu_{X_0}\Compl\Tensp{X_0}{\Proj0}}{\Tensp{X_1}{\Into}}}
      \Compl{\Seelyt_{\Tens{X_0}{\Into},\Tens{X_1}{\Into}}}
      \Compl\Tensp{\Sdiffst_{X_0}}{\Sdiffst_{X_1}}
      \Compl\Sym_{2,3}
      \Compl(\Excl{X_0}\ITens\Excl{X_1}\ITens\Scmont)
  \end{align*}
  by \Daxcwith{}, and notice that %
  $\Sym_{2,3}\in\cL(\Excl{X_0}\ITens\Excl{X_0}\ITens\Into\ITens\Into,
  \Excl{X_0}\ITens\Into\ITens\Excl{X_1}\ITens\Into)$. %
  So we have
  \begin{align*}
    \Excl{q_0}
    \Compl
    \Sdiffst_{\With{X_0}{X_1}}
    \Compl
    \Tensp{\Seelyt_{X_0,X_1}}{\Into}
    &=\Seelyt_{X_0,\Tens{X_1}{\Into}}
      \Compl
      \Tensp{\Exclp{\Rightu_{X_0}\Compl\Tensp{X_0}{\Proj0}}}
      {\Excl{\Tensp{X_1}{\Into}}}
      \Compl
      \Compl\Tensp{\Sdiffst_{X_0}}{\Sdiffst_{X_1}}
      \Compl\Sym_{2,3}
      \Compl(\Excl{X_0}\ITens\Excl{X_1}\ITens\Scmont)\\
    &=\Seelyt_{X_0,\Tens{X_1}{\Into}}
      \Compl\Tensp{(\Rightu_{\Excl{X_0}}
      \Compl\Tensp{\Excl{X_0}}{\Proj0})}{\Sdiffst_{X_1}}
      \Compl\Sym_{2,3}
      \Compl(\Excl{X_0}\ITens\Excl{X_1}\ITens\Scmont)
    \text{\quad by \Daxclin{}}\\
    &=\Seelyt_{X_0,\Tens{X_1}{\Into}}
      \Compl\Rightu_{\Excl{X_0}\ITens\Excl{\Tensp{X_1}\Into}}
      \Compl(\Excl{X_0}\ITens\Sdiffst_{X_1}\ITens\Proj0)
      \Compl(\Excl{X_0}\ITens\Excl{X_1}\ITens\Scmont)\\
    &=\Seelyt_{X_0,\Tens{X_1}{\Into}}
      \Compl\Tensp{\Excl{X_0}}{\Sdiffst_{X_1}}
  \end{align*}
  where, in the last equation, we use the following commutation
  \[
    \begin{tikzcd}
      \Excl{X_0}\ITens\Excl{X_1}\ITens\Into
      \ar[r,"\Excl{X_0}\ITens\Excl{X_1}\ITens\Scmont"]
      \ar[rrd,swap,"\Tens{\Excl{X_0}}{\Sdiffst_{X_1}}"]
      &[2em]
      \Excl{X_0}\ITens\Excl{X_1}\ITens\Into\ITens\Into
      \ar[r,"\Excl{X_0}\ITens\Sdiffst_{X_1}\ITens\Proj0"]
      %\ar[d,"\Excl{X_0}\ITens\Excl{X_1}\ITens\Proj0"]
      &[3em]
      \Excl{X_0}\ITens\Excl{\Tensp{X_1}{\Into}}\ITens\Sone
      \ar[d,"\Rightu_{\Excl{X_0}\ITens\Excl{\Tensp{X_1}\Into}}"]
      \\
      &
      &
      \Tens{\Excl{X_0}}{\Excl{\Tensp{X_1}{\Into}}}
    \end{tikzcd}
  \]
  which results from the coneutrality $\Proj0$ for $\Scmont$.
\end{proof}

The next result will be technically useful in the sequel and has also
its own interest as it deals with differentiation with respect to a
tensor product, showing essentially that it boils down to
differentiation with respect to one of the components of the tensor
product.
\begin{theorem}\label{th:sdiffst-mon-tens}
  The following diagram commutes, for all objects $X_0$, $X_1$ of $\cL$.
  \[
    \begin{tikzcd}
      \Excl{X_0}\ITens\Excl{X_1}\ITens\Into
      \ar[r,"\Tens{\Excl{X_0}}{\Sdiffst_{X_1}}"]
      \ar[d,swap,"\Mont_{X_0,X_1}\ITens\Into"]
      &[1.6em]
      \Excl{X_0}\ITens\Excl{\Tensp{X_1}\Into}
      \ar[d,"\Mont_{X_0,\Tens{X_1}\Into}"]\\
      \Excl{\Tensp{X_0}{X_1}}\ITens\Into
      \ar[r,"\Sdiffst_{\Tens{X_0}{X_1}}"]
      &
      \Exclp{X_0\ITens X_1\ITens\Into}
    \end{tikzcd}
  \]
\end{theorem}
\begin{proof}
  We recall that %
  $\Mont_{X,Y}\in\cL(\Tens{\Excl X}{\Excl Y},\Excl{\Tensp XY})$ is
  defined as the following composition of morphisms
  \[
    \begin{tikzcd}
      \Tens{\Excl X}{\Excl Y}
      \ar[r,"\Seelyt_{X,Y}"]
      &
      \Excl{\Withp XY}
      \ar[r,"\Digg{\With XY}"]
      &[1em]
      \Excll{\Withp XY}
      \ar[r,"\Excl{\Invp{\Seelyt_{X,Y}}}"]
      &[1em]
      \Excl{\Tensp{\Excl X}{\Excl Y}}
      \ar[r,"\Excl{\Tensp{\Der X}{\Der Y}}"]
      &[2.6em]
      \Exclp{\Tens XY}
    \end{tikzcd}
  \]
  so that we have
  \begin{align*}
    \Mont_{X_0,\Tens{X_1}{\Into}}\Compl\Tensp{\Excl{X_0}}{\Sdiffst_{X_1}}
    =\Excl{\Tensp{\Der{X_0}}{\Der{\Tens{X_1}{\Into}}}}
    \Compl\Excl{\Invp{\Seelyt_{X_0,\Tens{X_1}\Into}}}
    \Compl\Digg{\With{X_0}{\Tensp{X_1}\Into}}
    \Compl\Seelyt_{X_0,\Tens{X_1}\Into}
    \Compl\Tensp{\Excl{X_0}}{\Sdiffst_{X_1}}
  \end{align*}
  We start rewriting the right hand expression. We introduce notations
  $f_0,f_1\dots$ for subexpressions. %
  By Lemma~\ref{lemma:sdiffst-seelyt} we have
  \begin{align*}
    f_0=\Seelyt_{X_0,\Tens{X_1}\Into}
    \Compl\Tensp{\Excl{X_0}}{\Sdiffst_{X_1}}
    =\Excl{q_0}
    \Compl\Sdiffst_{\With{X_0}{X_1}}
    \Compl\Tensp{\Seelyt_{X_0,X_1}}\Into
  \end{align*}
  where $q_0=\Tuple{\Rightu_{X_0}\Tensp{\Proj0}{\Proj0},\Tens{\Proj1}{\Into}}
  \in\cL(\Tens{\Withp{X_0}{X_1}}{\Into},\With{X_0}{\Tensp{X_1}{\Into}})$. %
  Then
  \begin{align*}
    f_1
    =\Digg{\With{X_0}{\Tensp{X_1}\Into}}\Compl f_0
    =\Excll{q_0}
    \Compl\Digg{\Tens{\Withp{X_0}{X_1}}\Into}
    \Compl\Sdiffst_{\With{X_0}{X_1}}
    \Compl\Tensp{\Seelyt_{X_0,X_1}}\Into\,.
  \end{align*}
  by naturality of $\Digg{}$. Next
  \begin{align*}
    f_2=\Digg{\Tens{\Withp{X_0}{X_1}}\Into}
    \Compl\Sdiffst_{\With{X_0}{X_1}}
    =\Excl{\Sdiffst_{\With{X_0}{X_1}}}
    \Compl\Sdiffst_{\Excl{\Withp{X_0}{X_1}}}
    \Compl\Tensp{\Digg{\With{X_0}{X_1}}}\Into
  \end{align*}
  by \Daxcchain. By Lemma~\ref{lemma:sdiffst-seelyt} again (applied under
  the functor $\Excl\_$) we have
  \begin{align*}
    f_3
    &=\Excl{\Invp{\Seelyt_{X_0,\With{X_1}\Into}}}
      \Compl f_1\\
    &=\Excl{\Invp{\Seelyt_{X_0,\With{X_1}\Into}}}
      \Compl\Excll{q_0}
      \Compl\Excl{\Sdiffst_{\With{X_0}{X_1}}}
      \Compl\Sdiffst_{\Excl{\Withp{X_0}{X_1}}}
      \Compl\Tensp{\Digg{\With{X_0}{X_1}}}\Into
      \Compl\Tensp{\Seelyt_{X_0,X_1}}\Into\\
    &=\Excl{\Tensp{\Excl{X_0}}{\Sdiffst_{X_1}}}
      \Compl\Excl{\Tensp{\Invp{\Seelyt_{X_0,X_1}}}\Into}
      \Compl\Sdiffst_{\Excl{\Withp{X_0}{X_1}}}
      \Compl\Tensp{\Digg{\With{X_0}{X_1}}}\Into
      \Compl\Tensp{\Seelyt_{X_0,X_1}}\Into\\
    &=\Excl{\Tensp{\Excl{X_0}}{\Sdiffst_{X_1}}}
      \Compl\Sdiffst_{\Tens{\Excl{X_0}}{\Excl{X_1}}}
      \Compl(\Excl{\Invp{\Seelyt_{X_0,X_1}}}\ITens\Into)
      \Compl\Tensp{\Digg{\With{X_0}{X_1}}}\Into
      \Compl\Tensp{\Seelyt_{X_0,X_1}}\Into\,.
  \end{align*}
  Finally we have
  \begin{align*}
    \Mont_{X_0,\Tens{X_1}{\Into}}
    &\Compl\Tensp{\Excl{X_0}}{\Sdiffst_{X_1}}
    =\Excl{\Tensp{\Der{X_0}}{\Der{\Tens{X_1}{\Into}}}}
    \Compl f_3\\
    &=\Excl{\Tensp{\Der{X_0}}{\Der{\Tens{X_1}{\Into}}}}
      \Compl\Excl{\Tensp{\Excl{X_0}}{\Sdiffst_{X_1}}}
      \Compl\Sdiffst_{\Tens{\Excl{X_0}}{\Excl{X_1}}}
      \Compl(\Excl{\Invp{\Seelyt_{X_0,X_1}}}\ITens\Into)
      \Compl\Tensp{\Digg{\With{X_0}{X_1}}}\Into
      \Compl\Tensp{\Seelyt_{X_0,X_1}}\Into\\
    &=\Excl{\Tensp{\Der{X_0}}{\Tens{\Der{X_1}}\Into}}
      % \Compl\Excl{\Tensp{\Excl{X_0}}{\Sdiffst_{X_1}}}
      \Compl\Sdiffst_{\Tens{\Excl{X_0}}{\Excl{X_1}}}
      \Compl(\Excl{\Invp{\Seelyt_{X_0,X_1}}}\ITens\Into)
      \Compl\Tensp{\Digg{\With{X_0}{X_1}}}\Into
      \Compl\Tensp{\Seelyt_{X_0,X_1}}\Into\\
    &\Textsep\text{by \Daxcchain}\\
    &=\Sdiffst_{\Tens{X_0}{X_1}}
      \Compl\Tensp{\Excl{\Tensp{\Der{X_0}}{\Der{X_1}}}}\Into
      % \Compl\Sdiffst_{\Tens{\Excl{X_0}}{\Excl{X_1}}}
      \Compl(\Excl{\Invp{\Seelyt_{X_0,X_1}}}\ITens\Into)
      \Compl\Tensp{\Digg{\With{X_0}{X_1}}}\Into
      \Compl\Tensp{\Seelyt_{X_0,X_1}}\Into\\
    &\Textsep\text{by naturality of }\Sdiffst\\
    &=\Sdiffst_{\Tens{X_0}{X_1}}
      \Compl\Tensp{\Mont_{X_0,X_1}}\Into
  \end{align*}
  by definition of $\Mont_{X_0,X_1}$.
\end{proof}
We define %
$\Sdiffca\in\cL(\Into,\Excl\Into)$ as the following
composition of morphisms
\[
  \begin{tikzcd}
    \Into\ar[r,"\Inv{\Leftu_\Into}"]
    &
    \Tens\Sone\Into\ar[r,"\Tens\Monz\Into"]
    &
    \Tens{\Excl\Sone}\Into\ar[r,"\Sdiffst_\Sone"]
    &
    \Excl{\Tensp\Sone\Into}\ar[r,"\Excl{\Leftu_\Into}"]
    &
    \Excl\Into\,.
  \end{tikzcd}
\]
\begin{lemma}\label{lemma:sdiffca-to-sdiffst}
  The following diagram commutes
  \[
    \begin{tikzcd}
      \Tens{\Excl X}\Into
      \ar[r,"\Tens{\Excl X}{\Sdiffca}"]
      \ar[rd,swap,"\Sdiffst_X"]
      &
      \Tens{\Excl X}{\Excl\Into}
      \ar[d,"\Mont_{X,\Into}"]
      \\
      &
      \Exclp{\Tens X\Into}
    \end{tikzcd}
  \]
\end{lemma}
\begin{proof}
  We have
  \begin{align*}
    \Mont_{X,\Into}
    \Compl\Tensp{\Excl X}{\Sdiffca}
    &=\Mont_{X,\Into}
      \Compl\Tensp{\Excl X}{\Excl{\Leftu_\Into}}
      \Compl\Tensp{\Excl X}{\Sdiffst_\Sone}
      \Compl\Tensp{\Excl X}{\Tens\Monz\Into}
      \Compl\Tensp{\Excl X}{\Inv{\Leftu_\Into}}\\
    &=\Exclp{X\ITens\Leftu_\Into}
      \Compl\Mont_{X,\Tens\Sone\Into}
      \Compl\Tensp{\Excl X}{\Sdiffst_\Sone}
      \Compl\Tensp{\Excl X}{\Tens\Monz\Into}
      \Compl\Tensp{\Excl X}{\Inv{\Leftu_\Into}}\\
    &=\Exclp{X\ITens\Leftu_\Into}
      \Compl\Sdiffst_{\Tens X\Sone}
      \Compl\Tensp{\Mont_{X,\Sone}}{\Into}
      \Compl\Tensp{\Excl X}{\Tens\Monz\Into}
      \Compl\Tensp{\Excl X}{\Inv{\Leftu_\Into}}
  \end{align*}
  by Theorem~\ref{th:sdiffst-mon-tens}. %
  We obtain the announced equation by %
  $\Mont_{X,\Sone}\Compl\Tensp{\Excl X}{\Monz}
  =\Excl{\Invp{\Rightu_X}}\Compl\Rightu_{\Excl X}$, %
  the naturality of $\Sdiffst$ and the fact that %
  $\Tens{\Rightu_X}{Y}
  =\Tens{X}{\Leftu_Y}\in\cL(X\ITens\Sone\ITens Y,\Tens XY)$.
\end{proof}
\begin{theorem}
  The morphism $\Sdiffca$ is a $\oc$-coalgebra structure on $\Into$. %
  Moreover the following commutations hold.
  \begin{Axicond}{\Daxcalocal}
    \[
      \begin{tikzcd}
        \Into\ar[r,"\Sdiffca"]
        &
        \Excl\Into
        \\
        \Sone\ar[r,"\Monz"]\ar[u,"\Win0"]
        &
        \Excl\Sone\ar[u,swap,"\Excl{\Win0}"]
      \end{tikzcd}
    \]
  \end{Axicond}
  \begin{Axicond}{\Daxcalin}
    \[
      \begin{tikzcd}
        \Into\ar[r,"\Sdiffca"]\ar[d,swap,"\Proj0"]
        &
        \Excl\Into\ar[d,"\Excl{\Proj0}"]
        \\
        \Sone\ar[r,"\Monz"]
        &
        \Excl\Sone
      \end{tikzcd}
      \Treesep
      % \begin{tikzcd}
      %   \Into\ar[r,"\Sdiffca"]\ar[d,swap,"\Scmont"]
      %   &
      %   \Excl\Into\ar[d,"\Contr\Into"]
      %   \\
      %   \Tens\Into\Into\ar[r,"\Tens\Sdiffca\Sdiffca"]
      %   &
      %   \Tens{\Excl\Into}{\Excl\Into}
      % \end{tikzcd}
      \begin{tikzcd}
        \Into\ar[rr,"\Sdiffca"]\ar[d,swap,"\Scmont"]
        &&
        \Excl\Into\ar[d,"\Excl\Scmont"]
        \\
        \Tens\Into\Into\ar[r,"\Tens\Sdiffca\Sdiffca"]
        &
        \Tens{\Excl\Into}{\Excl\Into}\ar[r,"\Mont"]
        &
        \Excl{\Tensp\Into\Into}
      \end{tikzcd}
    \]
  \end{Axicond}
\end{theorem}
\begin{proof}
  We have, using the fact that $(\Sone,\Monz)$ is an %
  $\oc$-coalgebra,
  \begin{align*}
    \Der\Into\Compl\Sdiffca
    &=\Der\Into
      \Compl\Excl\Leftu
      \Compl\Sdiffst_\Sone
      \Compl\Tensp\Monz\Into
      \Compl\Inv\Leftu\\
    &=\Leftu
      \Compl\Der{\Tens\Sone\Into}
      \Compl\Sdiffst_\Sone
      \Compl\Tensp\Monz\Into
      \Compl\Inv\Leftu\\
    &=\Leftu
      \Compl\Tensp{\Der\Sone}\Into
      \Compl\Tensp\Monz\Into
      \Compl\Inv\Leftu
    \text{\quad by \Daxcchain}\\
    &=\Leftu\Compl\Inv\Leftu=\Id
  \end{align*}
  and
  \begin{align*}
    \Digg\Into\Compl\Sdiffca
    &=\Digg\Into
      \Compl\Excl\Leftu
      \Compl\Sdiffst_\Sone
      \Compl\Tensp\Monz\Into
      \Compl\Inv\Leftu\\
    &=\Excll\Leftu
      \Compl\Excl{\Sdiffst_\Sone}
      \Compl\Sdiffst_{\Excl\Sone}
      \Compl\Tensp{\Digg\Sone}{\Into}
      \Compl\Tensp\Monz\Into
      \Compl\Inv\Leftu
      \text{\quad by \Daxcchain}\\
    &=\Excll\Leftu
      \Compl\Excl{\Sdiffst_\Sone}
      \Compl\Sdiffst_{\Excl\Sone}
      \Compl\Tensp{\Excl\Monz}{\Into}
      \Compl\Tensp\Monz\Into
      \Compl\Inv\Leftu\\
    &=\Excll\Leftu
      \Compl\Excl{\Sdiffst_\Sone}
      \Compl\Excl{\Tensp{\Monz}{\Into}}
      \Compl\Sdiffst_{\Sone}
      \Compl\Tensp\Monz\Into
      \Compl\Inv\Leftu
    \text{\quad by nat.~of }\Sdiffst
  \end{align*}
  and observe now that %
  $\Sdiffst_{\Sone}\Compl\Tensp\Monz\Into
  =\Excl{\Inv\Leftu}\Compl\Sdiffca\Compl\Leftu$. %
  It follows that
  \begin{align*}
    \Digg\Into\Compl\Sdiffca
    =\Excll\Leftu
      \Compl\Exclp{\Excl{\Inv\Leftu}\Compl\Sdiffca\Compl\Leftu}
      \Excl{\Inv\Leftu}\Compl\Sdiffca\Compl\Leftu
      \Compl\Inv\Leftu
    =\Excl\Sdiffca\Compl\Sdiffca\,.
  \end{align*}
  We have proven that $(\Into,\Sdiffca)$ is an $\oc$-coalgebra.

  Let us prove that %
  $\Win0\in\Em\cL((\Sone,\Monz),(\Into,\Sdiffca))$. We have
  \begin{align*}
    \Sdiffca\Compl\Win0
    &=\Excl\Leftu
      \Compl\Sdiffst_\Sone
      \Compl\Tensp\Monz\Into
      \Compl\Inv\Leftu
      \Compl\Win0\\
    &=\Excl\Leftu
      \Compl\Sdiffst_\Sone
      \Compl\Tensp\Monz\Into
      \Compl\Tensp{\Sone}{\Win0}
      \Compl\Inv\Leftu\\
    &=\Excl\Leftu
      \Compl\Sdiffst_\Sone
      \Compl\Tensp{\Excl\Sone}{\Win0}
      \Compl\Tensp\Monz\Sone
      \Compl\Inv\Leftu\\
    &=\Excl\Leftu_\Into
      \Compl\Excl{\Tensp{\Sone}{\Win0}}
      \Compl\Excl{\Inv{\Rightu_\Sone}}
      \Compl\Rightu_{\Excl\Sone}
      \Compl\Tensp\Monz\Sone
      \Compl\Inv{\Leftu_\Sone}
      \text{\quad by \Daxclocal}\\
    &=\Excl\Leftu_\Into
      \Compl\Excl{\Tensp{\Sone}{\Win0}}
      \Compl\Excl{\Inv{\Rightu_\Sone}}
      \Compl\Monz
      \Compl\Rightu_\Sone
      \Compl\Inv{\Leftu_\Sone}\\
    &=\Excl{\Win0}
      \Compl\Excl{\Leftu_\Sone}
      \Compl\Excl{\Inv{\Rightu_\Sone}}
      \Compl\Monz
      \Compl\Rightu_\Sone
      \Compl\Inv{\Leftu_\Sone}\\
    &=\Excl{\Win0}\Compl\Monz
  \end{align*}
  since $\Rightu_\Sone=\Leftu_\Sone$.

  Let us prove that %
  $\Proj0\in\Em\cL((\Into,\Sdiffca),(\Sone,\Monz))$. %
  We have
  \begin{align*}
    \Excl{\Proj0}\Compl\Sdiffca
    &=\Excl{\Proj0}
      \Compl\Excl{\Leftu_\Into}
      \Compl\Sdiffst_\Sone
      \Compl\Tensp\Monz\Into
      \Compl\Inv{\Leftu_\Into}\\
    &=\Excl{\Leftu_\Sone}
      \Compl\Excl{\Tensp\Sone{\Proj0}}
      \Compl\Sdiffst_\Sone
      \Compl\Tensp\Monz\Into
      \Compl\Inv{\Leftu_\Into}\\
    &=\Excl\Leftu_\Sone
      \Compl\Excl{\Inv{\Rightu_\Sone}}
      \Compl\Rightu_{\Excl\Sone}
      \Compl\Tensp{\Excl\Sone}{\Proj0}
      \Compl\Tensp\Monz\Into
      \Compl\Inv{\Leftu_\Into}
    \text{\quad by \Daxclin}\\
    &=\Excl\Leftu_\Sone
      \Compl\Excl{\Inv{\Rightu_\Sone}}
      \Compl\Rightu_{\Excl\Sone}
      \Compl\Tensp\Monz\Sone
      \Compl\Tensp{\Sone}{\Proj0}
      \Compl\Inv{\Leftu_\Into}\\
    &=\Excl\Leftu_\Sone
      \Compl\Excl{\Inv{\Rightu_\Sone}}
      \Compl\Monz
      \Compl\Rightu_\Sone
      \Compl\Inv{\Leftu_\Sone}
      \Compl\Proj0\\
    &=\Monz\Compl\Proj0
  \end{align*}
  Last we prove that %
  $\Scmont
  \in\Em\cL((\Into,\Sdiffca),(\Into,\Sdiffca)\ITens(\Into,\Sdiffca))$.
  We have
  \begin{align*}
    \Excl\Scmont\Compl\Sdiffca
    &=\Excl\Scmont
      \Compl\Excl{\Leftu_\Into}
      \Compl\Sdiffst_\Sone
      \Compl\Tensp\Monz\Into
      \Compl\Inv{\Leftu_\Into}\\
    &=\Excl{\Leftu_{\Tens\Into\Into}}
      \Compl\Excl{\Tensp\Sone\Scmont}
      \Compl\Sdiffst_\Sone
      \Compl\Tensp\Monz\Into
      \Compl\Inv{\Leftu_\Into}\\
    &=\Excl{\Leftu_{\Tens\Into\Into}}
      \Compl\Sdiffst_{\Tens\Sone\Into}
      \Compl\Tensp{\Sdiffst_\Sone}{\Into}
      \Compl\Tensp{\Excl\Sone}{\Scmont}
      \Compl\Tensp\Monz\Into
      \Compl\Inv{\Leftu_\Into}
      \text{\quad by \Daxclin}\\
    &=\Excl{\Leftu_{\Tens\Into\Into}}
      \Compl\Mont_{\Tens\Sone\Into,\Into}
      \Compl\Tensp{\Excl{\Tensp\Sone\Into}}{\Sdiffca}
      \Compl\Tensp{(\Mont_{\Sone,\Sone}\Compl\Tensp{\Excl\Sone}
        {\Sdiffca})}{\Into}
      \Compl\Tensp{\Excl\Sone}{\Scmont}
      \Compl\Tensp\Monz\Into
      \Compl\Inv{\Leftu_\Into}
      \text{\quad by Lemma~\ref{lemma:sdiffca-to-sdiffst}, twice}\\
    &=\Excl{\Leftu_{\Tens\Into\Into}}
      \Compl\Mont_{\Tens\Sone\Into,\Into}
      \Compl\Tensp{\Excl{\Tensp\Sone\Into}}{\Sdiffca}
      \Compl\Tensp{\Mont_{\Sone,\Into}}{\Into}
      \Compl\Tensp{\Tens{\Excl\Sone}{\Sdiffca}}{\Into}
      \Compl\Tensp{\Excl\Sone}{\Scmont}
      \Compl\Tensp\Monz\Into
      \Compl\Inv{\Leftu_\Into}\\
    &=\Excl{\Leftu_{\Tens\Into\Into}}
      \Compl\Mont_{\Tens\Sone\Into,\Into}
      \Compl\Tensp{\Mont_{\Sone,\Into}}{\Into}
      \Compl\Tensp{\Tens{\Excl\Sone}{\Excl\Into}}{\Sdiffca}
      \Compl\Tensp{\Tens{\Excl\Sone}{\Sdiffca}}{\Into}
      \Compl\Tensp{\Excl\Sone}{\Scmont}
      \Compl\Tensp\Monz\Into
      \Compl\Inv{\Leftu_\Into}\\
    &=\Excl{\Leftu_{\Tens\Into\Into}}
      \Compl\Mont_{\Tens\Sone\Into,\Into}
      \Compl\Tensp{\Mont_{\Sone,\Into}}{\Into}
      \Compl\Tensp{\Tens\Monz{\Excl\Into}}{\Excl\Into}
      \Compl\Tensp{\Tens{\Sone}{\Sdiffca}}{\Sdiffca}
      \Compl\Tensp{\Sone}{\Scmont}
      \Compl\Inv{\Leftu_\Into}
      \text{\quad by functoriality of }\ITens\\
    &=\Excl{\Leftu_{\Tens\Into\Into}}
      \Compl\Mont_{\Tens\Sone\Into,\Into}
      \Compl\Tensp{\Mont_{\Sone,\Into}}{\Into}
      \Compl\Tensp{\Tens\Monz{\Excl\Into}}{\Excl\Into}
      \Compl\Inv{\Leftu_{\Tens{\Excl\Into}{\Excl\Into}}}
      \Compl\Tensp{\Sdiffca}{\Sdiffca}
      \Compl{\Scmont}\\
    &=\Tensp{\Sdiffca}{\Sdiffca}
      \Compl{\Scmont}
  \end{align*}
  by standard properties of the lax monoidality structure %
  $(\Monz,\Mont)$ of $\Excl\_$.
\end{proof}

\subsection{From a coalgebra structure on $\Into$ to a canonical
  differential structure.} %
\label{sec:coalg-to-diff}
Assume now conversely that $\cL$ is a canonically summable resource
category where $\Into$ is exponentiable and that we have a morphism %
$\Sdiffca\in\cL(\Into,\Excl\Into)$. Then we can define a morphism %
$\Sdiffst_X\in\cL(\Tens{\Excl X}{\Into},\Exclp{\Tens X\Into})$ as the
following composition of morphisms.
\[
  \begin{tikzcd}
    \Tens{\Excl X}{\Into}\ar[r,"\Tens{\Excl X}{\Sdiffca}"]
    &
    \Tens{\Excl X}{\Excl\Into}\ar[r,"\Mont_{X,\Into}"]
    &
    \Excl{\Tensp X\Into}
  \end{tikzcd}
\]
This morphism is natural in $X$ by the naturality of $\Mont$.

\begin{theorem} %
  \label{th:sdiffca-to-sdiffst}
  If $\Sdiffca$ satisfies the following properties:
  \begin{enumerate}
  \item\label{cond:sdiffca-coalg} $(\Into,\Sdiffca)$ is a $\oc$-coalgebra
  \item\label{cond:sdiffca-local} \Daxcalocal{} 
  \item\label{cond:sdiffca-lin} and \Daxcalin{}
  \end{enumerate}
  then the natural transformation $\Sdiffst$ satisfies \Daxcchain,
  \Daxclocal, \Daxclin, \Daxcwith{} and \Daxcschwarz.
\end{theorem}
\begin{proof}
  \Proofcase %
  \Daxcchain. We have
  \begin{align*}
    \Der{\Tens{X}{\Into}}\Compl\Sdiffst_X
    &=\Der{\Tens{X}{\Into}}
      \Compl\Mont_{X,\Into}
      \Compl\Tensp{\Excl X}{\Sdiffca}\\
    &=\Tensp{\Der X}{\Der\Into}
      \Compl\Tensp{\Excl X}{\Sdiffca}\\
    &=\Tens{\Der X}{\Into}\text{\quad by \ref{cond:sdiffca-coalg}.}
  \end{align*}
  and
  \begin{align*}
    \Digg{\Tens X\Into}\Compl\Sdiffst_X
    &=\Digg{\Tens X\Into}
      \Compl\Mont_{X,\Into}
      \Compl\Tensp{\Excl X}{\Sdiffca}\\
    &=\Excl{\Mont_{X,\Into}}
      \Compl\Mont_{\Excl X,\Excl\Into}
      \Compl\Tensp{\Digg X}{\Digg\Into}
      \Compl\Tensp{\Excl X}{\Sdiffca}\\
    &=\Excl\Mont_{X,\Into}
      \Compl\Mont_{\Excl X,\Excl\Into}
      \Compl\Tensp{\Digg X}{(\Excl\Sdiffca\Compl\Sdiffca)}
    \text{\quad by \ref{cond:sdiffca-coalg}.}\\
    &=\Excl\Mont_{X,\Into}
      \Compl\Mont_{\Excl X,\Excl\Into}
      \Compl\Tensp{\Excll X}{\Excl{\Sdiffca}}
      \Compl\Tensp{\Digg X}{\Sdiffca}\\
    &=\Excl\Mont_{X,\Into}
      \Compl\Excl{\Tensp{\Excl X}{\Sdiffca}}
      \Compl\Mont_{\Excl X,\Into}
      \Compl\Tensp{\Digg X}{\Sdiffca}
      \text{\quad by naturality of }\Mont\\
    &=\Excl\Mont_{X,\Into}
      \Compl\Excl{\Tensp{\Excl X}{\Sdiffca}}
      \Compl\Mont_{\Excl X,\Into}
      \Compl\Tensp{\Excll X}{\Sdiffca}
      \Compl\Tensp{\Digg X}{\Into}\\
    &=\Excl{\Sdiffst_X}
      \Compl\Sdiffst_{\Excl X}
      \Compl\Tensp{\Digg X}{\Into}\\
    % &=\Excl{\Sdiffst_X}
    %   \Compl\Sdiffst_{\Excl X}
    %   \Compl\Tensp{\Digg X}{\Into}\\
  \end{align*}
  as required.

  \Proofcase %
  \Daxclocal. We have
  \begin{align*}
    \Sdiffst_X\Compl\Tensp{\Excl X}{\Win 0}
    &=\Mont_{X,\Into}
      \Compl\Tensp{\Excl X}{(\Sdiffca\Compl\Win0)}\\
    &=\Mont_{X,\Into}
      \Compl\Tensp{\Excl X}{(\Excl{\Win0}\Compl\Monz)}
    \text{\quad by \ref{cond:sdiffca-local}.}\\
    &=\Excl{\Tensp X{\Win 0}}
      \Compl\Mont_{X,\Sone}
      \Compl\Tensp{\Excl X}{\Monz}\\
    &=\Excl{\Tensp X{\Win 0}}
      \Compl\Excl{\Inv{\Rightu_X}}
      \Compl\Rightu_{\Excl X}
  \end{align*}
  by the properties of the lax monoidality $(\Mont,\Monz)$.

  \Proofcase %
  \Daxclin. We have
  \begin{align*}
    \Excl{\Tensp X{\Proj0}}\Compl\Sdiffst_X
    &=\Excl{\Tensp X{\Proj0}}
      \Compl\Mont_{X,\Into}
      \Compl\Tensp{\Excl X}{\Sdiffca}\\
    &=\Mont_{X,\Sone}
      \Compl\Tensp{\Excl X}{\Excl{\Proj 0}}
      \Compl\Tensp{\Excl X}{\Sdiffca}\\      
    &=\Mont_{X,\Sone}
      \Compl\Tensp{\Excl X}{\Monz}
      \Compl\Tensp{\Excl X}{\Proj0}
    \text{\quad by \ref{cond:sdiffca-lin}.}\\      
    &=\Excl{\Inv{\Rightu_X}}
      \Compl\Rightu_{\Excl X}
      \Compl\Tensp{\Excl X}{\Proj0}
  \end{align*}
  and
  \begin{align*}
    \Excl{\Tensp X{\Scmont}}\Compl\Sdiffst_X
    &=\Excl{\Tensp X{\Scmont}}
      \Compl\Mont_{X,\Into}
      \Compl\Tensp{\Excl X}{\Sdiffca}\\
    &=\Mont_{X,\Tens\Into\Into}
      \Compl\Tensp{\Excl X}{\Excl{\Scmont}}
      \Compl\Tensp{\Excl X}{\Sdiffca}\\      
    &=\Mont_{X,\Tens\Into\Into}
      \Compl\Tensp{\Excl X}{\Mont_{\Into,\Into}}
      \Compl\Tensp{\Excl X}{\Tens{\Sdiffca}{\Sdiffca}}
      \Compl\Tensp{\Excl X}{\Scmont}
      \text{\quad by \ref{cond:sdiffca-lin}.}\\      
    &=\Mont_{\Tens X\Into,\Into}
      \Compl\Tensp{\Mont_{X,\Into}}\Into
      \Compl\Tensp{\Excl X}{\Tens{\Sdiffca}{\Sdiffca}}
      \Compl\Tensp{\Excl X}{\Scmont}\\
    &=\Mont_{\Tens X\Into,\Into}
      \Compl(X\ITens\Into\ITens\Sdiffca)
      \Compl\Tensp{\Mont_{X,\Into}}\Into
      \Compl\Tensp{\Excl X}{\Tens{\Sdiffca}{\Into}}
      \Compl\Tensp{\Excl X}{\Scmont}\\
    &=\Sdiffst_{\Tens X\Into}
      \Compl\Tensp{\Sdiffst_X}\Into
      \Compl\Tensp{\Excl X}{\Scmont}\,.
   \end{align*}

   \Proofcase{} %
   \Daxcwith. By Proposition~\ref{prop:coalg-comon} we
   have %
   $\Proj0=\Weak\Into\Compl\Sdiffca$ and %
   $\Scmont=\Tensp{\Der\Into}{\Der\Into}\Compl\Contr\Into\Compl\Sdiffca$. We
   us these expressions in the next computations.

For the first diagram, we have
\begin{align*}
  \Seelyz
  \Compl\Leftu_\Into
  \Compl\Tensp{\One}{\Proj0}
  &=\Seelyz
    \Compl\Leftu_\Into
    \Compl\Tensp\Sone{\Weak\Into}
    \Compl\Tensp\Sone\Sdiffca\\
  &=\Excl 0
    \Compl\Mont_{\Top,\Into}
    \Compl\Tensp{\Seelyz}{\Excl\Into}
    \Compl\Tensp\Sone\Sdiffca
    \text{\quad by Lemma~\ref{lemma:seelyt-mont-commut}}\\
  &=\Excl 0
    \Compl\Mont_{\Top,\Into}
    \Compl\Tensp{\Excl\Top}{\Sdiffca}
    \Compl\Tensp{\Seelyz}{\Into}\\
  &=\Excl 0
    \Compl\Sdiffst_\Top
    \Compl\Tensp{\Seelyz}{\Into}\,.
\end{align*}
And the second one is proved by the following computation.
\begin{align*}
  \Seelyt_{\Tens{X_0}\Into,\Tens{X_1}\Into}
  \Compl&\Tensp{\Sdiffst_{X_0}}{\Sdiffst_{X_1}}
          \Compl\Sym_{2,3}
          \Compl(\Excl{X_0}\ITens\Excl{X_1}\ITens\Scmont)\\
        &=\Seelyt_{\Tens{X_0}\Into,\Tens{X_1}\Into}
          \Compl\Tensp{\Mont_{X_0,\Into}}{\Mont_{X_0,\Into}}
          \Compl\Sym_{2,3}\\
        &\Textsep
          (\Excl{X_0}\ITens\Excl{X_1}
          \ITens(\Sdiffca\Compl\Der\Into)
          \ITens(\Sdiffca\Compl\Der\Into))
          \Compl(\Excl{X_0}\ITens\Excl{X_1}\ITens\Contr\Into)
          \Compl(\Excl{X_0}\ITens\Excl{X_1}\ITens\Sdiffca)\\
        % &\Textsep
        %   \text{by Lemma~\ref{lemma:seelyt-mont-commut}}\\
        &=\Seelyt_{\Tens{X_0}\Into,\Tens{X_1}\Into}
          \Compl\Tensp{\Mont_{X_0,\Into}}{\Mont_{X_0,\Into}}
          \Compl\Sym_{2,3}\\
        &\Textsep
          (\Excl{X_0}\ITens\Excl{X_1}
          \ITens(\Der{\Excl\Into}\Compl\Excl\Sdiffca)
          \ITens(\Der{\Excl\Into}\Compl\Excl\Sdiffca))
          \Compl(\Excl{X_0}\ITens\Excl{X_1}\ITens\Contr\Into)
          \Compl(\Excl{X_0}\ITens\Excl{X_1}\ITens\Sdiffca)\\
        &=\Seelyt_{\Tens{X_0}\Into,\Tens{X_1}\Into}
          \Compl\Tensp{\Mont_{X_0,\Into}}{\Mont_{X_0,\Into}}
          \Compl\Sym_{2,3}\\
        &\Textsep
          (\Excl{X_0}\ITens\Excl{X_1}
          \ITens\Der{\Excl\Into}
          \ITens\Der{\Excl\Into})
          \Compl(\Excl{X_0}\ITens\Excl{X_1}\ITens\Contr{\Excl\Into})
          \Compl(\Excl{X_0}\ITens\Excl{X_1}\ITens(\Excl\Sdiffca\Compl\Sdiffca))\\
        &=\Seelyt_{\Tens{X_0}\Into,\Tens{X_1}\Into}
          \Compl\Tensp{\Mont_{X_0,\Into}}{\Mont_{X_0,\Into}}
          \Compl\Sym_{2,3}\\
        &\Textsep(\Excl{X_0}\ITens\Excl{X_1}
          \ITens\Der{\Excl\Into}
          \ITens\Der{\Excl\Into})
          \Compl(\Excl{X_0}\ITens\Excl{X_1}\ITens\Contr{\Excl\Into})
          \Compl(\Excl{X_0}\ITens\Excl{X_1}\ITens(\Digg\Into\Compl\Sdiffca))\\
        &\Textsep\text{because }\Sdiffca\text{ is a coalgebra structure}\\
        &=\Seelyt_{\Tens{X_0}\Into,\Tens{X_1}\Into}
          \Compl\Tensp{\Mont_{X_0,\Into}}{\Mont_{X_0,\Into}}
          \Compl\Sym_{2,3}\\
        &\Textsep
          (\Excl{X_0}\ITens\Excl{X_1}
          \ITens(\Der{\Excl\Into}\Compl\Digg\Into)
          \ITens(\Der{\Excl\Into}\Compl\Digg\Into))
          \Compl(\Excl{X_0}\ITens\Excl{X_1}\ITens\Contr\Into)
          \Compl(\Excl{X_0}\ITens\Excl{X_1}\ITens\Sdiffca)\\
        &=\Seelyt_{\Tens{X_0}\Into,\Tens{X_1}\Into}
          \Compl\Tensp{\Mont_{X_0,\Into}}{\Mont_{X_0,\Into}}
          \Compl\Sym_{2,3}
          \Compl(\Excl{X_0}\ITens\Excl{X_1}\ITens\Contr\Into)
          \Compl(\Excl{X_0}\ITens\Excl{X_1}\ITens\Sdiffca)\\
        &=\Excl{\Tuple{\Tens{\Proj0}\Into,\Tens{\Proj1}\Into}}
          \Compl\Mont_{\With{X_0}{X_1},\Into}
          \Compl\Tensp{\Seelyt_{X_0,X_1}}{\Excl\Into}
          \Compl(\Excl{X_0}\ITens\Excl{X_1}\ITens\Sdiffca)\\
        &\Textsep\text{by Lemma~\ref{lemma:seelyt-mont-commut}}\\
        &=\Excl{\Tuple{\Tens{\Proj0}\Into,\Tens{\Proj1}\Into}}
          \Compl\Mont_{\With{X_0}{X_1},\Into}
          \Compl\Tensp{\Excl{\Withp{X_0}{X_1}}}{\Sdiffca}
          \Compl\Tensp{\Seelyt_{X_0,X_1}}{\Into}\\
        &=\Excl{\Tuple{\Tens{\Proj0}\Into,\Tens{\Proj1}\Into}}
          \Compl\Sdiffst_{\With{X_0}{X_1}}
          \Compl\Tensp{\Seelyt_{X_0,X_1}}{\Into}
\end{align*}
as required.

\Proofcase %
\Daxcschwarz. We have
\begin{align*}
  \Excl{\Tensp X{\Sym_{\Into,\Into}}}
  \Compl\Sdiffst_{\Tens X\Into}
  \Compl\Tensp{\Sdiffst_X}\Into
  &=\Excl{\Tensp X{\Sym_{\Into,\Into}}}
    \Compl\Mont_{\Tens X\Into,\Into}
    \Compl\Tensp{\Excl{\Tensp X\Into}}{\Sdiffca}
    \Compl\Tensp{\Mont_{X,\Into}}{\Into}
    \Compl(\Excl X\ITens\Sdiffca\ITens\Into)\\
  &=\Excl{\Tensp X{\Sym_{\Into,\Into}}}
    \Compl\Monoidal^3_{X,\Into,\Into}
    \Compl(\Excl X\ITens\Sdiffca\ITens\Sdiffca)\\
  &=\Excl{\Tensp X{\Sym_{\Into,\Into}}}
    \Compl\Mont_{X,\Tens\Into\Into}
    \Compl\Tensp{\Excl X}{\Mont_{\Into,\Into}}
    \Compl(\Excl X\ITens\Sdiffca\ITens\Sdiffca)\\
  &=\Mont_{X,\Tens\Into\Into}
    \Compl\Tensp{\Excl X}{\Excl{\Sym_{\Into,\Into}}}
    \Compl\Tensp{\Excl X}{\Mont_{\Into,\Into}}
    \Compl(\Excl X\ITens\Sdiffca\ITens\Sdiffca)
    \text{\quad by naturality of }\Mont\\
  &=\Mont_{X,\Tens\Into\Into}
    \Compl\Tensp{\Excl X}{\Mont_{\Into,\Into}}
    \Compl\Tensp{\Excl X}{{\Sym_{\Excl\Into,\Excl\Into}}}
    \Compl(\Excl X\ITens\Sdiffca\ITens\Sdiffca)
    \text{\quad by symmetry of }\Mont\\
  &=\Mont_{X,\Tens\Into\Into}
    \Compl\Tensp{\Excl X}{\Mont_{\Into,\Into}}
    \Compl(\Excl X\ITens\Sdiffca\ITens\Sdiffca)
    \Compl\Tensp{\Excl X}{{\Sym_{\Into,\Into}}}
    \text{\quad by naturality of }\Sym
\end{align*}
which ends the proof of the theorem.
\end{proof}

We can summarize the results obtained in this section as follows.
\begin{theorem}
  Let $\cL$ be a resource category which is canonically summable. Then
  there is a bijective correspondence between
  \begin{itemize}
  \item the differential structures %
    $(\Sdiff_X)_{X\in\cL}$ on the canonical summability structure %
    $(\Scfun,\Sproj0,\Sproj1,\Ssum)$ of $\cL$
  \item and the $\oc$-coalgebra structures %
    $\Sdiffca$ on $\Into$ which satisfy \Daxcalocal{} and \Daxcalin.
  \end{itemize}
  When the second condition holds, the associated differentiation
  $\Sdiff_X\in\cL(\Excl{\Scfun X},\Scfun{\Excl X})$ is $\Curlin d$
  where $d$ is the following composition of morphisms.
  \[
    \begin{tikzcd}
      \Tens{\Exclp{\Limpl\Into X}}{\Into}
      \ar[r,"\Tens{\Exclp{\Limpl\Into X}}{\Sdiffca}"]
      &[2em]
      \Tens{\Exclp{\Limpl\Into X}}{\Excl\Into}
      \ar[r,"\Mont_{\Limpl\Into X,\Into}"]
      &[0.6em]
      \Exclp{\Tens{\Limplp\Into X}{\Into}}
      \ar[r,"\Excl\Evlin"]
      &[-1em]
      \Excl X
    \end{tikzcd}\,.
  \]
\end{theorem}

\begin{remark}
  This correspondence can certainly be made functorial, this is
  postponed to further work.
\end{remark}

\begin{theorem} %
  \label{th:lafont-differential}
  If $\cL$ is a Lafont resource category which is canonically summable
  then there is exactly one differential structure on the canonical
  summability structure of $\cL$.
\end{theorem}
\begin{proof}
  Since $(\Into,\Proj0,\Scmont)$ is a commutative comonoid, we know by
  Lemma~\ref{lemma:comon-coalg} that there is exactly one morphism %
  $\Sdiffca\in\cL(\Into,\Excl\Into)$ such that the following diagrams
  commute
  \[
    \begin{tikzcd}
      \Into \ar[r,"\Sdiffca"] \ar[rd,swap,"\Id"]
      & \Excl{\Into}
      \ar[d,"\Der{\Into}"]
      \\
      & \Into
    \end{tikzcd}
    \Treesep
    \begin{tikzcd}
      \Into \ar[r,"\Sdiffca"] \ar[rd,swap,"\Proj0"] &
      \Excl{\Into} \ar[d,"\Weak\Into"]
      \\
      & \Sone
    \end{tikzcd}
    \Treesep
    \begin{tikzcd}
      \Into \ar[r,"\Sdiffca"]
      \ar[d,swap,"\Scmont"]
      &[1em]
      \Excl{\Into}
      \ar[d,"\Contr{\Into}"]
      \\
      \Tens{\Into}{\Into}
      \ar[r,"\Tens{\Sdiffca}{\Sdiffca}"]
      & \Tens{\Excl{\Into}}{\Excl{\Into}}
    \end{tikzcd}
  \]
  By Theorem~\ref{th:lafont-weak-contr-coalg-morph} %
  $\Sdiffca$ satisfies \Daxcalin{} and hence we are left with proving
  \Daxcalocal. This readily follows from the bijective correspondence
  of Theorem~\ref{th:comon-coalg-lafont} and from the fact that
  $\Win0\in\Cm\cL(\Sone,(\Into,\Proj0,\Scmont))$. %
  Indeed $\Proj0\Compl\Win0=\Id_\Sone$ and %
  $\Scmont\Compl\Win0=\Tens{\Win0}{\Win0}$.
\end{proof}

\section{The differential structure of coherence spaces} %
\label{sec:coh-diff-str} %
Equipped with the multiset exponential introduced in
Section~\ref{sec:coh-diff} it is well known that $\COH$ is a Lafont
resource category as observed initially by Van de Wiele (unpublished,
see~\cite{Mellies09}). Since $\COH$ is canonically summable, we
already know that it has a unique differential structure by
Theorem~\ref{th:lafont-differential}. We will show that we retrieve in
that way the differential structure outlined in
Section~\ref{sec:coh-diff}.

Remember that $\Into=\With\Sone\Sone$ so that %
$\Web\Into=\{0,1\}\times\{(0,1)\}$ with %
$\Coh\Into{(i,j)}{(i',j')}$ for each $i,j,i',j'\in\{0,1\}$.  The
comonoid structure of %
$\Into=\With\Sone\Sone$ is given by %
$\Proj0=\{(0,\Sonelem)\}\in\COH(\Into,\Sone)$ and %
$\Scmont=\{(0,(0,0)),(1,(1,0)),(1,(0,1))\}\in\COH(\Into,\Tens\Into\Into)$. The
$n$-ary comultiplication of this comonoid is %
$\Scmontn n\in\COH(\Into,\Tot\Into n)$ given by
\begin{align*}
  \Scmontn n
  &=\{(0,(0,\dots,0))\}\cup
    \{(1,(\overbrace{0,\dots,0}^{k-1},1,\overbrace{0,\dots,0}^{n-k}))
    \St k\in\{1,\dots,n\}\}
\end{align*}
The unique $\Sdiffca\in\COH(\Into,\Excl\Into)$ specified by
Theorem~\ref{th:lafont-differential} is given by
\begin{align*}
  \Sdiffca
  =\{(0,k\Mset 0)\St k\in\Nat\}\cup\{(1,k\Mset 0+\Mset 1)\St k\in\Nat\}
\end{align*}
and then the associated natural transformation %
$\Sdiff_E\in\COH(\Excl{\Scfun E},)$ is $\Curlin d$ where $d$ is the
following composition of morphisms:
\[
  \begin{tikzcd}
    \Tens{\Exclp{\Limpl\Into E}}{\Into}
    \ar[r,"\Tens{\Exclp{\Limpl\Into E}}{\Sdiffca}"]
    &[2em]
    \Tens{\Exclp{\Limpl\Into E}}{\Excl\Into}
    \ar[r,"\Mont_{\Limpl\Into E,\Into}"]
    &[0.6em]
    \Exclp{\Tens{\Limplp\Into E}{\Into}}
    \ar[r,"\Excl\Evlin"]
    &[-1em]
    \Excl E
  \end{tikzcd}\,.
\]
Since %
$\Mont_{E_0,E_1}\in\COH(\Tens{\Excl{E_0}}{\Excl{E_1}},\Exclp{\Tens{E_0}{E_1}})$
is given by
\begin{multline*}
  \Mont_{E_0,E_1}
  =\{((\Mset{\List a1n},\Mset{\List b1n}),\Mset{(a_1,b_1),\dots,(a_n,b_n)})
  \St\\
  \Mset{\List a1n}\in\Web{\Excl{E_0}}\text{ and }
  \Mset{\List b1n}\in\Web{\Excl{E_1}}\}
\end{multline*}
and since
\begin{align*}
  \Excl\Evlin
  =\{(\Mset{((i_1,a_1),i_1),\dots,((i_n,a_n),i_n)},\Mset{\List a1n})\St
  \Mset{\List a1n}\in\Web{\Excl E}\text{ and }\List i1n\in\{0,1\}\}
\end{align*}
we have
\begin{multline*}
  d=\{((\Mset{(0,a_1),\dots,(0,a_n)},0),\Mset{\List a1n})\St
     \Mset{\List a1n}+\Mset a\in\Web{\Excl E}\}\\
   \cup\{((\Mset{(0,a_1),\dots,(0,a_n),(1,a)},1),
   \Mset{\List a1n}+\Mset a)\St\\
     \Mset{\List a1n}+\Mset a\in\Web{\Excl E}
     \text{ and }a\notin\{\List a1n\}\}\,.
\end{multline*}
The proviso that $a\notin\{\List a1n\}$ arises from uniformity of the
exponential: we must have
\[
  \{(0,a_1),\dots,(0,a_n),(1,a)\}\in\COH(\Limpl\Into E)\,.
\]
Finally, upon
identifying $\Web{\Excl{\Limpl\Into E}}$ with
\begin{align*}
  \{(m_0,m_1)\in\Web{\Excl E}^2\St m_0+m_1\in\Web{\Excl E}
  \text{ and }\Supp{m_0}\cap\Supp{m_1}=\emptyset\}
\end{align*}
we get
\begin{multline*}
  \Sdiff_E=\{((m_0,\Emptymset),(0,m_0))\St m_0\in\Web{\Excl E}\}\\
  \cup\{((m_0,\Mset a),(1,m_0+\Mset a))
  \St m_0+\Mset a\in\Web{\Excl E}\text{ and }a\notin \Supp{m_0}\}
\end{multline*}
which is exactly the definition announced in
Equation~\Eqref{eq:def-sdiff-coh}. The fact that this is a natural
transformation satisfying all the commutations required to turn $\COH$
into a differential summable category results from
Theorem~\ref{eq:def-sdiff-coh} and Theorem~\ref{th:sdiffst-sdiff}.

\subsection{Differentiation in non-uniform coherence spaces}
\label{sec:nucs-diff}
In Remark~\ref{rk:diff-coh-uniform} we have pointed out that
the uniform definition of $\Excl E$ in coherence spaces makes our
differentials ``too thin'' in general although they
are non trivial and satisfy all the required rules of the differential
calculus.  We show briefly how this situation can be
remedied using non-uniform coherence spaces.

A non-uniform coherence space (NUCS) is a triple
$E=(\Web E,\Scoh{E}{}{},\Sincoh{E}{}{})$ where $\Web E$ is a set and
$\Scoh E{}{}$ and $\Sincoh E{}{}$ are two \emph{disjoint} binary
symmetric relations on $\Web E$ called \emph{strict coherence} and
\emph{strict incoherence}. The important point of this definition is
not what is written but what is not: contrarily to usual coherence
spaces \emph{we do not require} the complement of the union of these
two relations to be the diagonal: it can be any (of course symmetric)
binary relation on $\Web E$ that we call \emph{neutrality} and denote
as $\Neu E{}{}$ (warning: it needs not even be an equivalence
relation!). Then we define coherence as
$\mathord{\Coh E{}{}}=\mathord{\Scoh E{}{}}\cup\mathord{\Neu E{}{}}$
and incoherence
$\mathord{\Incoh E{}{}}=\mathord{\Sincoh E{}{}}\cup\mathord{\Neu
  E{}{}}$ and any pair of relations among these 5 (with suitable
relation between them such as
$\mathord{\Neu E{}{}}\subseteq\mathord{\Incoh E{}{}}$), apart from the
trivially complementary ones
$(\mathord{\Sincoh E{}{}},\mathord{\Coh E{}{}})$ and
$(\mathord{\Scoh E{}{}},\mathord{\Incoh E{}{}})$, are sufficient to
define such a structure.

Cliques are defined as usual:
$\Cl E=\{x\subseteq\Web E\St \forall a,a'\in x\ \Coh E a{a'}\}$. Then
$(\Cl E,\mathord\subseteq)$ is a cpo (a dI-domain actually) but now
there can be some $a\in\Web E$ such that $\Sincoh Eaa$, and hence
$\Eset a\notin\Cl E$ (we show below that this really happens). Given
NUCS $E$ and $F$ we define $\Limpl EF$ by
$\Web{\Limpl EF}=\Web E\times\Web F$ and:
$\Coh{\Limpl EF}{(a_0,b_0)}{(a_1,b_1)}$ if
$\Coh E{a_0}{a_1}\Implies(\Coh F{b_0}{b_1}\text{ and }\Neu
F{b_0}{b_1}\Implies\Neu E{a_0}{a_1})$ and
$\Neu{\Limpl EF}{(a_0,b_0)}{(a_1,b_1)}$ if $\Neu E{a_0}{a_1}$ and
$\Neu F{b_0}{b_1}$. Then we define a category $\NCOH$ by
$\NCOH(E,F)=\Cl{\Limpl EF}$, taking the diagonal relations as
identities and ordinary composition of relations as composition of
morphisms.

This is a cartesian SMCC with tensor product given by
$\Web{\Tens{E_0}{E_1}}=\Web{E_0}\times\Web{E_1}$ and
$\Coh{\Tens{E_0}{E_1}}{(a_{00},a_{01})}{(a_{10},a_{11})}$ if
$\Coh{E_j}{a_{0j}}{a_{1j}}$ for $j=0,1$, and
$\Neu{\Tens{E_0}{E_1}}{}{}$ is defined similarly; the unit is $\Sone$
with $\Web\Sone=\Eset\Sonelem$ and $\Neu\Sone\Sonelem\Sonelem$ %
(so that $\Orth\Sone=\Sone$ meaning that the model satisfies a strong
form of the MIX rule of LL). The object of linear morphisms from $E$
to $F$ is of course $\Limpl EF$ and $\NCOH$ is $\ast$-autonomous with
$\Sone$ as dualizing object. The dual $\Orth E$ is given by
$\Web{\Orth E}=\Web E$,
$\mathord{\Scoh{\Orth E}{}{}}=\mathord{\Sincoh E{}{}}$ and
$\mathord{\Sincoh{\Orth E}{}{}}=\mathord{\Scoh E{}{}}$. The cartesian
product $\Bwith_{i\in I}E_i$ of a family $(E_i)_{i\in I}$ of NUCS is
given by
$\Web{\Bwith_{i\in I}E_i}=\cup_{i\in I}\Eset i\times\Web{E_i}$ with
$\Neu{\Bwith_{i\in I}E_i}{(i_0,a_0)}{(i_1,a_1)}$ if $i_0=i_1=i$ and
$\Neu{E_i}{a_0}{a_1}$, and
$\Coh{\Bwith_{i\in I}E_i}{(i_0,a_0)}{(i_1,a_1)}$ if
$i_0=i_1=i\Implies\Coh{E_i}{a_0}{a_1}$. We do not give the definition
of the operations on morphisms as they are the most obvious ones (the
projections of the product are the relations
$\Proj i=\{((i,a),a)\St i\in I\text{ and }a\in\Web{E_i}\}$). Notice
that in the object $\Bool=\Plus\Sone\Sone=\Orth{(\With\One\One)}$, the
two elements $0,1$ of the web satisfy $\Sincoh\Bool 01$ so that
$\{0,1\}\notin\Cl{\Bool}$ which is expected in a model of
deterministic computations.

We come to the most interesting feature of this model, which is the
possibility of defining a \emph{non-uniform} exponential $\Excl E$; we
choose here the one of~\cite{Boudes11} which is the free exponential
(so that $\NCOH$ is a Lafont resource category).
One sets $\Web{\Excl E}=\Mfin{\Web E}$ (without any uniformity
restrictions), $\Coh{\Excl E}{m_0}{m_1}$ if
$\forall a_0\in\Supp{m_0},a_1\in\Supp{m_1}\ \Coh{E}{a_0}{a_1}$, and
$\Neu{\Excl E}{m_0}{m_1}$ if $\Coh{\Excl E}{m_0}{m_1}$ and
$m_j=\Mset{a_{j0},\dots,a_{jn}}$ (for $j=0,1$) with
$\forall i\in\Eset{1,\dots,n}\ \Neu E{a_{0i}}{a_{1i}}$ (in particular
$m_0$ and $m_1$ must have the same size). Observe that
$\Mset{0,1}\in\Web{\Excl\Bool}$ and that
$\Sincoh{\Excl\Bool}{\Mset{0,1}}{\Mset{0,1}}$.  The action of this
functor on morphisms is defined as in the relational model of LL: if
$s\in\NCOH(E,F)$ then
$\Excl s=\Eset{(\Mset{\List a1n},\Mset{\List b1n})\St n\in\Nat\text{
    and }\forall i\ (a_i,b_i)\in s}\in\NCOH(\Excl E,\Excl F)$.

The object $\Into=\With\Sone\Sone$ is characterized by %
$\Web\Into=\{0,1\}$ and $\Scoh\Into01$, %
$\Neu\Into ii$ for $i\in\Web\Into$. %
The injections $\Win0,\Win1\in\NCOH(\Sone,\Into)$ are given by %
$\Win i=\{(\Sonelem,i)\}$ and are clearly jointly epic.
Two cliques $x_0,x_1\in\Cl E$ are summable if there is
$x\in\Cl{\Limpl\Into E}$ such that $x_i=x\Compl\Win i$, that is, if %
$\Eset 0\times x_0\cup\Eset 1\times x_1\in\Cl{\Limpl\Into E}$ which means that
\begin{align*}
  \forall a_0\in x_0,\, a_1\in x_1\quad \Scoh E{a_0}{a_1}\,.
\end{align*}
This implies $x_0\cup x_1\in\Cl E$ but not $x_0\cap x_1=\emptyset$
since we can have $\Scoh Eaa$ in a non-uniform coherence space.
It follows that the condition \Saxwit{} holds: %
let $x_{ij}\in\Cl E$ for $i,j\in\Eset{0,1}$ and assume that %
$x_{i0},x_{i1}$ are summable for $i=0,1$ and that moreover %
$x_{00}\cup x_{01},x_{10}\cup x_{11}$ are summable. 
Given $a_{ij}\in x_{ij}$ for $i,j\in\Eset{0,1}$, the two first
assumptions imply that %
$\Scoh E{a_{i0}}{a_{i1}}$ for $i\in\Eset{0,1}$ and the second
condition implies that %
$\Scoh E{a_{0i}}{a_{1j}}$ for $i,j\in\Eset{0,1}$. %
Finally, if $(i,j)\not=(i',j')\in\Eset{0,1}^2$ then %
$\Scoh E{a_{ij}}{a_{i'j'}}$ which implies that %
$\Eset 0\times x_{00}\cup\Eset 1\times x_{01}$ and %
$\Eset 0\times x_{10}\cup\Eset 1\times x_{11}$ are summable in %
$\Limpl\Into E$.

The comonoid structure $(\Proj0,\Scmont)$ is exactly the same as in
$\COH$ and therefore the morphism %
$\Sdiffca\in\NCOH(\Into,\Excl\Into)$ (whose existence and properties
result from the fact that $\NCOH$ is Lafont) is defined exactly as in
$\COH$:
\begin{align*}
  \Sdiffca
  =\{(0,k\Mset 0)\St k\in\Nat\}\cup\{(1,k\Mset 0+\Mset 1)\St k\in\Nat\}
\end{align*}

The functor $\Scfun$ can be described as follows:
$\Web{\Scfun E}=\{0,1\}\times\Web E$ and
$\Neu{\Scfun E}{(i_0,a_0)}{(i_1,a_1)}$ if $i_0=i_1$ and
$\Neu E{a_0}{a_1}$, and $\Coh{\Scfun E}{(i_0,a_0)}{(i_1,a_1)}$ if
($\Coh E{a_0}{a_1}$ and $\Neu E{a_0}{a_1}\Implies i_0=i_1$). Given
$s\in\cL(E,F)$ we have
$\Scfun s=\{((i,a),(i,b))\St i\in\Eset{0,1}\text{ and }(a,b)\in s\}$.
By the same computation as in $\COH$ (but now without the uniformity
restrictions of $\COH$) we get that
\begin{multline*}
  \Sdiff_E=\{((m_0,\Emptymset),(0,m_0))\St m_0\in\Mfin{\Web E}\}
  \cup\{((m_0,\Mset a),(1,m_0+\Mset a)) \St m_0+\Mset a\in\Mfin{\Web E}\}
\end{multline*}
which is in $\NCOH(\Excl{\Scfun E},\Scfun{\Excl E})$ and satisfies all
the required properties by Theorem~\ref{eq:def-sdiff-coh} and
Theorem~\ref{th:sdiffst-sdiff}.

\begin{remark}
  This means that the issue with Girard's uniform coherence spaces
  with respect to differentiation that we explained in
  Remarks~\ref{rk:diff-coh-multi} and~\ref{rk:diff-coh-uniform}
  disappears in the non-uniform coherence space setting, at least if
  we use Boude's exponentials
  % \footnote{Concerning the Bucciarelli-Ehrhard exponential, we do
  % not know, we should simply check whether the $\Sdiffst_X$
  % definition given here is also a clique for this exponential.},
  so that any morphism will coincide with its Taylor expansion in this
  model. This non-uniform model preserves the main feature of
  coherence spaces, namely that in the type $\Bool$ for instance, the
  only possible values are $\mathsf{true}$ and $\mathsf{false}$ (and
  not the non-deterministic superposition of these values) as we have
  seen above with the description of $\Plus\Sone\Sone$.
\end{remark}

\begin{remark}
  The category $\REL$ of sets of relation, being a model of
  differential linear logic, is a special case of summable
  differential resource category. That model is actually \emph{exactly
    the same} as $\NUCS$ where objects are stripped from their
  coherence structure: the logical constructs in $\REL$ coincide with
  the constructs we perform on the webs of the objects of $\NUCS$. For
  instance, given a set $X$, the object $\Excl X$ in $\REL$ is simply
  $\Mfin X$. And similarly for the operation on morphisms: as
  constructions on relations, they are exactly the same as in
  $\NUCS$. This identification extends even to $\Sdiff_X$. So one of
  the outcomes of this paper is the fact that the constructions of
  differential linear logic in $\REL$ are compatible with the
  coherence structure of $\NUCS$, \emph{if we are careful enough with
    morphism addition}. This is all the point of our categorical
  axiomatization to explain what this carefulness means.
\end{remark}

\section{Summability in a SMCC}\label{sec:SMCC-summability}
Assume now that $\cL$ is a summable resource category which is closed
with respect to its monoidal product $\ITens$, so that $\Kl\cL$ is
cartesian closed. We use $\Limpl XY$ for the internal hom object and
$\Evlin\in\cL(\Tens{(\Limpl XY)}{X},Y)$ for the evaluation
morphism. If $f\in\cL(\Tens ZX,Y)$ we use $\Curlin f$ for its
transpose $\in\cL(Z,\Limpl XY)$.

We can define a natural morphism
$\Sstrc=\Curlin{((\Sfun\Evlin)\Compl\Sstrs_{\Limpl XY,X}{})}
\in\cL(\Sfun(\Limpl XY),\Limpl X{\Sfun Y})$ where
$\Evlin\in\cL(\Tens{\Limplp XY}X,Y)$.

\begin{lemma}\label{lemma:Sstrc-sproj-ssum}
  We have $\Limplp X{\Sproj i}\Compl\Sstrc=\Sproj i$ for $i=0,1$ and
  $\Limplp X{\Ssum_Y}\Compl\Sstrc=\Ssum_{\Limpl XY}$.
\end{lemma}
\begin{proof}
  The first two equations come from the fact that
  $\Sproj i\Compl\Sstrs=\Tens{\Sproj i}{X}$. The last one results from
  Lemma~\ref{lemma:Sstr-sum}.
\end{proof}

Then we introduce a further axiom, required in the case of an SMCC. Its
intuitive meaning is that two morphisms $f_0,f_1$ are summable if they
map any element to a pair of summable elements, and that their sum is
computed pointwise.

\begin{Axicond}{\Saxfun}
  The morphism $\Sstrc$ is an iso.
\end{Axicond}

\begin{lemma}
  If \Saxfun{} holds then $f_0,f_1\in\cL(\Tens ZX,Y)$ are summable iff
  $\Curlin{f_0}$ and $\Curlin{f_1}$ are summable. Moreover when this
  property holds we have
  $\Curlin{(f_0+f_1)}=\Curlin{f_0}+\Curlin{f_1}$. 
\end{lemma}
\begin{proof}
  Assume that $f_0,f_1$ are summable so that we have the witness
  $\Stuple{f_0,f_1}\in\cL(\Tens ZX,\Sfun Y)$ and hence
  $\Curlin\Stuple{f_0,f_1}\in\cL(Z,\Limpl X{\Sfun Y})$, so let
  $h=\Inv{(\Sstrc)}\Compl\Curlin\Stuple{f_0,f_1}\in\cL(Z,\Sfun{\Limplp
    XY})$. By Lemma~\ref{lemma:Sstrc-sproj-ssum} we have
  $\Sproj i\Compl h=\Limplp\Into{\Sproj
    i}\Compl\Curlin\Stuple{f_0,f_1}=\Curlin{f_i}$ for
  $i=0,1$. Conversely if $\Curlin{f_0},\Curlin{f_1}$ are summable we
  have the witness
  $\Stuple{\Curlin{f_0},\Curlin{f_1}}\in\cL(Z,\Sfun\Limplp XY)$ and
  hence
  $\Sstrc\Compl\Stuple{\Curlin{f_0},\Curlin{f_1}}\in\cL(Z,\Limpl
  X{\Sfun Y})$ so that
  $g=\Evlin\Tensp{(\Sstrc\Compl\Stuple{\Curlin{f_0},\Curlin{f_1}})}{X}
  \in\cL(\Tens ZX,Y)$. Then by naturality of $\Evlin$ and by
  Lemma~\ref{lemma:Sstrc-sproj-ssum} we get $\Sproj i\Compl g=f_i$ for
  $i=0,1$ and hence $f_0,f_1$ are summable.

  Assume that these equivalent properties hold so that
  $\Stuple{\Curlin{f_0},\Curlin{f_1}}
  =\Inv{(\Sstrc)}\Compl\Curlin{\Stuple{f_0,f_1}}$. So
  $\Curlin{f_0}+\Curlin{f_1}=\Ssum_{\Limpl
    XY}\Stuple{\Curlin{f_0},\Curlin{f_1}}
  =\Limplp\Into{\Ssum_Y}\Compl\Curlin{\Stuple{f_0,f_1}}
  =\Curlin{(\Ssum_Y\Compl\Stuple{f_0,f_1})}=\Curlin{(f_0+f_1)}$.
\end{proof}

\begin{theorem}
  If $\cL$ is canonically summable then the axiom \Saxfun{} holds.
\end{theorem}
\begin{proof}
  In this case, we know from Section~\ref{sec:strength-can} that $\Sstrc$ is the double transpose of the following morphism of $\cL$
  \[
    \begin{tikzcd}
      \Limplp{\Into}{\Limplp XY}\ITens X\ITens\Into
      \ar[r,"\Tens\Id\Sym"]
      &\Limplp{\Into}{\Limplp XY}\ITens\Into\ITens X\ar[r,"\Evlin\ITens X"]
      &\Limplp XY\ITens X\ar[r,"\Evlin"]
      & Y
    \end{tikzcd}
  \]
  and therefore is an iso.
\end{proof}

We know that $\Kl\cL$ is a cartesian closed category, with internal
hom-object $(\Simpl XY,\Ev)$ (with $\Simpl XY=\Limplp{\Excl X}Y$ and
$\Ev$ defined using $\Evlin$). Then if $\cL$ is a differential
summable resource category which is closed wrt.~$\ITens$ and satisfies
\Saxfun{}, we have a canonical iso between $\Sdfun{(\Simpl XY)}$ and
$\Simpl X{\Sdfun Y}$ and two morphisms $f_0,f_1\in\Kl\cL(\With ZX,Y)$
are summable (in $\cL$) iff
$\Cur{f_0},\Cur{f_1}\in\Kl\cL(Z,\Simpl XY)$ are summable and then
$\Cur{f_0}+\Cur{f_1}=\Cur{(f_0+f_1)}$.

%%%%%%%%%%%%%%%%%%%%%%%%%%%%%%%%%%%%%%%%%%%%%%%%%%%%%%%%%%%%%%%%%%%%%%%%
% La version de la syntaxe dans le papier soumis à POPL est trop fausse
% pour être maintenue, il faut au moins donner des degrés d'enfouissement
% aux constructions basiques comme \Lproj etc.
%
% Comme je dis bien que c'est une tentative de syntaxe je ne crois pas
% que ce soit un motif pour retirer la soumission. On verra bien ce qu'ils
% diront si jamais le papier n'est pas rejeté d'emblée.
%%%%%%%%%%%%%%%%%%%%%%%%%%%%%%%%%%%%%%%%%%%%%%%%%%%%%%%%%%%%%%%%%%%%%%%%

\section{Sketch of a syntax}\label{sec:syntax}
We outline a tentative syntax corresponding the semantic framework of
this paper and strongly inspired by it. Our choice of notations is
fully coherent with the notations chosen to describe the model,
suggesting a straightforward denotational interpretation. This section
should only be considered as an introduction for another paper which
will introduce a differential version of PCF fully compatible with our
new semantics.

The types are
\[
  A,B,\dots
  \Bnfeq \Tdnat d
  \Bnfor \Timpl AB
\]
and then for any type $A$ we define $\Tdiff A$ as follows:
$\Tdiff{(\Tdnat d)}=\Tdnat{d+1}$ and
$\Tdiff{(\Timpl AB)}=(\Timpl A{\Tdiff B})$.  Terms are
given by
$$
M,N,\dots
\Bnfeq x
\Bnfor \Abst xAM
\Bnfor \App MN
\Bnfor \Ldiff M
\Bnfor \Lprojd idM
\Bnfor \Linjd idM
\Bnfor \Lsumd dM
\Bnfor \Lflipd dM
\Bnfor \Lzero
\Bnfor\Lplus MN
$$
where $i\in\Eset{0,1}$ and $d\in\Nat$. The integer $d$ represents the
depth (in terms of applications of the functor $\Sdfun$) where the
corresponding construct is applied.
\begin{figure}
  \begin{center}
    \begin{prooftree}
      \hypo{}
      \infer1{\Tseq{\Gamma,x:A}xA}
    \end{prooftree}
    \Treesep
    \begin{prooftree}
      \hypo{\Tseq{\Gamma,x:A}{M}{B}}
      \infer1{\Tseq{\Gamma}{\Abst xAM}{\Timpl AB}}
    \end{prooftree}
    \Treesep
    \begin{prooftree}
      \hypo{\Tseq\Gamma M{\Timpl AB}}
      \hypo{\Tseq\Gamma NA}
      \infer2{\Tseq\Gamma{\App MN}B}
    \end{prooftree}
  \end{center}
  \begin{center}
    \begin{prooftree}
      \infer0{\Tseq\Gamma\Lzero A}
    \end{prooftree}
    \Treesep
    \begin{prooftree}
      \hypo{\Tseq\Gamma{M}{\Tdiffm{d+1} A}}
      \hypo{i\in\Eset{0,1}}
      \infer2{\Tseq\Gamma{\Lprojd idM}{\Tdiffm dA}}
    \end{prooftree}
    \Treesep
    \begin{prooftree}
      \hypo{\Tseq\Gamma M{\Tdiffm dA}}
      \hypo{i\in\Eset{0,1}}
      \infer2{\Tseq\Gamma{\Linjd idM}{\Tdiffm{d+1}A}}
    \end{prooftree}
  \end{center}
  \begin{center}
    \begin{prooftree}
      \hypo{\Tseq\Gamma M{\Tdiffm{d+2}A}}
      \infer1{\Tseq\Gamma{\Lsumd dM}{\Tdiffm{d+1}A}}
    \end{prooftree}
    \Treesep
    \begin{prooftree}
      \hypo{\Tseq\Gamma M{\Tdiffm{d+2}A}}
      \infer1{\Tseq\Gamma{\Lflipd dM}{\Tdiffm{d+2}A}}
    \end{prooftree}
    \Treesep
    \begin{prooftree}
      \hypo{\Tseq\Gamma M{\Timpl AB}}
      \infer1{\Tseq\Gamma{\Ldiff M}{\Timpl{\Tdiff A}{\Tdiff B}}}
    \end{prooftree}
  % \Treesep
  % \begin{prooftree}
  %   \hypo{\Wseq\Gamma{M_1}A}
  %   \hypo{\Wseq\Gamma{M_2}A}
  %   \infer2{\Wseq\Gamma{\Lplus{M_1}{M_2}}A}
  % \end{prooftree}
  \end{center}
  \begin{center}
    \begin{prooftree}
      \hypo{\Tseq\Gamma M{\Tdiffm{d+1}A}}
      \infer1{\Tseq\Gamma{\Lplus{\Lprojd0dM}{\Lprojd1dM}}{\Tdiffm dA}}
    \end{prooftree}
    \Treesep
    \begin{prooftree}
      \hypo{\Tseq\Gamma{M_0+M_1}{\Tdiffm{d+1}A}}
      \infer1{\Tseq\Gamma{\Lprojd1d{M_0}+\Lprojd0d{M_1}}{\Tdiffm dA}}
    \end{prooftree}
    \Treesep
    \begin{prooftree}
      \hypo{\Tseq\Gamma MA}
      \hypo{M\Linred M'}
      \infer2{\Tseq\Gamma{M'}A}
    \end{prooftree}
  \end{center}
  \caption{Typing rules}
  \label{fig:typing-rules}
\end{figure}
% We also introduce a weak typing system allowing to prove typing
% judgments $\Wseq\Gamma MA$; this weak system has exactly the same
% rules, plus one rule which is also given in
% Figure~\ref{fig:typing-rules}.
%
Given a variable $x$ and a term $N$, we define a term $\Ldlet xNM$ as
follows.
\begin{align*}
  \Ldlet xNy
  &=
    \begin{cases}
      N & \text{if }y=x\\
      \Linj0 y & \text{otherwise}
    \end{cases}
        &\Ldlet xN{\Abst yBP}&=\Abst yB{\Ldlet xNP}\\
  \Ldlet xN{\Ldiff M}
  &=\Lflipd 0{\Ldiff{\Ldlet xNM}}
        &\Ldlet xN{\App PQ}
               &=\Lsumd 0{\App{\Ldiff{\Ldlet xNP}}{\Ldlet xNQ}}\\
  \Ldlet xN\Lzero
  &=\Lzero
        &\Ldlet xN{\Lplus{M_0}{M_1}}
               &=\Lplus{\Ldlet xN{M_0}}{\Ldlet xN{M_1}}\\
  \Ldlet xN{\Lprojd idM}
  &=\Lprojd i{d+1}{\Ldlet xNM}
        &\Ldlet xN{\Lsumd dM}
               &=\Lsumd{d+1}{\Ldlet xNM}\\
  \Ldlet xN{\Linjd idM}&=\Linjd i{d+1}{\Ldlet xNM}
        &\Ldlet xN{\Lflipd dM}
          &=\Lflipd{d+1}{\Ldlet xNM}\,.
\end{align*}
One checks easily that if $\Tseq{\Gamma,x:A}{M}{B}$ and
$\Tseq\Gamma N{\Tdiff A}$ then $\Tseq\Gamma{\Ldlet xNM}{\Tdiff B}$.
% We
% deal with a few cases. Assume that $M=\Ldiff P$ with
% $\Tseq{\Gamma,x:A}P{\Timpl CE}$ so that
% $B=\Timplp{\Tdiff C}{\Tdiff E}$. By inductive hypothesis we have
% $\Tseq\Gamma{\Ldlet xNP}{\Timpl C{\Tdiff E}}$ and hence
% $\Tseq\Gamma{\Ldiff{\Ldlet xNP}}{\Timpl{\Tdiff C}{\Tdiffm
%     2E}}=\Tdiff{\Timplp{\Tdiff C}{\Tdiff E}}$. Assume that $M=\App PQ$
% with $\Tseq{\Gamma,x:A}P{\Timpl CB}$ and $\Tseq{\Gamma,x:A}QC$. Then
% we have $\Tseq\Gamma{\Ldlet xNP}{\Timpl C{\Tdiff B}}$ and hence
% $\Tseq\Gamma{\Ldiff{\Ldlet xNP}}{\Timpl{\Tdiff C}{\Tdiffm2 B}}$. On
% the other hand we have $\Tseq\Gamma{\Ldlet xNQ}{\Tdiff C}$ and
% therefore
% $\Tseq{\Gamma}{\App{\Ldiff{\Ldlet xNP}}{\Ldlet xNQ}}{\Tdiffm2 B}$ so
% that finally $\Tseq\Gamma{\Ldlet xN{\App PQ}}{\Tdiff B}$ as required.
%
            %             The reduction rules are as follows.

\paragraph{Typing rules.} %
We provide some of the typing rules in
Figure~\ref{fig:typing-rules}. The most important feature of this
typing system is that \emph{it does not contain} the rule
\begin{center}
  \begin{prooftree}
    \hypo{\Tseq\Gamma{M_0}A}
    \hypo{\Tseq\Gamma{M_1}A}
    \infer2{\Tseq\Gamma{M_0+M_1}A}
  \end{prooftree}
\end{center}
typical of the original differential $\lambda$-calculus
of~\cite{EhrhardRegnier02}. So the most tricky rules have to do with
term addition: some such rules are required since sums are allowed in
the syntax, and actually occur during the reduction. We arrived to the
three rules mentioned in this figure, where $\Linred$ is a very simple
rewriting system expressing that sums commute with the linear
constructs of the syntax, for instance %
$\App{\Lplus{M_0}{M_1}}N\Linred\App{M_0}N+\App{M_0}N$.

It is then possible to prove that if %
$\Tseq{\Gamma,x:A}{M}{B}$ and %
$\Tseq\Gamma NA$ then $\Tseq\Gamma{\Subst MNx}B$, and if %
$\Tseq\Gamma P{\Tdiff A}$ then %
$\Tseq\Gamma{\Ldlet xPM}{\Tdiff B}$.

\paragraph{Reduction rules.} %
Our rewriting system contains the rules of the already mentioned
system %
$\Linred$ which expresses that most constructs are linear with respect
to $\Lzero$ and to sums of terms, for instance %
$\Ldiff{(\Lplus{M_0}{M_1})}\Linred\Lplus{\Ldiff{M_0}}{\Ldiff{M_1}}$ %
or $\App\Lzero M\Linred\Lzero$; the only non-linear construct is the
argument side of application. Here are some of the other reduction
rules:
\begin{align*}
  \App{\Abst xAM}{N}
  &\Red\Subst MNx
  &\Ldiffp{\Abst xAM}
  &\Red\Abst y{\Tdiff A}{\Ldlet xyM}\\
  \Lprojd id{\Abst xAM}
  &\Red\Abst xA{\Lprojd idM}
  &\Lprojd id{\App MN}
  &\Red \App{\Lprojd idM}N\\
  \Lprojd id{\Linjd jdM}
    &\Red
      \begin{cases}
        M & \text{if }i=j\\
        \Lzero & \text{otherwise}
      \end{cases}
  &\Lprojd 0d{\Lsumd dM}
  &\Red \Lprojd0d{\Lprojd0 dM}\\
  \Lprojd 1d{\Lsumd dM}
  &\Red \Lprojd 1d{\Lprojd 0dM}+\Lprojd 0d{\Lprojd 1dM}
\end{align*}
Some additional rules are also required, expressing in particular how
constructs applied at different depths commute.

Semantically, the definition of $\Ldlet xNM$ and the reduction rules
are justified by the fact that when $\cL$ is a differential summable
resource SMCC, the category $\Kl\cL$ is cartesian closed and the
functor $\Sdfun$ acts on it as a strong monad; of course the type
$\Tdiff A$ will be interpreted by $\Sdfun X$ where $X$ is the
% TYPO
% interpretation of $X$. The syntactic construct $\Ldiff M$ corresponds
interpretation of $A$. The syntactic construct $\Ldiff M$ corresponds
to the ``internalization'' $\Timplp XY\to\Timplp{\Sdfun X}{\Sdfun Y}$
made possible by the strength of $\Sdfun$ (see
Section~\ref{sec:Kleisli-derivatives}). The reduction rules concerning
$\Lprojs i$ are based on the basic properties of the functor $\Sfun$
and on the definition of the ``multiplication'' $\Sdfmult$ of the
monad $\Sdfun$.

% One could add more reduction rules corresponding to the categorical
% properties of the model, such as for instance
% $\App{\Ldiff M}{\Linj 0N}\Red\Linj 0{\App MN}$. This is probably not
% necessary if we are mainly interested in reducing terms of types $A$
% which are not of shape $\Tdiff B$: the reduction rules involving the
% construction $\Lproj iM$ seem sufficient for extracting the required
% information. Of course this requires a proof.

With these reduction rules, one can prove a form of subject reduction:
if %
$\Tseq\Gamma MA$ and $M\Red M'$ then $\Tseq\Gamma{M'}A$.

\begin{remark}
  The only rule introducing sums of terms is the reduction of
  $\Lprojd1d{\Lsumd dM}$. Since the terms $\Lsumd dM$ are created only
  by the definition of $\Ldlet xN{\App PQ}$ we retrieve the fact that,
  in the differential $\lambda$-calculus, sums are introduced by the
  definition of $\Diffp{\App st}xu$. Therefore the reduction of a term
  which contains no $\Lprojs i^d$'s will lead to as sum-free term. It
  is only when we will want to ``read'' some information about the
  differential content of this term that we will apply to it some
  $\Lprojs i^d$ which will possibly create sums when interacting with
  the $\Lsums$'s contained in the term and typically created by the
  reduction. These $\Lsums$'s are markers of the places where sums
  will be created. But we can try to be clever and create as few sums
  as necessary, whereas the differential $\lambda$-calculus creates
  all possible sums immediately in the course of the reduction. This
  possible parsimony in the creation of sums is very much in the
  spirit of the effectiveness considerations
  of~\cite{BrunelMazzaPagani20,MazzaPagani21}.
\end{remark}

\begin{remark}
  This is only the purely functional core of a differential
  programming language where the ground type $\Tnat$ is
  unspecified. We will extend the language with constants
  $\Num n:\Tnat$ for $n\in\Nat$, and with successor, predecessor, and
  conditional constructs turning it into a type of natural
  numbers. Since these primitives (as well as many others such as
  arbitrary recursive types) are easy to interpret in our coherent
  differential models (such as $\COH$, $\NCOH$ or PCS), they can be
  integrated smoothly in the language as well. Notice to finish that,
  contrarily to what happens in Automatic Differentiation, the
  operation $\Lplus{}{}$ on terms is not related to an operation of
  addition on a ground numerical data type: in AD, one of the the
  ground types is $\Real$ and the $+$ on terms extends the usual
  addition of real numbers pointwise. In AD, the derivatives are
  accordingly defined with respect to this structure of ground types
  whereas in our setting the derivatives are taken with respect to the
  summability structure.
\end{remark}

\subsection{Recursion}
One major feature of the models of the differential $\lambda$-calculus
that we can tackle with the new approach developed in this paper is
that they can have fixpoint operators in $\Kl\cL(\Timpl XX,X)$
implementing general recursion. This is often impossible in an
additive category (typically the categories of topological vector
spaces where the differential $\lambda$-calculus is usually
interpreted): given a closed term $M$ of type $A$, one can define a
term $\Abst xA{(x+M)}:\Timpl AA$ which cannot have a fixpoint in
general if addition is not idempotent.

In contrast consider for instance the category
$\PCOH$~\cite{DanosEhrhard08}. It is a differential canonically summable
resource SMCC where addition is not idempotent and where all least
fixpoint operators are available. And accordingly we can extend our
language with a construct $\Lfix M$ typed by $\Tseq\Gamma{\Lfix M}A$
if $\Tseq\Gamma M{\Timpl AA}$, with the usual reduction rule
$\Lfix M\Red\App M{\Lfix M}$ and so morphisms defined by such
fixpoints can also be differentiated. It turns out that we can easily
extend the definition of $\Ldlet xN{P}$ to the case where
$P=\Lfix M$ with $\Tseq{\Gamma,x:A}{M}{\Timpl BB}$ and
$\Tseq\Gamma N{\Tdiff A}$. The correct definition seems to be
\begin{align*}
  \Ldlet xN{\Lfix M}
  =\Lfix{(\Abst y{\Tdiff B}{\Lsumd 0{\App{\Ldiff{\Ldlet xNM}}{y}}})}\,.
\end{align*}

%%% Local Variables:
%%% mode: latex
%%% TeX-master: "cohdiff"
%%% End:

\section*{Conclusion}
This coherent setting for the formal differentiation of functional
programs should allow to integrate differentiation as an ordinary
construct in any functional programming language, without breaking the
determinism of its evaluation, contrarily to the original differential
$\lambda$-calculus, whose operational meaning was unclear essentially
for its non-determinism. Moreover the differential construct features
commutative monadic structures strongly suggesting to consider it as
an effect. The fact that this differentiation is compatible with
models such as (non uniform) coherence spaces which have nothing to do
with ``analytic'' differentiation suggests that it could also be used
for other operational goals, more internal to the scope of general
purpose functional languages, such as \emph{incremental computing}.

% This work was partly funded by the ANR project \emph{Probabilistic
%   Program Semantics} ANR-19-CE48-0014.

%%% Local Variables:
%%% mode: latex
%%% TeX-master: "cohdiff.tex"
%%% End:

%% Bibliography
\bibliography{newbiblio.bib}

%% Appendix
% \appendix
% \section{Appendix}
% \input{poplapp.tex}

\end{document}